\documentclass[aip,graphicx,linenumbers,preprint]{revtex4-1}

\usepackage{amsmath,amsfonts,amssymb,amsthm}
\usepackage{bm}
\usepackage{hyperref}
\usepackage{enumerate}

\draft 

\begin{document}
\theoremstyle{plain}
\newtheorem{thm}{Theorem}[section]
\newtheorem{prop}[thm]{Proposition}
\newtheorem{cor}[thm]{Corollary}
\theoremstyle{remark}
\newtheorem{rem}[thm]{Remark}
\title{Development of algebraic techniques for the atomic open-shell MBPT(3)}

\author{Rytis Jur\v{s}\.{e}nas}
\email[R. Jursenas: ]{Rytis.Jursenas@tfai.vu.lt}
\author{Gintaras Merkelis}
\email[G. Merkelis: ]{Gintaras.Merkelis@tfai.vu.lt}
\affiliation{Institute of Theoretical Physics and Astronomy of Vilnius University,\\
A. Go\v{s}tauto 12, LT-01108, Vilnius, Lithuania}

\date{\today}

\begin{abstract}
The atomic third-order open-shell many-body perturbation theory is developed.
Special attention is paid to the generation and algebraic analysis of terms
of the wave operator and the effective Hamiltonian as well.
Making use of occupation-number representation and intermediate normalization, the third-order deviations
are worked out by employing a computational software program that embodies the generalized Bloch equation. 
We prove that in the most general
case, the terms of effective interaction operator on the proposed complete model space are generated by 
not more than eight types of the $n$-body ($n\geq2$) parts of the wave operator. To compose the effective Hamiltonian
matrix elements handily, the operators are written in irreducible tensor form. We present the reduction
scheme in a versatile disposition form, thus it is suited for the coupled-cluster approach.
\end{abstract}


\maketitle 

\section{Introduction}
\label{Intro}

From the mathematical point of view, many-body perturbation theory (MBPT) is represented by a number of
recurrence equations which permit to construct a total wave function
of many-body system up to the fixed amendment. The higher order excitations are
involved, the larger number of corrections is found. Due to interactions
of the valence electrons, this number drastically
increases for open-shell atoms. Over the past decades,
it has become generally accepted \cite{Derevianko,Ho,Safronova,Blundell}
to build the wave operator which maps states in selected model
space onto exact states of complex system Hamiltonian, upon the
generalized Bloch equation, introduced by Lindgren et.
al. \cite{Lindgren0,Lindgren,Morrison}. Although, the practical application
of this equation is far from being as mere as the theory is. The situation
exceptionally changed, when symbolic programming tools appropriated
sufficiently high level. Particularly, we exclude a computational
software program \emph{Mathematica} \cite{Mathematica}. There exist  
several symbolic packages, written to evaluate the diagrams or
algebraic expressions of MBPT \cite{Derevianko2,Dzuba,Csepes}. Nevertheless, the
specific features, necessary in computations, are known,
as a rule, to the authors only. Therefore we developed the symbolic package
\emph{NCopertors} \cite{Jursenas}, and the computational part of the present study relies on
this package. The package has been tested by generating the terms
of the first-order wave operator and the second-order effective Hamiltonian.

The other two parts of this paper are devoted to the construction of model space and the reduction of generated terms.
A good survey to various aspects of the problems that merge the selection of model space  can be found
in the reports of Lindgren, Kutzelnigg et.al. \cite{Lindgren2,Kutzelnigg}. On the one hand, we follow the traditional
MBPT \cite{Lindgren} by excluding three types of the one-electron orbitals:
core, virtual (or excited) and valence. On the other hand, we propose an algebraic investment that accounts for
the operation of valence creation and annihilation operators on the model space in a strict manner (Sec. \ref{model}). 
In the result, we  formulate a precise statement which determines the amount of Fock space operators by 
the differing types of one-electron orbitals. That is, it ascertains
the terms of the wave operator producing none zero contributions to the effective
interaction operator. The theorem enables to simplify algebraic procedure, improving further on analysis (especially
reduction) of terms. To apply the suggested formulation of PT, we
display all obtained one-body and two-body terms of the third-order effective
Hamiltonian in irreducible tensor form (Sec. \ref{reduction}). The
reduction is performed by using the angular momentum theory \cite{Jucys,Jucys2}.
The reduction scheme is presented in a versatile disposition form, thus it is suited with 
the coupled-cluster (CC) approach, if some elementary replacements are initiated.
Meanwhile, we show how suggested algebraic approach results to a compactness
and easily accessible form of terms. The method is based
on the construction of effective $n$-particle matrix elements, rather
than the writing of each algebraic expression side by side its diagrammatic
visualization, as, for example, in the classical works of Ho, Lindgren et. al. \cite{Ho,Safronova,Lindgren,Blundell2}
as well as in the recently appeared papers \cite{Porsev,Derevianko3,Derevianko4,Johnson}.

\section{The selection of model space}
\label{model}

\subsection{Hilbert space}\label{D1}

Suppose given a set $X\equiv\{h_{k} \}_{k=1}^{\infty}$ of orthonormal
$N$-electron eigenstates $h_{k}$ of the central-field Hamiltonian
$H_{0}$ with the eigenvalues $E_{0k}$. The {\em configuration state functions} (CSF) $h_{k}$ are characterized
by the sets $\{\Pi_{k}^{X}\Lambda_{k}M_{k}\}$ of numbers: (i) the configuration parity $\Pi_{k}^{X}=\pm 1$;
(ii) the irreducible representations $\Lambda_{k}$ of group $\mathrm{G}$, where $\Lambda_{k}\equiv L_{k}S_{k}$ if
$\mathrm{G}=\mathrm{SO}^{L}(3)\times\mathrm{SU}^{S}(2)$ and $\Lambda_{k}\equiv J_{k}$ if 
$\mathrm{G}=\mathrm{SU}(2)$. The indices 
$M_{k}\equiv M_{L_{k}}M_{S_{k}}$ or $M_{k}\equiv M_{J_{k}}$ enumerate the basis for $\Lambda_{k}$.

By \cite{Lehmann}, we let $\langle,\cdot,\rangle_{\mathcal{H}}:X \times X\longrightarrow \mathbb{R}$ be a real function,
called the {\em scalar product} on $N$-particle Hilbert space $\mathcal{H}$, where

\begin{equation}
\mathcal{H}=\{h_{k}: \langle h_{k}\cdot h_{l}\rangle_{\mathcal{H}}=\delta_{kl}, \forall k,l\in\mathbb{Z}^{+} \}.
\label{eq:N1}
\end{equation}

\noindent{}The infinite Hilbert space $\mathcal{H}$ is assumed to be {\em separable}, thus there exist the linear forms
$\Psi(i)\in\mathcal{H}$, the parameter $\epsilon>0$ and the integer $I_{\epsilon}$ such that
$\Vert \Psi(i)-\sum_{k=1}^{M}c_{k}(i)h_{k}\Vert_{\mathcal{H}}<\epsilon$ for any $M>I_{\epsilon}$ and $c_{k}(i)\in\mathbb{R}$,
$i=1,2,\ldots,D$. If particularly $M=\infty$, then $\Psi(i)$ will denote the eigenstates of the $N$-electron
atomic Hamiltonian $H=H_{0}+V$ with the eigenvalues $E_{i}$, where $V$ represents the perturbation.

Along with the scalar product (or equivalently the norm $\Vert,\Vert=\sqrt{\langle,\cdot,\rangle}$), 
we also propose a unit operator on $\mathcal{H}$.

\begin{prop}\label{propN1}
The form $\widehat{\mathbf{1}}:\mathcal{H}\longrightarrow \mathcal{H}$, expressed by

\begin{equation}
\widehat{\mathbf{1}}=\sum_{k=1}^{\infty}h_{k}h_{k}^{\dagger},
\label{eq:N2}
\end{equation}

\noindent{}is a unit operator on $N$-particle Hilbert space $\mathcal{H}$.
\end{prop}

\begin{proof}
For the vectors $h_{k}$, we immediately gather 
$\widehat{\mathbf{1}}h_{k}=\sum_{l=1}^{\infty}h_{l}\langle h_{l}\cdot h_{k}\rangle_{\mathcal{H}}=h_{k}$. 
In $\mathcal{H}$, we may construct various linear forms
$\Psi\equiv\sum_{k=1}^{M}c_{k}h_{k}\in\mathcal{H}$, $c_{k}\in\mathbb{R}$. Thus for the vectors $\Psi$, we get 
$\widehat{\mathbf{1}}\Psi=\sum_{k=1}^{M}c_{k}\widehat{\mathbf{1}}h_{k}=\Psi$. This proves the proposition.
\end{proof}

\subsection{Orthogonal subspaces}\label{D2}

From the given set $X$ of orthonormal functions $h_{k}$ (Sec. \ref{D1}), 
we form a subset $Y\subset X$ that is formed from the countable functions 
$h_{k}\equiv\phi_{k}$ $\forall k=1,2,\ldots,d<\infty$. We insist the functions $\phi_{k}$ to be identified
by the sets of numbers $\{\Pi^{Y}\Lambda_{k}M_{k} \}$. Here the configuration parity $\Pi^{Y}$ is identical
for all $\phi_{k}$, for all $k=1,2,\ldots,d$. The representations $\Lambda_{k}$ are obtained by reducing the
Kronecker products of irreducible representations $\lambda$ which in turn label the one-electron orbitals.
These orbitals are represented by two types \cite{Lindgren}: core ($\mathrm{c}$) and valence ($\mathrm{v}$). We also insist
the subset $Y$ to be complete by means of the allocation of valence orbitals in all possible ways.

Since $Y=\{\phi_{k} \}_{k=1}^{d}$ is the subset of vectors of $\mathcal{H}$, it is sufficient to introduce
a finite-dimensional subspace $\mathcal{P}$ ($\mathrm{dim}\:\mathcal{P}=d$) of $\mathcal{H}$ by

\begin{equation}
\mathcal{P}=\{\phi_{k}:\langle\phi_{k}\cdot\phi_{l}\rangle_{\mathcal{H}}=\delta_{kl},\forall k,l=1,2,\ldots,d \}.
\label{eq:N3}
\end{equation}

\noindent{}Then the orthogonal complement $\mathcal{P}^{\bot}\equiv\mathcal{Q}=\mathcal{H}\ominus
\mathcal{P}$ of $\mathcal{P}$ is defined by

\begin{equation}
\mathcal{Q}=\{\theta_{k}\equiv h_{d+k}:\langle\theta_{k}\cdot\theta_{l}\rangle_{\mathcal{H}}=
\delta_{kl},\forall k,l\in\mathbb{Z}^{+} \}.
\label{eq:N4}
\end{equation}

\noindent{}This immediately implies that

\begin{equation}
\langle\phi_{k}\cdot\theta_{l}\rangle_{\mathcal{H}}=0\quad \forall k=1,2,\ldots,d\quad \forall 
l\in\mathbb{Z}^{+},
\label{eq:N5}
\end{equation}

\noindent{}thus the orthonormal functions $\theta_{l}$
form the complement $Z=X\backslash Y=
\{\theta_{l} \}_{l=1}^{\infty}$, and they are particularly characterized by the sets of numbers
$\{\Pi_{l}^{Z}\Lambda_{l}M_{l} \}$.
Let us study the properties of $\theta_{l}$ that are predetermined by 
the orthogonality in Eq. (\ref{eq:N5}). First of all, Eq. (\ref{eq:N5}) is true, regardless of whether
$\Pi_{l}^{Z}=\Pi^{Y}$ or not, as $Z\cap Y=\emptyset$ and $h_{k}\in\mathcal{H}$ (see Eq. (\ref{eq:N1})).
Secondly, Eq. (\ref{eq:N5}) does not exclude the functions $\theta_{l}$ that could contain core and/or valence
orbitals of $\phi_{k}$ if $\Pi_{l}^{Z}\neq \Pi^{Y}$. 
By Eq. (\ref{eq:N1}), on the other hand, we are not confined to form the infinite set $X$ of functions
$h_{k}$ in any manner if reserving the formal conditions, presented in Sec. \ref{D1}. Considerably,
we insist the functions $\theta_{l}$ to include the one-electron orbitals that are absent in all $\phi_{k}$, 
irrespectively whether $\Pi_{l}^{Z}=\Pi^{Y}$ or not. These orbitals will be called excited ($\mathrm{e}$) or virtual. 

\subsection{The model functions}\label{D3}
From now on, we assume that the number $D$ of the selected eigenstates $\Psi(i)$ of $H$ equals to 
$\mathrm{dim}\:\mathcal{P}=d$. Then by Prop. \ref{propN1},

\begin{align}
\widehat{\mathbf{1}}\Psi(i)=\Psi(i)=&
\sum_{k=1}^{d}\phi_{k}\langle\phi_{k}\cdot\Psi(i)\rangle_{\mathcal{H}}+
\sum_{l=1}^{\infty}\theta_{l}\langle\theta_{l}\cdot\Psi(i)\rangle_{\mathcal{H}} \nonumber \\
=&\Phi(i)+\widehat{Q}\Psi(i),\quad \Phi(i)=\widehat{P}\Psi(i)=\sum_{k=1}^{d}c_{k}(i)\phi_{k},
\label{eq:N6}
\end{align}

\noindent{}where

\begin{align}
&\widehat{P}=\sum_{k=1}^{d}\phi_{k}\phi_{k}^{\dagger},\quad
\widehat{Q}=\sum_{l=1}^{\infty}\theta_{l}\theta_{l}^{\dagger},\quad 
\widehat{P}+\widehat{Q}=\widehat{\mathbf{1}},\label{eq:N7}\\
&c_{k}(i)=\langle\phi_{k}\cdot\Psi(i)\rangle_{\mathcal{H}}=\langle\phi_{k}\cdot\Phi(i)\rangle_{\mathcal{H}}
\in\mathbb{R}.\label{eq:N8}
\end{align}

\noindent{}The functions $\Phi(i)$ will be called the model functions of $\mathcal{P}$ (see Ref. \cite{Lindgren}).
In Eq. (\ref{eq:N6}), if replacing $\widehat{Q}$ with 
$\widehat{Q}=(\widehat{\Omega}-\widehat{\mathbf{1}})\widehat{P}$, then
we simply get $\Psi(i)=\widehat{\Omega}\Phi(i)$, where $\widehat{\Omega}$ is called the wave operator. The procedure
how to construct the generalized Bloch equation for $\widehat{\Omega}$ is well known and it will not be presented here.

\subsection{The improvement of generation of Hilbert space operators}\label{D4}

In this section, we wish to select the rules that allow to generate the terms of the wave operator as well as the
effective interaction operator efficiently. As it will be demonstrated later, the rules under consideration 
significantly improve the computation of terms of higher-order PT.

The wave operator is represented by the infinite series, where the first term is $\widehat{\mathbf{1}}$. The rest of terms
$\widehat{\Omega}^{(m)}$ ($m\geq 1$) consist of the sum of $n$-body parts $\widehat{\Omega}_{n}^{(m)}$,
where $n=1,2,\ldots$ depends on $m$. We define $\widehat{\Omega}_{n}$
(for any $m$, in general) by

\begin{align}
\widehat{\Omega}_{n}&={\displaystyle \sum_{\alpha_{i}\beta_{j}}\widehat{Q}\widehat{O}_{n}(\alpha\bar{\beta})\widehat{P}
\;\omega_{\alpha_{1}\alpha_{2}\ldots\alpha_{n}\bar{\beta}_{1}\bar{\beta}_{2}\ldots\bar{\beta}_{n}}},\label{eq:2.2} \\
\widehat{O}_{n}(\alpha\bar{\beta})&=a_{\alpha_{1}}a_{\alpha_{2}}\ldots a_{\alpha_{n-1}}a_{\alpha_{n}}a_{\bar{\beta}_{n}}^{
\dagger}a_{\bar{\beta}_{n-1}}^{\dagger}\ldots a_{\bar{\beta}_{2}}^{\dagger}a_{\bar{\beta}_{1}}^{\dagger}.
\label{eq:2.3}
\end{align}

\noindent{}In Eq. (\ref{eq:2.2}), the summation is performed over
all types ($\mathrm{e}$, $\mathrm{v}$, $\mathrm{c}$) of the one-electron
orbitals. Hereafter, we do not write in the sum the concrete values
of orbitals ($\ldots\alpha_{i},\ldots,\bar{\beta}_{j},\ldots$), only
designating their type ($\alpha_{i}=\mathrm{z}_{i}$, $\beta_{j}=\mathrm{z}_{j}$,
$\mathrm{z}=\mathrm{e},\mathrm{v},\mathrm{c}$). The quantities $\omega_{\alpha_{1}\ldots\alpha_{n}\bar{\beta}_{1}
\ldots\bar{\beta}_{n}}$
denote some structure coefficients, composed of the product of one-particle
and/or two-particle matrix elements (energy denominators included). In accordance with Eq. (\ref{eq:N7}),
the Fock space operators $a_{\alpha}$ (creation) and $a_{\bar{\beta}}^{\dagger}$
(annihilation) are assigned as follows

\begin{equation}
\begin{array}{ll}
(\textrm{A})\;\;a_{\mathrm{c}}\widehat{P}=0,             
&\quad (\textrm{C})\;\;a_{\mathrm{v}}\widehat{P}\neq0,\\
(\textrm{B})\;\;a_{\mathrm{\bar{e}}}^{\dagger}\widehat{P}=0,
&\quad (\textrm{D})\;\;a_{\mathrm{\bar{v}}}^{\dagger}\widehat{P}\neq0.
\end{array}
\label{eq:2.4}
\end{equation}

\noindent{}On the one hand, items (A)-(B) correspond to the definitions, presented in Ref. 
\cite[Eq. (13.3), p. 288]{Lindgren}.
On the other hand, items (C)-(D) are strictly determined and they realize the arguments \cite[Sec. 13.1.2, p. 288]{Lindgren}
that it is possible to create as well as annihilate valence electrons in $\mathcal{P}$. It will be demonstrated later that
these items are of special significance.
Moreover, items (C)-(D) point to the definition of the complete model space \cite{Lindgren2} by means of the allocation
of valence states in $\phi_{k}$ in all possible ways (see Sec. \ref{D2}).

The definition of $\widehat{\Omega}_{n}$
conforms to the elucidation, enunciated in accordance with the generalized Bloch equation
\cite[Eq. (11.62), p. 247]{Lindgren}, as $\widehat{\Omega}_{n}$
plays a role of operator, connecting $\mathcal{Q}$ and $\mathcal{P}$
spaces. (In diagrammatic representation $\widehat{\Omega}_{n}$ includes
open diagrams.)
 
Here are the rules which establish
the distributions of $\alpha_{i}$, $\bar{\beta}_{j}$ providing
none zero contributions of $\widehat{O}_{n}(\alpha\bar{\beta})$.

\begin{prop}[\cite{Lindgren}, p. 292]\label{prop1}
The operator $\widehat{P}\widehat{O}_{n}(\alpha\bar{\beta})\widehat{P}\neq0$
$\forall \alpha,\beta=\mathrm{v}$. 
\end{prop}

\begin{prop}\label{prop2}
The operator $\widehat{Q}\widehat{O}_{n}(\alpha\bar{\beta})\widehat{Q}\neq0$
$\forall \alpha,\beta=\mathrm{z}$. 
\end{prop}

\begin{prop}\label{prop3}
The operator $\widehat{Q}\widehat{O}_{n}(\alpha\bar{\beta})\widehat{P}\neq0$
$\forall \alpha=\mathrm{e},\mathrm{v}$, $\forall \beta=\mathrm{v},\mathrm{c}$.
\end{prop}

\noindent{}By Eq. (\ref{eq:2.3}), the self-adjoint operator $\widehat{O}_{n}^{\dagger}(\alpha\bar{\beta})
=\widehat{O}_{n}(\bar{\beta}\alpha)$.
Thus, Proposition \ref{prop3} may be reformulated in a distinct way.

\begin{cor}\label{cor1}
The operator $\widehat{P}\widehat{O}_{n}(\alpha\bar{\beta})\widehat{Q}\neq0$
$\forall \alpha=\mathrm{v},\mathrm{c}$, $\forall \beta=\mathrm{e},\mathrm{v}$. 
\end{cor}

\noindent{}The proof of Proposition \ref{prop1} is obvious, and it directly follows from Eq. (\ref{eq:2.4}). 

\begin{proof}[Proof \textrm{$(${\bf Proposition \ref{prop2}}$)$}]

It suffices to prove that $\widehat{Q}a_{\mathrm{z}}\widehat{Q}\neq0$,
since in this case, the expression
$\widehat{Q}a_{\mathrm{z}}^{\dagger}\widehat{Q}=(\widehat{Q}a_{\mathrm{z}}\widehat{Q})^{\dagger}\neq0$.
By Eq. (\ref{eq:2.4}) and the expression $\widehat{P}+\widehat{Q}=\widehat{\mathbf{1}}$,
we obtain: 
(1) $\widehat{Q}a_{\mathrm{c}}\widehat{Q}=\widehat{Q}a_{\mathrm{c}}$.
But $\widehat{Q}a_{\mathrm{c}}=(a_{\mathrm{c}}^{\dagger}\widehat{Q})^{\dagger}\neq 0$.
(2) The operator $\widehat{Q}a_{\mathrm{e}}\widehat{Q}=a_{\mathrm{e}}\widehat{Q}\neq0$.
(3) $\widehat{Q}a_{\mathrm{v}}\widehat{Q}=[a_{\mathrm{v}},\widehat{P}]\neq 0$, 
where $[\cdot,\cdot]$ denotes a commutator.

\end{proof}

\noindent{}The proof of Proposition \ref{prop3} immediately follows from Propositions \ref{prop1}-\ref{prop2},
and there is no point to present it here. 

Proposition \ref{prop3} agrees with the statement of Lindgren \cite[p. 292]{Lindgren},
if the identity $\alpha=\beta=\mathrm{v}$ is neglected. The last condition,
as a rule, is simply postulated to be false; otherwise, due to zero-valued
energy denominators, the infinite terms $\widehat{\Omega}_{n}$
are observed. However, we will study the operators $\widehat{Q}\widehat{O}_{n}(\mathrm{v\bar{v}})\widehat{P}$
on the Hilbert space $\mathcal{H}$
in a more detail to show in which cases the rejection of the excluded condition $\alpha=\beta=\mathrm{v}$
is true, since it has a direct connection to the properties of the model
functions $\Phi(i)$ (see Eq. (\ref{eq:N6})) which form the subspace $\mathcal{P}$ (see Eq. (\ref{eq:N3})).

Suppose for simplicity that $n=1$. Also, let us mark (see Sec. \ref{D2})

$$
\phi_{k}\equiv h\bigl(\lambda_{k1}^{N_{k1}}\lambda_{k2}^{N_{k2}}\ldots\lambda_{ku}^{N_{ku}}
\Gamma_{k}\Pi^{Y}\Lambda_{k} M_{k} \bigr),\quad k=1,2,\ldots d,
$$

\noindent{}where $\lambda_{ku}^{N_{ku}}$ is the $u$th electron shell in $\phi_{k}$; $\Gamma_{k}$ denotes additional
quantum numbers. Then by Eq. (\ref{eq:N7}), the operator 

\begin{align}
&\widehat{Q}\widehat{O}_{1}\left(\mathrm{v}\bar{\mathrm{v}}\right)\widehat{P}
={\displaystyle\sum_{k=1}^{d}\sum_{l=1}^{\infty}}
\theta_{l}\langle\theta_{l}\cdot\phi_{k}^{\prime}\rangle_{\mathcal{H}}\phi_{k}^{\dagger},\label{eq:2.5}\\
&\phi_{k}^{\prime}\equiv a_{\mathrm{v}_{1}}a_{\mathrm{\bar{v}}_{1}}^{\dagger}\phi_{k}
=(-1)^{\sum_{x=1}^{i-1}N_{kx}+\sum_{y=1}^{j-1}N_{ky}+\delta_{ij}+1}\delta_{\lambda_{\mathrm{\bar{v}}_{1}}\lambda_{ki}}
\nonumber \\
&\times h\bigl(n_{k1}\lambda_{k1}^{N_{k1}}n_{k2}\lambda_{k2}^{N_{k2}}\ldots n_{k\:i-1}\lambda_{k\:i-1}^{N_{k\:i-1}}
n_{ki}\lambda_{ki}^{N_{ki}-1}n_{k\:i+1}\lambda_{k\:i+1}^{N_{k\:i+1}}\ldots \nonumber \\
&\ldots n_{k\:j-1}\lambda_{k\:j-1}^{N_{k\:j-1}}n_{kj}\lambda_{kj}^{N_{kj}+1}n_{k\:j+1}\lambda_{k\:j+1}^{N_{k\:j+1}}\ldots
n_{ku}\lambda_{ku}^{N_{ku}}
\overline{\Pi}_{k}\overline{\Lambda}_{k}\overline{M}_{k}\bigr).
\label{eq:2.5b}
\end{align}

\noindent{}The parity of the $N$-electron CSF $\phi_{k}^{\prime}$ equals to
$\overline{\Pi}_{k}=(-1)^{l_{\mathrm{v}_{1}}+l_{\mathrm{\bar{v}_{1}}}}\Pi^{Y}$.
It is assumed that $a_{\mathrm{\bar{v}_{1}}}^{\dagger}$ and $a_{\mathrm{v}_{1}}$
annihilate and create the $i$th and $j$th valence states of $\phi_{k}$, respectively.
If $\overline{\Pi}_{k}=\Pi^{Y}$,
then, by the definition of $\mathcal{P}$, obtained functions $h(\ldots)$ on the right
hand side of Eq. (\ref{eq:2.5b}) belong to the set $Y$, and thus 

$$
\phi_{k}^{\prime}=(-1)^{\sum_{x=1}^{i-1}N_{kx}+\sum_{y=1}^{j-1}N_{ky}+\delta_{ij}+1}\delta_{\lambda_{\mathrm{\bar{v}}_{1}}\lambda_{ki}}\phi_{k^{\prime}},
\quad k^{\prime}=1,2,\ldots,d.
$$

\noindent{}Thus for $i=j$, we obtain a particular case $k=k^{\prime}$, if the shell $\lambda_{ki}^{N_{ki}}$ in $\phi_{k}$ is
labeled by $\lambda_{ki}=\lambda_{\mathrm{\bar{v}}}$. In general, for some $k$, the functions $\phi_{k^{\prime}}$ are zeroes.
Nevertheless, due to the completeness of the finite set $Y$ (Sec. \ref{D2}), there will always be at least one
function $\phi_{k}$ with $\lambda_{ki}=\lambda_{\mathrm{\bar{v}}}$. But
$\langle\theta_{l}\cdot\phi_{k^{\prime}}\rangle_{\mathcal{H}}=0$,
$\forall k^{\prime}=1,2,\ldots,d$ (see Eq. (\ref{eq:N5})). 
This implies $\widehat{Q}\widehat{O}_{1}\left(\mathrm{v}\bar{\mathrm{v}}\right)\widehat{P}=0$.
The generalization of studied example holds for any $n$.

\begin{cor}\label{cor2}
The operator $\widehat{Q}\widehat{O}_{n}(\mathrm{v\bar{v}})\widehat{P}=0$,
if $\sum_{i=1}^{n}(l_{\mathrm{v}_{i}}+l_{k\mathrm{\bar{v}}_{i}})$ is even
for those $k$ values which determine the electron shell 
$n_{k\mathrm{\bar{v}}_{i}}\lambda_{k\mathrm{\bar{v}}_{i}}^{N_{k\mathrm{\bar{v}}_{i}}}$
in $\phi_{k}\in\mathcal{P}$.
\end{cor}

\noindent{}(Obviously, the condition in Corollary \ref{cor2} holds for 
$\widehat{P}\widehat{O}_{n}(\mathrm{v\bar{v}})\widehat{Q}$
too.) In other words, Corollary \ref{cor2} represents an additional parity
selection rule. The practical treatment of present rule in the study
of triple excitations in CC approach can be found
in Ref.\cite[Sec. III-E, p. 6]{Porsev}.

We can now summarize. The selected $d$-dimensional subspace $\mathcal{P}$ of $N$-electron separable Hilbert space 
$\mathcal{H}$ is assumed to be constructed of the set $Y$ of same parity configuration state functions $\phi_{k}$
by allocating the valence electrons in all possible
ways (complete model space). Additionally, in order to avoid the divergence
of terms of the PT, we select the parity conservation rule in Corollary
\ref{cor2} to be true. The subspace $\mathcal{P}$ will be called the model space.

Finally, let us define the effective interaction operator by \cite[Eq. (15.5), p. 386]{Lindgren}

\begin{equation}
W={\displaystyle \sum_{n}}{\displaystyle \sum_{\xi\leq 4}}W_{n,\xi}={\displaystyle \sum_{n}}{\displaystyle \sum_{\xi\leq 4}}
\{\widehat{P}(V_{1}+V_{2})\widehat{\Omega}_{n}\widehat{P}\}_{\xi},
\label{eq:2.6}
\end{equation}

\noindent{}where $\xi=1,2,3,4$ denotes the number of contractions
between the $i$-body parts ($i=1,2$) of perturbation $V$ and the $n$-body parts of $\widehat{\Omega}$
(for $n+i-\xi\geq0$).
On behalf of the definition of $\mathcal{P}$, the associated operators $W_{n,\xi}$ are generated
by applying the following theorem.

\begin{thm}\label{theor1}
If the $n$-body part of the wave operator $\widehat{\Omega}$ is defined by Eq. (\ref{eq:2.2}), 
then none zero terms of the effective interaction
operator $W$ on the model space $\mathcal{P}$ are generated by maximal
eight types of the Fock space operators $\widehat{O}_{n}(\alpha\bar{\beta})$ for
all $n\geq2$.
\end{thm}

\begin{proof}
We refer to Proposition \ref{prop3} and Corollary \ref{cor2}. None zero terms $W_{n,\xi}$
are obtained, if at least for $\xi=4$, the operators 
$\{\widehat{P}_{s}V_{2}\widehat{\Omega}_{n}\widehat{P}_{s}\}_{4}\neq 0$
are derived. This means $\widehat{\Omega}_{n}$ includes at least
$(n-2)$ creation and $(n-2)$ annihilation operators, designating the
valence orbitals. Possible allocations of $\alpha_{i}$ and $\bar{\beta}_{j}$
(see Eq. (\ref{eq:2.3})) orbitals for all $i,j=1,2,\ldots,n$ are
these:

$\begin{array}{lll}
(1)\;\left\{ \begin{array}{l}
\alpha_{1}=\mathrm{e},\quad\alpha_{i}=\mathrm{v}_{i}\\
\bar{\beta}_{j}=\bar{\mathrm{v}}_{j}\\
i=2,3,\ldots,n\\
j=1,2,\ldots,n\end{array}\right. 
&\quad (2)\;\left\{ \begin{array}{l}
\alpha_{i}=\mathrm{v}_{i}\\
\bar{\beta}_{1}=\bar{\mathrm{c}},\quad\bar{\beta}_{j}=\bar{\mathrm{v}}_{j}\\
i=1,2,\ldots,n\\
j=2,3,\ldots,n\end{array}\right.
&\quad(3)\;\left\{ \begin{array}{l}
\alpha_{1}=\mathrm{e},\quad\alpha_{i}=\mathrm{v}_{i}\\
\bar{\beta}_{1}=\bar{\mathrm{c}},\quad\bar{\beta}_{j}=\bar{\mathrm{v}}_{j}\\
i=2,3,\ldots,n\\
j=2,3,\ldots,n\end{array}\right. \\

(4)\;\left\{ \begin{array}{l}
\alpha_{1}=\mathrm{e},\:\alpha_{2}=\mathrm{e}^{\prime},\:\alpha_{i}=\mathrm{v}_{i}\\
\bar{\beta}_{j}=\bar{\mathrm{v}}_{j}\\
i=3,4,\ldots,n\\
j=1,2,\ldots,n\end{array}\right.
&\quad(5)\;\left\{ \begin{array}{l}
\alpha_{i}=\mathrm{v}_{i}\\
\bar{\beta}_{1}=\bar{\mathrm{c}},\quad\bar{\beta}_{2}=\bar{\mathrm{c}}^{\prime}\\ \bar{\beta}_{j}=\bar{\mathrm{v}}_{j}\\
i=1,2,\ldots,n\\
j=3,4,\ldots,n\end{array}\right.
&\quad (6)\;\left\{ \begin{array}{l}
\alpha_{1}=\mathrm{e},\quad\alpha_{i}=\mathrm{v}_{i}\\
\bar{\beta}_{1}=\bar{\mathrm{c}},\:\bar{\beta}_{2}=\bar{\mathrm{c}}^{\prime},\:\bar{\beta}_{j}=\bar{\mathrm{v}}_{j}\\
i=2,3,\ldots,n\\
j=3,4,\ldots,n\end{array}\right.
\end{array}$

$$
(7)\;\left\{ \begin{array}{l}
\alpha_{1}=\mathrm{e},\quad\alpha_{2}=\mathrm{e}^{\prime}\\ \alpha_{i}=\mathrm{v}_{i}\\
\bar{\beta}_{1}=\bar{\mathrm{c}},\quad\bar{\beta}_{j}=\bar{\mathrm{v}}_{j}\\
i=3,4,\ldots,n\\
j=2,3,\ldots,n\end{array}\right.
\quad (8)\;\left\{ \begin{array}{l}
\alpha_{1}=\mathrm{e},\:\alpha_{2}=\mathrm{e}^{\prime},\,\alpha_{i}=\mathrm{v}_{i}\\
\bar{\beta}_{1}=\bar{\mathrm{c}},\:\bar{\beta}_{2}=\bar{\mathrm{c}}^{\prime},\,\bar{\beta}_{j}=\bar{\mathrm{v}}_{j}\\
i=3,4,\ldots,n\\
j=3,4,\ldots,n\end{array}\right.
$$

\noindent{}Because of the anticommutation properties of creation and annihilation operators,
the similar allocations of none valence orbitals hold for any selected $i$, $j$, not only $i,j=1,2$.
\end{proof}

\noindent{}Theorem \ref{theor1} also fits the effective Hamiltonian, given by the formula 
$H_{eff}^{(m+1)}=\widehat{P}V\widehat{\Omega}^{(m)}\widehat{P}$. In this case, the structure coefficients $\omega$ in Eq.
(\ref{eq:2.2}) are replaced with $\omega^{(m)}$, obtained from $\widehat{\Omega}^{(m)}$. For $m=2$, these coefficients
will be displayed in the next section, where we examine a special case of the application of proposed formalism.

\section{The third-order effective Hamiltonian}
\label{reduction}

In open-shell MBPT, the procedure to determine some fixed number $i=1,2,\ldots,d$ of
energy levels $E_{i}$ of $N$-electron atomic Hamiltonian $H$ is addressed to
the solution of eigenvalue equations $H_{eff}\Phi(i)=E_{i}\Phi(i)$, where the model functions $\Phi(i)$ are determined
in Eq. (\ref{eq:N6}). However, in practical applications, the accuracy
of effective Hamiltonian $H_{eff}$ is finite. In this section, we
consider the third-order contribution to $H_{eff}$. 

The third-order approximation $H_{eff}^{(3)}$ is represented by Eq.
(\ref{eq:2.6}), replacing $\widehat{\Omega}_{n}$ with the second-order contribution
$\widehat{\Omega}^{(2)}=\sum_{n=1}^{4}\widehat{\Omega}_{n}^{(2)}$. Then

\begin{equation}
H_{eff}^{(3)}={\displaystyle \sum_{m=1}^{2}\sum_{n=1}^{4}\sum_{\xi=1}^{\min(2m,2n)}}\widehat{h}_{mn;\xi}^{(3)};
\quad\widehat{h}_{mn;\xi}^{(3)}=\{\widehat{P}V_{m}\widehat{\Omega}_{n}^{(2)}\widehat{P}\}_{\xi}.
\label{eq:3.1}
\end{equation}

\noindent{}The task under consideration is divided into two parts: 
(1) the determination of $\widehat{\Omega}_{n}^{(2)}$ with $n=1,2,3,4$;
(2) the construction of $\widehat{h}_{mn;\xi}^{(3)}$ for $m+n-\xi=1,2$.

\subsection{The determination of terms of the second-order wave operator}
\label{part1}

In the first part, the operators $\widehat{\Omega}_{n}^{(2)}$ are generated
in accordance with Ref. \cite[Eq. (13.30), p. 302]{Lindgren}. The terms
are computed by using the \emph{NCoperators} package which is programmed
upon Propositions \ref{prop1}-\ref{prop3} and Corollaries \ref{cor1}-\ref{cor2}. This part of computation 
is the most time consuming process. The generated terms are arranged by passing to Theorem \ref{theor1}. The coefficients
$\omega_{\alpha\bar{\beta}}^{(2)}$, $\omega_{\alpha\beta\bar{\mu}\bar{\nu}}^{(2)}$,
$\omega_{\alpha\beta\zeta\bar{\mu}\bar{\nu}\bar{\eta}}^{(2)}$ and
$\omega_{\alpha\beta\zeta\rho\bar{\mu}\bar{\nu}\bar{\eta}\bar{\sigma}}^{(2)}$,
located next to the operators $\widehat{O}_{1}$,
$\widehat{O}_{2}$, $\widehat{O}_{3}$ and $\widehat{O}_{4}$ in Eq.
(\ref{eq:2.2}), are treated as the effective one-, two-, three- and four-particle
matrix elements on the basis of the product of accordingly same number one-electron
eigenstates $\varphi(n_{\alpha}\lambda_{\alpha}m_{\alpha})$.
Here $\lambda_{\alpha}$ denotes the irreducible representation of $\mathrm{G}$ (Sec. \ref{D1}).
The basis index acquires the values $m_{\alpha}=-\lambda_{\alpha},-\lambda_{\alpha}+1,\ldots,\lambda_{\alpha}-1,
\lambda_{\alpha}$. Then the Wigner-Eckart theorem is applied to each matrix element, and
the basis indices are excluded to the Clebsch-Gordan coefficients
of $\mathrm{SU(2)}$. Despite of a large number of matrix elements $\omega^{(2)}$, there are only a few of fundamental
constructions to be examined; other elements are derived by varying the given ones. These $\mathrm{SO(3)}$-invariant
constructions are produced in Tabs. \ref{Tab1}-\ref{Tab3}. All computed effective matrix elements $\omega^{(2)}$, 
necessary to form one-body and two-body terms of the effective
Hamiltonian, are presented in an explicit form in Appendix \ref{A}. The analysis of these elements is 
performed in the next part of computation.

\begin{table}
\caption{\label{Tab1}The multipliers for effective one-particle matrix elements of $\widehat{\Omega}^{(2)}$}
\begin{tabular}{ll}
\hline\hline
$(ij\xi)$ & Element: Expression \\
\hline
$(111)$& \raisebox{2.ex}{}$S_{\alpha\bar{\beta}}(\tau_{1}\tau_{2}\tau):$ \\
\raisebox{5.ex}{} & $\begin{array}{l}
(-1)^{\lambda_{\bar{\beta}}-\lambda_{\alpha}-\tau}[\tau_{1},\tau_{2}]^{1/2}\sum_{\mu}\frac{f(\tau_{1}
\lambda_{\alpha}\lambda_{\mu})f(\tau_{2}\lambda_{\mu}\lambda_{\bar{\beta}})}{\varepsilon_{\bar{\beta}}-
\varepsilon_{\mu}}\left\{ \begin{smallmatrix}\tau_{1} & \tau_{2} & \tau\\
\lambda_{\bar{\beta}} & \lambda_{\alpha} & \lambda_{\mu}\end{smallmatrix}\right\} \\
\times\langle\tau_{1}m_{1}\tau_{2}m_{2}\vert\tau m\rangle\end{array}$ \\
$(122)$&  \raisebox{2.5ex}{}$\widetilde{S}_{\alpha\bar{\beta}}(\tau_{1}):$ \\
\raisebox{5.ex}{} & $\begin{array}{l}
2(-1)^{\lambda_{\alpha}+\lambda_{\bar{\beta}}}\sum_{\zeta\bar{\mu}}(-1)^{\lambda_{\zeta}+\lambda_{\bar{\mu}}}
\frac{f(\tau_{1}\lambda_{\zeta}\lambda_{\bar{\mu}})}{\varepsilon_{\zeta\bar{\beta}}-\varepsilon_{\alpha\bar{\mu}}}
\sum_{u}\left[u\right]^{1/2}\left\{ \begin{smallmatrix}\tau_{1} & \lambda_{\zeta} & \lambda_{\bar{\mu}}\\
u & \lambda_{\alpha} & \lambda_{\bar{\beta}}\end{smallmatrix}\right\} \\
\times\widetilde{z}(0\lambda_{\bar{\mu}}\lambda_{\alpha}\lambda_{\bar{\beta}}\lambda_{\zeta}uu)\end{array}$ \\
$(212)$&  \raisebox{2.5ex}{}$\widetilde{S^{\prime}}_{\alpha\bar{\beta}}(\tau_{2}):$ \\
 & \raisebox{3.ex}{}\raisebox{-2.ex}{}$\widehat{R}\bigl(\substack{\bar{\beta}
\zeta\rightarrow\bar{\mu}\\
\alpha\bar{\mu}\rightarrow\zeta}
\bigr)\widetilde{S}_{\alpha\bar{\beta}}(\tau_{2})$ \\
$(223)$&  \raisebox{2.5ex}{}$\widetilde{S}_{\alpha\bar{\beta}}:$ \\
 & \raisebox{5.ex}{}\raisebox{-2.ex}{}$\begin{array}{l}
4\delta_{\lambda_{\alpha}\lambda_{\bar{\beta}}}[\lambda_{\alpha}]^{-1/2}\sum_{u}\sum_{\rho\eta\zeta}
(-1)^{\lambda_{\rho}-\lambda_{\eta}-u}\frac{1}{\varepsilon_{\bar{\beta}\zeta}
-\varepsilon_{\rho\eta}}\\
\times\widetilde{z}(0\lambda_{\alpha}\lambda_{\zeta}\lambda_{\eta}\lambda_{\rho}uu)
\widetilde{z}(0\lambda_{\rho}\lambda_{\eta}\lambda_{\alpha}\lambda_{\zeta}uu)\end{array}$ \\
\hline\hline
\end{tabular} 
\end{table}

\begin{table}
\caption{\label{Tab2}The multipliers for effective two-particle matrix elements of $\widehat{\Omega}^{(2)}$}
\begin{tabular}{ll}
\hline\hline
$(ij\xi)$ & \raisebox{2.ex}{}Element: Expression \\
\hline
$(110)$& $D_{\alpha\beta\bar{\mu}\bar{\nu}}(ud\tau):$ \\
\raisebox{5.5ex}{} & $\begin{array}{l}
\left[\tau_{1},\tau_{2},u,d\right]^{1/2}\frac{f(\tau_{1}\lambda_{\alpha}\lambda_{\bar{\mu}})f(\tau_{2}\lambda_{\beta}
\lambda_{\bar{\nu}})}{\varepsilon_{\bar{\nu}}-
\varepsilon_{\beta}}\langle\tau_{1}m_{1}\tau_{2}m_{2}\vert\tau m\rangle\\
\times\left\{ \begin{smallmatrix}\lambda_{\alpha} & \lambda_{\beta} & u\\
\lambda_{\bar{\mu}} & \lambda_{\bar{\nu}} & d\\
\tau_{1} & \tau_{2} & \tau\end{smallmatrix}\right\} 
\end{array}$ \\
$(121)$& \raisebox{2.ex}{}$D_{\alpha\beta\bar{\mu}\bar{\nu}}(Uu\tau_{1}):$ \\
\raisebox{5.ex}{} & $\begin{array}{l}
2(-1)^{\lambda_{\alpha}-\lambda_{\beta}+\lambda_{\bar{\mu}}+\lambda_{\bar{\nu}}+\tau_{1}}[U]^{1/2}\sum_{\zeta}\frac{f
(\tau_{1}\lambda_{\alpha}\lambda_{\zeta})}{\varepsilon_{\bar{\mu}\bar{\nu}}-\varepsilon_{\beta\zeta}}
z(0\lambda_{\zeta}\lambda_{\beta}\lambda_{\bar{\nu}}\lambda_{\bar{\mu}}uu)\\
\times \left\{ \begin{smallmatrix}\tau_{1} & 
\lambda_{\alpha} & \lambda_{\zeta}\\
\lambda_{\beta} & u & U\end{smallmatrix}\right\} \end{array}$ \\
$(211)$& \raisebox{2.ex}{}$D_{\alpha\beta\bar{\mu}\bar{\nu}}^{\prime}(Uu\tau_{2}):$ \\
 & \raisebox{3.5ex}{}\raisebox{-2.ex}{}$\widehat{R}\bigl(\substack{\bar{\mu}\bar{\nu}
\rightarrow\zeta\\
\beta\zeta\rightarrow\alpha}
\bigr)D_{\alpha\beta\bar{\mu}\bar{\nu}}(Uu\tau_{2})$ \\
$(222)$& \raisebox{2.ex}{}$D_{\alpha\beta\bar{\mu}\bar{\nu}}(uu):$ \\
\raisebox{3.5ex}{} & $
4(-1)^{\lambda_{\bar{\mu}}+\lambda_{\bar{\nu}}+u}[u]^{-1/2}\sum_{\zeta\rho}\frac{z(0\lambda_{\alpha}\lambda_{\beta}\lambda_{\rho}
\lambda_{\zeta}uu)z(0\lambda_{\rho}\lambda_{\zeta}\lambda_{\bar{\nu}}\lambda_{\bar{\mu}}uu)
}{\varepsilon_{\bar{\mu}\bar{\nu}}-\varepsilon_{\zeta\rho}}$ \\
$(222)$& \raisebox{2.ex}{}$\Delta_{\alpha\beta\bar{\mu}\bar{\nu}}(UU):$ \\
& \raisebox{6.ex}{}$\begin{array}{l}
4(-1)^{U-\lambda_{\bar{\mu}}}[U]^{1/2}\sum_{ud}\sum_{\zeta\rho}(-1)^{\lambda_{\zeta}+d}\frac{[u,d]^{1/2}}{\varepsilon_{\bar{\nu}
\zeta}-\varepsilon_{\beta\rho}}\\
\times\left\{ \begin{smallmatrix}\lambda_{\alpha} & \lambda_{\beta} & U\\
\lambda_{\zeta} & d & \lambda_{\bar{\nu}}\\
u & \lambda_{\rho} & \lambda_{\bar{\mu}}\end{smallmatrix}\right\} z(0\lambda_{\alpha}\lambda_{\zeta}\lambda_{\rho}
\lambda_{\bar{\mu}}uu)z(0\lambda_{\rho}\lambda_{\beta}\lambda_{\zeta}\lambda_{\bar{\nu}}dd)\end{array}$ \\
\hline\hline
\end{tabular}
\end{table}

\begin{table}
\caption{\label{Tab3}The multipliers for effective three- and four-particle matrix elements of
$\widehat{\Omega}^{(2)}$}
\begin{tabular}{ll}
\hline\hline
$(ij\xi)$ & Element: Expression\raisebox{2.ex}{} \\
\hline
$(120)$ & $T_{\alpha\beta\zeta\bar{\mu}\bar{\nu}\bar{\eta}}(u\tau_{1}):$ \\
& $\begin{array}{l}
2(-1)^{\lambda_{\bar{\nu}}+\lambda_{\bar{\eta}}+u}\frac{f(\tau_{1}\lambda_{\alpha}\lambda_{\bar{\mu}})}{\varepsilon_{\bar{\nu}
\bar{\eta}}-\varepsilon_{\beta\zeta}}z(0\lambda_{\beta}\lambda_{\zeta}\lambda_{\bar{\eta}}\lambda_{\bar{\nu}}uu)
\end{array}$ \\
$(210)$ & $T_{\alpha\beta\zeta\bar{\mu}\bar{\nu}\bar{\eta}}^{\prime}(u\tau_{2}):$ \\
& \raisebox{3.ex}{}$\widehat{R}\bigl(\substack{\bar{\nu}\bar{\eta}
\rightarrow\bar{\mu}\\
\beta\zeta\rightarrow\alpha}
\bigr)T_{\alpha\beta\zeta\bar{\mu}\bar{\nu}\bar{\eta}}(\Lambda_{1}\tau_{2})$ \\
$(221)$ &  \raisebox{2.5ex}{}$(-1)^{M}T_{\alpha\beta\zeta\bar{\mu}\bar{\nu}\bar{\eta}}(DdU):$ \\
\raisebox{4.5ex}{} &$\begin{array}{l}
4(-1)^{\lambda_{\bar{\eta}}-\lambda_{\bar{\nu}}+\lambda_{\zeta}+U}[D]^{1/2}\sum_{u\rho}\frac{(-1)^{\lambda_{\rho}}}{
\varepsilon_{\bar{\nu}\bar{\eta}}-\varepsilon_{\zeta\rho}}[u]^{1/2}
z(0\lambda_{\alpha}\lambda_{\beta}\lambda_{\rho}\lambda_{\bar{\mu}}uu)\\
\times z(0\lambda_{\rho}\lambda_{\zeta}
\lambda_{\bar{\eta}}\lambda_{\bar{\nu}}dd)\left\{ \begin{smallmatrix}D & d & U\\
\lambda_{\rho} & \lambda_{\beta} & \lambda_{\zeta}\end{smallmatrix}\right\}
 \left\{ \begin{smallmatrix}\lambda_{\alpha} & 
\lambda_{\beta} & u\\
\lambda_{\rho} & \lambda_{\bar{\mu}} & U\end{smallmatrix}\right\}\end{array}$ \\
$(220)$ &  \raisebox{2.ex}{}$Q_{\alpha\beta\zeta\rho\bar{\mu}\bar{\nu}\bar{\eta}\bar{\sigma}}(ud):$ \\
\raisebox{2.ex}{} &$
\frac{4}{\varepsilon_{\bar{\eta}\bar{\sigma}}-\varepsilon_{\zeta\rho}}(-1)^{u+d+\lambda_{\bar{\mu}}+\lambda_{\bar{\nu}}+
\lambda_{\bar{\eta}}+\lambda_{\bar{\sigma}}}
z(0\lambda_{\alpha}\lambda_{\beta}\lambda_{\bar{\nu}}\lambda_{\bar{\mu}}uu)z(0\lambda_{\zeta}\lambda_{\rho}
\lambda_{\bar{\sigma}}\lambda_{\bar{\eta}}dd)$ \\
\hline\hline
\end{tabular}
\end{table}

Let us study the structure
$\sum_{\zeta\bar{\mu}}v_{\zeta\bar{\mu}}\widetilde{v}_{\bar{\mu}
\alpha\zeta\bar{\beta}}/(\varepsilon_{
\zeta\bar{\beta}}-\varepsilon_{\alpha\bar{\mu}})$ as an example.
Here $\varepsilon_{\alpha\beta\ldots\zeta}=\varepsilon_{\alpha}+
\varepsilon_{\beta}+\ldots+\varepsilon_{\zeta}$, where $\varepsilon_{\alpha}$ denotes the one-electron energy.
The one-electron matrix element is given by 
$v_{\zeta\bar{\mu}}=\langle \varphi(n_{\zeta}\lambda_{\zeta}m_{\zeta})\cdot v_{m_{1}}^{\tau_{1}}
\varphi(n_{\bar{\mu}}\lambda_{\bar{\mu}}m_{\bar{\mu}})\rangle_{\mathcal{\mathfrak{H}}}
\equiv\langle n_{\zeta}\lambda_{\zeta}m_{\zeta}\vert v_{m_{1}}^{
\tau_{1}}\vert n_{\bar{\mu}}\lambda_{\bar{\mu}}m_{\bar{\mu}}\rangle$.
The operator $v$ acts on a single-particle Hilbert space $\mathfrak{H}$. Hereafter, the irreducible
representation $\tau_{1}$ labels the one-electron operator $v$ which
befits to a second quantized form $V_{1}$ in the generalized Bloch equation,
written for $\widehat{\Omega}^{(2)}$; the one-electron operator
$v$, located in $\widehat{\Omega}^{(1)}$, will be labeled
by $\tau_{2}$; the operator $v$ in $V_{1}$ of Eq.
(\ref{eq:3.1}) will be labeled by $\tau_{0}$. We assume that 
$v_{m_{1}}^{\tau_{1}\dagger}=(-1)^{\Upsilon(\tau_{1}m_{1})}v_{-m_{1}}^{\tau_{1}}$ and $v_{\zeta\bar{\mu}}=v_{\bar{\mu}\zeta}$. This implies

\begin{align}
\overline{f(\tau_{1}\lambda_{\bar{\mu}}\lambda_{\zeta})}&=\epsilon(\tau_{1}\lambda_{\zeta}\lambda_{\bar{\mu}})f(\tau_{1}
\lambda_{\zeta}\lambda_{\bar{\mu}}),\label{eq:3.2} \\
\epsilon(\tau_{1}\lambda_{\zeta}\lambda_{\bar{\mu}})&=(-1)^{\lambda_{\zeta}-\lambda_{\bar{\mu}}}[\tau_{1}]^{-1}{
\displaystyle \sum_{m_{1}=-\tau_{1}}^{+\tau_{1}}}(-1)^{m_{1}-\Upsilon(\tau_{1}m_{1})},\label{eq:3.3} \\
f(\tau_{1}\lambda_{\zeta}\lambda_{\bar{\mu}})&=-\frac{[\lambda_{\zeta}]^{1/2}}{[\tau_{1}]^{1/2}}[n_{\zeta}\lambda_{
\zeta}||v^{\tau_{1}}||n_{\bar{\mu}}\lambda_{\bar{\mu}}],\label{eq:3.4}
\end{align}

\noindent{}where $[x]\equiv 2x+1$. The phase multiplier $\Upsilon$ is optional. Usually
it is chosen to be equal to \cite{Jucys} $\Upsilon(\tau_{1}m_{1})=\tau_{1}-m_{1}$.
Then $\epsilon(\tau_{1}\lambda_{\zeta}\lambda_{\bar{\mu}})=(-1)^{\lambda_{\zeta}-\lambda_{\bar{\mu}}+\tau_{1}}$.
Particularly, if $v_{q}^{k}=rC_{q}^{k}$ represents the multipole momentum,
then $\overline{f(kl_{\bar{\mu}}l_{\zeta})}=f(kl_{\zeta}l_{\bar{\mu}})$,
since $\lambda_{\zeta,\bar{\mu}}=l_{\zeta,\bar{\mu}}\frac{1}{2}$
and the reduced matrix element $[l_{\zeta}||C^{k}||l_{\bar{\mu}}]\neq0$,
if $l_{\zeta}+l_{\bar{\mu}}+k$ is even. The state $\varphi(n_{\alpha}\lambda_{\alpha}m_{\alpha})$
denotes either $2$-spinor or $4$-spinor, thus Eqs. (\ref{eq:3.2})-(\ref{eq:3.4})
hold for both -- none relativistic and relativistic -- approaches. We
assume that the one-electron Slater integrals are involved in
the definition of $f$. The two-particle matrix element with tilde equals to
$\widetilde{v}_{\bar{\mu}\alpha\zeta\bar{\beta}}=v_{\bar{\mu}\alpha\zeta\bar{\beta}}-v_{\bar{\mu}\alpha\bar{\beta}\zeta}$,
where the element $v_{\bar{\mu}\alpha\zeta\bar{\beta}}=\langle n_{\bar{\mu}}\lambda_{\bar{\mu}}m_{\bar{\mu}}n_{\alpha}\lambda_{\alpha}m_{
\alpha}\vert g_{m}^{\gamma}\vert n_{\zeta}\lambda_{\zeta}m_{\zeta}n_{\bar{\beta}}\lambda_{\bar{\beta}}m_{\bar{\beta}}\rangle$.
In general, the two-particle interaction operator $g_{12}$ acts on
$\mathfrak{H}\times\mathfrak{H}$. However, $g^{\gamma}$ is reduced
and it acts on irreducible tensor space $\mathfrak{H}^{\gamma}$, obtained by reducing \cite[Sec. 2, Eq. (3)]{Jursenas2}
$\mathfrak{H}^{\gamma_{1}}\times\mathfrak{H}^{\gamma_{2}}$. We also account for only scalar
representations $\gamma=0$, and self-adjoint operators $g_{12}^{\dagger}=g_{12}$.
This implies $v_{\bar{\mu}\alpha\zeta\bar{\beta}}=v_{\alpha\bar{\mu}\bar{\beta}\zeta}=v_{\bar{\beta}\zeta\alpha\bar{\mu}}$.
In Ref. \cite{Jursenas2}, it was showed that $v_{\bar{\mu}\alpha\zeta\bar{\beta}}$
may be constructed in two distinct ways: (i) reducing the Kronecker
product $(\lambda_{\bar{\mu}}\times\lambda_{\zeta})\times(\lambda_{\alpha}\times\lambda_{\bar{\beta}})$
($b$-scheme); (ii) reducing the Kronecker product $(\lambda_{\bar{\mu}}\times\lambda_{\alpha})\times(\lambda_{\zeta}\times
\lambda_{\bar{\beta}})$
($z$-scheme). Then the two-particle reduced matrix element is formed
in terms of either $b(0\lambda_{\bar{\mu}}\lambda_{\alpha}\lambda_{\bar{\beta}}\lambda_{\zeta}\gamma_{1}\gamma_{1})$
or $z(0\lambda_{\bar{\mu}}\lambda_{\alpha}\lambda_{\bar{\beta}}\lambda_{\zeta}\Gamma_{1}\Gamma_{1})$
coefficients \cite[Sec. 2, Eqs. (24), (30)]{Jursenas2} for $\gamma=0$.

\begin{table}
\caption{\label{Tab4}The expansion coefficients for one-body terms of the third-order contribution to the
effective Hamiltonian}
\begin{tabular}{ll}
\hline\hline
$(mn\xi)$ & \raisebox{3.ex}{}\raisebox{-1.2ex}{}$\mathfrak{h}_{mn;\xi}^{(3)+}(\Lambda)$ \\
\hline
\raisebox{4.5ex}{} $(111)$ & $\begin{array}{l}
(-1)^{\lambda_{\mathrm{v}}-\lambda_{\mathrm{\bar{v}}}}[\tau_{0}]^{1/2}\sum_{\overline{\Lambda}}
[\overline{\Lambda}]^{1/2}\langle\tau_{0}m_{0}\overline{\Lambda}\:\overline{M}\vert\Lambda M
\rangle\Bigl((-1)^{\Lambda}\sum_{\mathrm{e}}f(\tau_{0}\lambda_{\mathrm{v}}\lambda_{\mathrm{e}})
\Omega_{\mathrm{e\bar{v}}}^{(2)+}(\overline{\Lambda})\left\{ \begin{smallmatrix}\tau_{0} & \overline{\Lambda} & \Lambda\\
\lambda_{\mathrm{\bar{v}}} & \lambda_{\mathrm{v}} & \lambda_{\mathrm{e}}\end{smallmatrix}\right\}\\
-(-1)^{\overline{\Lambda}}\sum_{\mathrm{c}}f(\tau_{0}\lambda_{\mathrm{c}}\lambda_{\mathrm{\bar{v}}})
\Omega_{\mathrm{vc}}^{(2)+}(\overline{\Lambda})\left\{ \begin{smallmatrix}\tau_{0} & \overline{\Lambda} & \Lambda\\
\lambda_{\mathrm{v}} & \lambda_{\mathrm{\bar{v}}} & \lambda_{\mathrm{c}}\end{smallmatrix}\right\} \Bigr)\end{array}$ \\ \\
\raisebox{3.ex}{} $(212)$ & $
2(-1)^{\lambda_{\mathrm{v}}-\lambda_{\mathrm{\bar{v}}}}\sum_{\alpha=\mathrm{v},\mathrm{e}}
\sum_{\mathrm{c}}(-1)^{\lambda_{\alpha^{\prime}}-\lambda_{\mathrm{c}}}
\sum_{u}\widetilde{z}(0\lambda_{\mathrm{c}}\lambda_{\mathrm{v}}\lambda_{\mathrm{\bar{v}}}\lambda_{\alpha^{\prime}}uu)
\Omega_{\alpha^{\prime}\mathrm{c}}^{(2)+}(\Lambda)[u]^{1/2}\left\{ \begin{smallmatrix}\Lambda & \lambda_{\alpha^{\prime}} & 
\lambda_{\mathrm{c}}\\
u & \lambda_{\mathrm{v}} & \lambda_{\mathrm{\bar{v}}}\end{smallmatrix}\right\}$ \\ \\
\raisebox{6.5ex}{} $(122)$ & $\begin{array}{l}
(-1)^{\Lambda}[\tau_{0}]^{1/2}\sum_{\Lambda_{1}\Lambda_{2}\overline{\Lambda}}(-1)^{\overline{\Lambda}}
[\Lambda_{1},\Lambda_{2},\overline{\Lambda}]^{1/2}
\langle\tau_{0}m_{0}\overline{\Lambda}\:\overline{M}\vert\Lambda M\rangle\sum_{\mathrm{c}}
\Bigl(\sum_{\mathrm{v^{\prime}}}(-1)^{\lambda_{\mathrm{c}}-\lambda_{\mathrm{v^{\prime}}}}f(\tau_{0}\lambda_{\mathrm{c}}
\lambda_{\mathrm{v^{\prime}}})\\
\times\widetilde{\Omega}_{\mathrm{v^{\prime}vc\bar{v}}}^{(2)+}(\Lambda_{1}\Lambda_{2}\overline{\Lambda})
\left\{ \begin{smallmatrix}\lambda_{\mathrm{v^{\prime}}} & \lambda_{\mathrm{v}} & \Lambda_{1}\\
\lambda_{\mathrm{c}} & \lambda_{\mathrm{\bar{v}}} & \Lambda_{2}\\
\tau_{0} & \Lambda & \overline{\Lambda}\end{smallmatrix}\right\}
-\sum_{\mathrm{e}}
a(\lambda_{\mathrm{e}}\lambda_{\mathrm{\bar{v}}}\Lambda_{2})
f(\tau_{0}\lambda_{\mathrm{c}}\lambda_{\mathrm{e}})\Omega_{\mathrm{ev\bar{v}c}}^{(2)+}
(\Lambda_{1}\Lambda_{2}\overline{\Lambda})\left\{ \begin{smallmatrix}\lambda_{\mathrm{e}} & 
\lambda_{\mathrm{v}} & \Lambda_{1}\\
\lambda_{\mathrm{c}} & \lambda_{\mathrm{\bar{v}}} & \Lambda_{2}\\
\tau_{0} & \Lambda & \overline{\Lambda}\end{smallmatrix}\right\} \Bigr)\end{array}$ \\ \\
\raisebox{6.5ex}{} $(223)$ & $\begin{array}{l}
2\sum_{\Lambda_{1}\Lambda_{2}}[\Lambda_{1}]^{1/2}\sum_{\mathrm{c}\mathrm{c^{\prime}}}
\Bigl(a(\lambda_{\mathrm{v}}\lambda_{\mathrm{\bar{v}}}\Lambda)
\sum_{\mathrm{v^{\prime}}}\widetilde{z}(0\lambda_{\mathrm{c}}\lambda_{\mathrm{c^{\prime}}}
\lambda_{\mathrm{v^{\prime}}}\lambda_{\mathrm{\bar{v}}}\Lambda_{2}\Lambda_{2})
\widetilde{\Omega}_{\mathrm{vv^{\prime}cc^{\prime}}}^{(2)+}(\Lambda_{1}\Lambda_{2}\Lambda)
\left\{ \begin{smallmatrix}\Lambda_{1} & \Lambda_{2} & \Lambda\\
\lambda_{\mathrm{\bar{v}}} & \lambda_{\mathrm{v}} & \lambda_{\mathrm{v^{\prime}}}\end{smallmatrix}\right\} \\
-a(\Lambda_{1}\Lambda_{2}\Lambda)\sum_{\mathrm{e}}\widetilde{z}(0\lambda_{\mathrm{c}}\lambda_{\mathrm{c^{\prime}}}
\lambda_{\mathrm{\bar{v}}}\lambda_{\mathrm{e}}\Lambda_{2}\Lambda_{2})
\Omega_{\mathrm{evcc^{\prime}}}^{(2)+}(\Lambda_{1}\Lambda_{2}\Lambda)\left\{ \begin{smallmatrix}
\Lambda_{1} & \Lambda_{2} & \Lambda\\
\lambda_{\mathrm{\bar{v}}} & \lambda_{\mathrm{v}} & \lambda_{\mathrm{e}}\end{smallmatrix}\right\} 
\Bigr)+2\sum_{\Lambda_{1}\Lambda_{2}}(-1)^{\Lambda_{1}}
[\Lambda_{2}]^{1/2} \\
\times\sum_{\mathrm{c}}\left\{ \begin{smallmatrix}\Lambda_{1} & \Lambda_{2} & \Lambda\\
\lambda_{\mathrm{\bar{v}}} & \lambda_{\mathrm{v}} & \lambda_{\mathrm{c}}\end{smallmatrix}\right\} 
\Bigl(a(\lambda_{\mathrm{v}}\lambda_{\mathrm{\bar{v}}}\Lambda)\sum_{\mu=\mathrm{v},\mathrm{e}}
(-1)^{\lambda_{\mu^{\prime}}+\lambda_{\mu^{\prime\prime}}}
\widetilde{z}(0\lambda_{\mathrm{v}}\lambda_{\mathrm{c}}\lambda_{\mu^{\prime\prime}}\lambda_{\mu^{\prime}}
\Lambda_{1}\Lambda_{1})\Omega_{\mu^{\prime}\mu^{\prime\prime}\mathrm{c\bar{v}}}^{(2)+}(\Lambda_{1}\Lambda_{2}\Lambda) \\
-a(\Lambda_{1}\Lambda_{2}\Lambda) 
\sum_{\mathrm{e}\mathrm{v^{\prime}}}(-1)^{\lambda_{\mathrm{e}}+\lambda_{\mathrm{v^{\prime}}}}
\widetilde{z}(0\lambda_{\mathrm{c}}\lambda_{\mathrm{v}}\lambda_{\mathrm{v^{\prime}}}\lambda_{\mathrm{e}}\Lambda_{1}\Lambda_{1})
\Omega_{\mathrm{ev^{\prime}\bar{v}c}}^{(2)+}(\Lambda_{1}\Lambda_{2}\Lambda)\Bigr)\end{array}$ \\ \\
\raisebox{6.5ex}{} $(234)$ & $\begin{array}{l}
2\sum_{\mathrm{c}\mathrm{c^{\prime}}}\sum_{\mu=\mathrm{v},\mathrm{e}}
\sum_{\Lambda_{2}}a(\lambda_{\mathrm{c}}\lambda_{\mathrm{c^{\prime}}}\Lambda_{2})
\Bigl(\widetilde{z}(0\lambda_{\mathrm{c}}\lambda_{\mathrm{c^{\prime}}}\lambda_{\mathrm{v^{\prime\prime}}}
\lambda_{\mu^{\prime}}\Lambda_{2}\Lambda_{2})
\Omega_{\mathrm{vv^{\prime\prime}}\mu^{\prime}\mathrm{\bar{v}c^{\prime}c}}^{(2)+}
(\Lambda_{2}\Lambda_{2}\Lambda0)\\
+\sum_{\Lambda_{1}\Lambda_{3}\overline{\Lambda}}
(-1)^{\lambda_{\mathrm{\bar{v}}}+\Lambda_{3}+\overline{M}}
[\Lambda_{1},\Lambda_{3},\overline{\Lambda}]^{1/2}
\langle\Lambda_{3}M_{3}\overline{\Lambda}\:
\overline{M}\vert\Lambda M\rangle\bigl((-1)^{\lambda_{\mathrm{v^{\prime\prime}}}} \\
\times\widetilde{z}(0\lambda_{\mathrm{c}}\lambda_{\mathrm{c^{\prime}}}\lambda_{\mu^{\prime}}
\lambda_{\mathrm{v^{\prime\prime}}}\Lambda_{2}\Lambda_{2})\Omega_{\mathrm{v^{\prime\prime}v}
\mu^{\prime}\mathrm{\bar{v}c^{\prime}c}}^{(2)+}(\Lambda_{1}\Lambda_{2}\Lambda_{3}\overline{\Lambda})
\left\{ \begin{smallmatrix}\Lambda_{1} & \Lambda_{2} & \overline{\Lambda}\\
\lambda_{\mathrm{v^{\prime\prime}}} & \lambda_{\mathrm{v}} & \lambda_{\mu^{\prime}}
\end{smallmatrix}\right\} \left\{ \begin{smallmatrix}\Lambda_{3} & \overline{\Lambda} & \Lambda\\
\lambda_{\mathrm{v}} & \lambda_{\mathrm{\bar{v}}} & \lambda_{\mathrm{v^{\prime\prime}}}
\end{smallmatrix}\right\} +a(\Lambda_{1}\Lambda_{2}\lambda_{\mathrm{v}})\\
\times\widetilde{z}(0\lambda_{\mathrm{c}}\lambda_{\mathrm{c^{\prime}}}\lambda_{\mu^{\prime\prime}}
\lambda_{\mu^{\prime}}\Lambda_{2}\Lambda_{2})\Omega_{\mu^{\prime\prime}
\mu^{\prime}\mathrm{v\bar{v}c^{\prime}c}}^{(2)+}(\Lambda_{1}\Lambda_{2}\Lambda_{3}\overline{\Lambda})
\left\{ \begin{smallmatrix}\Lambda_{1} & \Lambda_{2} & \Lambda\\
\lambda_{\mu^{\prime\prime}} & \lambda_{\mathrm{v}} & \lambda_{\mu^{\prime}}
\end{smallmatrix}\right\} \left\{ \begin{smallmatrix}\Lambda_{3} & \overline{\Lambda} & \Lambda\\
\lambda_{\mathrm{v}} & \lambda_{\mathrm{\bar{v}}} & \lambda_{\mu^{\prime\prime}}
\end{smallmatrix}\right\} \bigr)\Bigr)\end{array}$ \\
\hline\hline
\end{tabular}
\end{table}

Usually the authors (see above cited works) better prefer $b$-scheme,
as it is more convenient to produce the algebraic expression, written
for a particular Goldstone diagram. Then a two-particle matrix element
is denoted $v_{\bar{\mu}\alpha\zeta\bar{\beta}}\equiv g_{\bar{\mu}\alpha\zeta\bar{\beta}}$
and $\widetilde{v}_{\bar{\mu}\alpha\zeta\bar{\beta}}\equiv\widetilde{g}_{\bar{\mu}\alpha\zeta\bar{\beta}}$;
the correspondent reduced matrix elements are given by $X_{\bar{\mu}\alpha\zeta\bar{\beta}}$
and $Z_{\bar{\mu}\alpha\zeta\bar{\beta}}$ (see, for example, Ref. \cite[Appendix, Eqs. (A1), (A5)]{Ho}).
Contrarily, in this paper we give priority exceptionally to the algebraic
analysis of terms, and, along with Theorem \ref{theor1}, $z$-scheme is more preferable for the arrangement
of their irreducible tensor form. Then $v_{\bar{\mu}\alpha\zeta\bar{\beta}}$
($\widetilde{v}_{\bar{\mu}\alpha\zeta\bar{\beta}}$) will be expressed
by $z(0\lambda_{\bar{\mu}}\lambda_{\alpha}\lambda_{\bar{\beta}}\lambda_{\zeta}\Gamma_{1}\Gamma_{1})$
($\widetilde{z}(0\lambda_{\bar{\mu}}\lambda_{\alpha}\lambda_{\bar{\beta}}\lambda_{\zeta}\Gamma_{1}\Gamma_{1})$),
where 

\begin{equation}
\widetilde{z}(0\lambda_{\bar{\mu}}\lambda_{\alpha}\lambda_{\bar{\beta}}\lambda_{\zeta}\Gamma_{1}\Gamma_{1})
=z(0\lambda_{\bar{\mu}}\lambda_{\alpha}\lambda_{\bar{\beta}}\lambda_{\zeta}\Gamma_{1}\Gamma_{1})
-a(\lambda_{\zeta}\lambda_{\bar{\beta}}\Gamma_{1})z(0\lambda_{\bar{\mu}}\lambda_{\alpha}\lambda_{\zeta}
\lambda_{\bar{\beta}}\Gamma_{1}\Gamma_{1}),\label{eq:3.5}
\end{equation}

\noindent{}where $a(\lambda_{\zeta}\lambda_{\bar{\beta}}\Gamma_{1})=(-1)^{\lambda_{\zeta}+\lambda_{\bar{\beta}}+\Gamma_{1}}$.
The studied effective one-particle matrix element is represented by

\begin{equation}
\sum_{\zeta\bar{\mu}}\frac{v_{\zeta\bar{\mu}}\widetilde{v}_{\bar{\mu}\alpha\zeta\bar{\beta}}}{\varepsilon_{\zeta\bar{\beta}}
-\varepsilon_{\alpha\bar{\mu}}}=(-1)^{\lambda_{\bar{\beta}}+m_{\bar{\beta}}}\widetilde{S}_{\alpha\bar{\beta}}(\tau_{1})
\langle\lambda_{\alpha}m_{\alpha}\lambda_{\bar{\beta}}-m_{\bar{\beta}}\vert\tau_{1}m_{1}\rangle,\label{eq:3.6}\end{equation}

\begin{table}
\caption{\label{Tab5}The expansion coefficients for two-body terms of the third-order contribution to the 
effective Hamiltonian}
\begin{tabular}{ll}
\hline\hline
$(mn\xi)$ & \raisebox{3.ex}{}\raisebox{-1.2ex}{}$\mathfrak{h}_{mn;\xi}^{(3)+}(\Lambda_{1}\Lambda_{2}\Lambda)$ \\
\hline
$(121)$ & $\begin{array}{l}
-[\tau_{0}]^{1/2}\sum_{\overline{\Lambda}_{1}\overline{\Lambda}}[\overline{\Lambda}_{1},\overline{\Lambda}]^{1/2}
\Bigl((-1)^{\overline{\Lambda}_{1}}a(\lambda_{\mathrm{v}}\lambda_{\mathrm{v^{\prime}}}\tau_{0})
a(\Lambda_{1}\Lambda_{2}\Lambda)[\Lambda_{1}]^{1/2}\langle\tau_{0}m_{0}\overline{\Lambda}\:\overline{M}\vert\Lambda M
\rangle\left\{ \begin{smallmatrix}\tau_{0} & \Lambda_{1} & \overline{\Lambda}_{1}\\
\Lambda_{2} & \overline{\Lambda} & \Lambda\end{smallmatrix}\right\} \\
\times\sum_{\mathrm{e}}
f(\tau_{0}\lambda_{\mathrm{v}}\lambda_{\mathrm{e}})\Omega_{\mathrm{ev^{\prime}\bar{v}\bar{v}^{\prime}}}^{(2)+}
(\overline{\Lambda}_{1}\Lambda_{2}\overline{\Lambda})
\left\{ \begin{smallmatrix}\tau_{0} & \lambda_{\mathrm{v}} & \lambda_{\mathrm{e}}\\
\lambda_{\mathrm{v^{\prime}}} & \overline{\Lambda}_{1} & \Lambda_{1}\end{smallmatrix}\right\} 
+(-1)^{m_{0}}a(\Lambda_{1}\Lambda_{2}\overline{\Lambda})[\Lambda_{2}]^{1/2}\langle\tau_{0}-m_{0}\overline{\Lambda}\:
\overline{M}\vert\Lambda M\rangle\\
\times\left\{ \begin{smallmatrix}\tau_{0} & \Lambda_{2} & \overline{\Lambda}_{1}\\
\Lambda_{1} & \overline{\Lambda} & \Lambda\end{smallmatrix}\right\} \sum_{\mathrm{c}}
f(\tau_{0}\lambda_{\mathrm{c}}\lambda_{\mathrm{\bar{v}^{\prime}}})\Omega_{\mathrm{vv^{\prime}c\bar{v}}}^{(2)+}
(\Lambda_{1}\overline{\Lambda}_{1}\overline{\Lambda})
\left\{ \begin{smallmatrix}\tau_{0} & \lambda_{\mathrm{\bar{v}^{\prime}}} & \lambda_{\mathrm{c}}\\
\lambda_{\mathrm{\bar{v}}} & \overline{\Lambda}_{1} & \Lambda_{2}\end{smallmatrix}\right\} \Bigr)\end{array}$ \\ \\
\raisebox{5.ex}{}$(211)$ & $\begin{array}{l}
(-1)^{\lambda_{\mathrm{\bar{v}^{\prime}}}+\Lambda}\Bigl([\Lambda_{2}]^{1/2}
\sum_{\mathrm{e}}(-1)^{\lambda_{\mathrm{e}}}\widetilde{z}(0\lambda_{\mathrm{v}}\lambda_{\mathrm{v^{\prime}}}
\lambda_{\mathrm{e}}\lambda_{\mathrm{\bar{v}}}\Lambda_{1}\Lambda_{1})\Omega_{\mathrm{e\bar{v}^{\prime}}}^{(2)+}
(\Lambda)\left\{ \begin{smallmatrix}\Lambda_{1} & \Lambda_{2} & \Lambda\\
\lambda_{\mathrm{\bar{v}^{\prime}}} & \lambda_{\mathrm{e}} & \lambda_{\mathrm{\bar{v}}}\end{smallmatrix}\right\} 
-(-1)^{\lambda_{\mathrm{\bar{v}}}+\Lambda_{1}}[\Lambda_{1}]^{1/2}\\
\times\sum_{\mathrm{c}}\widetilde{z}(0\lambda_{\mathrm{c}}\lambda_{\mathrm{v}}\lambda_{\mathrm{\bar{v}^{\prime}}}
\lambda_{\mathrm{\bar{v}}}\Lambda_{2}\Lambda_{2})
\Omega_{\mathrm{v^{\prime}c}}^{(2)+}(\Lambda)
\left\{ \begin{smallmatrix}\Lambda_{1} & \Lambda_{2} & \Lambda\\
\lambda_{\mathrm{c}} & \lambda_{\mathrm{v^{\prime}}} & \lambda_{\mathrm{v}}\end{smallmatrix}\right\} 
\Bigr)\end{array}$ \\ \\
\raisebox{7.ex}{}$(222)$ & $\begin{array}{l}
a(\lambda_{\mathrm{\bar{v}}}\lambda_{\mathrm{\bar{v}^{\prime}}}\Lambda_{2})[\Lambda_{2}]^{-1/2}
\sum_{\mathrm{cc^{\prime}}}\widetilde{z}(0\lambda_{\mathrm{c}}\lambda_{\mathrm{c^{\prime}}}
\lambda_{\mathrm{\bar{v}^{\prime}}}\lambda_{\mathrm{\bar{v}}}\Lambda_{2}\Lambda_{2})
\Omega_{\mathrm{vv^{\prime}cc^{\prime}}}^{(2)+}(\Lambda_{1}\Lambda_{2}\Lambda)
+[\Lambda_{1}]^{-1/2}\\
\times\sum_{\mathrm{e}}\sum_{\mu=\mathrm{v},\mathrm{e}}\widetilde{z}(0\lambda_{\mathrm{v}}
\lambda_{\mathrm{v^{\prime}}}\lambda_{\mu^{\prime\prime}}\lambda_{\mathrm{e}}\Lambda_{1}\Lambda_{1})
\Omega_{\mathrm{e}\mu^{\prime\prime}\mathrm{\bar{v}\bar{v}^{\prime}}}^{(2)+}
(\Lambda_{1}\Lambda_{2}\Lambda)+2(-1)^{\Lambda_{1}+\Lambda_{2}}[\Lambda_{1},\Lambda_{2}]^{1/2}
\sum_{\overline{\Lambda}_{1}\overline{\Lambda}_{2}u}[\overline{\Lambda}_{1}]^{1/2} \\
\times[\overline{\Lambda}_{2},u]^{1/2}
\sum_{\mathrm{c}}\Bigl((-1)^{\overline{\Lambda}_{1}+\overline{\Lambda}_{2}}\sum_{\mathrm{e}}
\widetilde{z}(0\lambda_{\mathrm{v}}\lambda_{\mathrm{c}}\lambda_{\mathrm{e}}\lambda_{\mathrm{\bar{v}^{\prime}}}uu)
\Omega_{\mathrm{ev^{\prime}\bar{v}c}}^{(2)+}(\overline{\Lambda}_{1}\overline{\Lambda}_{2}\Lambda)
\left\{ \begin{smallmatrix}\lambda_{\mathrm{\bar{v}^{\prime}}} &  & \lambda_{\mathrm{e}} &  & 
\overline{\Lambda}_{1} &  & \overline{\Lambda}_{2}\\
\  & u &  & \lambda_{\mathrm{v^{\prime}}} &  & \Lambda &  & \lambda_{\mathrm{\bar{v}}}\\
\lambda_{\mathrm{c}} &  & \lambda_{\mathrm{v}} &  & \Lambda_{1} &  & \Lambda_{2}\end{smallmatrix}\right\} \\
+(-1)^{\lambda_{\mathrm{v^{\prime}}}+\lambda_{\mathrm{\bar{v}}}}
\sum_{\mathrm{v^{\prime\prime}}}
(-1)^{\lambda_{\mathrm{c}}+\lambda_{\mathrm{v^{\prime\prime}}}}\widetilde{z}(0\lambda_{\mathrm{v}}
\lambda_{\mathrm{c}}\lambda_{\mathrm{v^{\prime\prime}}}\lambda_{\mathrm{\bar{v}^{\prime}}}uu)
\widetilde{\Omega}_{\mathrm{v^{\prime}v^{\prime\prime}c\bar{v}}}^{(2)+}(\overline{\Lambda}_{1}
\overline{\Lambda}_{2}\Lambda)
\left\{ \begin{smallmatrix}\lambda_{\mathrm{\bar{v}^{\prime}}} &  & 
\lambda_{\mathrm{v^{\prime\prime}}} &  & \overline{\Lambda}_{1} &  & \overline{\Lambda}_{2}\\
\  & u &  & \lambda_{\mathrm{v^{\prime}}} &  & \Lambda &  & \lambda_{\mathrm{\bar{v}}}\\
\lambda_{\mathrm{c}} &  & \lambda_{\mathrm{v}} &  & \Lambda_{1} &  & \Lambda_{2}\end{smallmatrix}\right\} 
\Bigr)\end{array}$ \\ \\
\raisebox{8.ex}{}$(132)$ & $\begin{array}{l}
(-1)^{\tau_{0}+M}[\tau_{0},\Lambda_{2},\Lambda]^{1/2}\sum_{\mathrm{c}}
\sum_{\mu=\mathrm{v},\mathrm{e}}\:\:\sum_{\overline{\Lambda}_{2}\Lambda_{3}\overline{\Lambda}\vartheta}
(-1)^{\overline{\Lambda}}[\overline{\Lambda}_{2},\Lambda_{3},\vartheta]^{1/2}\langle\Lambda_{3}M_{3}
\vartheta\varrho\vert\overline{\Lambda}\:\overline{M}\rangle \\
\times\langle\tau_{0}m_{0}\Lambda-M
\vert\vartheta\varrho\rangle
\Bigl((-1)^{\lambda_{\mathrm{c}}-\lambda_{\mu^{\prime\prime}}+\Lambda}
[\Lambda_{1}]^{1/2}\sum_{\overline{\Lambda}_{1}}[\overline{\Lambda}_{1}]^{1/2}
f(\tau_{0}\lambda_{\mathrm{c}}\lambda_{\mu^{\prime\prime}}) \\
\times\left\{ \begin{smallmatrix}
\lambda_{\mathrm{\bar{v}}} &  & \overline{\Lambda}_{2}\:\:\Lambda_{3} &  & 
\overline{\Lambda}\:\:\lambda_{\mathrm{v}} &  & \overline{\Lambda}_{1}\\
\  & \lambda_{\mathrm{\bar{v}^{\prime}}} &  & \vartheta &  & \lambda_{\mathrm{v^{\prime}}}\\
\Lambda_{2} &  & \lambda_{\mathrm{c}}\:\:\Lambda &  & \tau_{0}\:\:\Lambda_{1} &  & 
\lambda_{\mu^{\prime\prime}}\end{smallmatrix}\right\}
\bigl(\Omega_{\mathrm{vv^{\prime}}
\mu^{\prime\prime}\mathrm{\bar{v}\bar{v}^{\prime}c}}^{(2)+}(\overline{\Lambda}_{1}
\overline{\Lambda}_{2}\Lambda_{3}\overline{\Lambda})
-\delta_{\mu\mathrm{v}}
a(\lambda_{\mathrm{v^{\prime}}}\lambda_{\mathrm{v^{\prime\prime}}}\Lambda_{1}) \\
\times\Omega_{\mathrm{vv^{\prime\prime}v^{\prime}\bar{v}\bar{v}^{\prime}c}}^{(2)+}
(\overline{\Lambda}_{1}\overline{\Lambda}_{2}\Lambda_{3}\overline{\Lambda})\bigr)
+(-1)^{\lambda_{\mathrm{\bar{v}^{\prime}}}+\lambda_{\mathrm{v^{\prime\prime}}}
+\Lambda_{2}+\Lambda_{3}+\vartheta}
\delta_{\mu\mathrm{v}}f(\tau_{0}\lambda_{\mathrm{c}}
\lambda_{\mathrm{v^{\prime\prime}}})
\Omega_{\mathrm{v^{\prime\prime}vv^{\prime}
\bar{v}\bar{v}^{\prime}c}}^{(2)+}(\Lambda_{1}\overline{\Lambda}_{2}\Lambda_{3}
\overline{\Lambda}) \\
\times\left\{ \begin{smallmatrix}\lambda_{\mathrm{c}} &  & \overline{\Lambda}_{2} &  & 
\overline{\Lambda} &  & \Lambda_{3}\\
\  & \lambda_{\mathrm{\bar{v}^{\prime}}} &  & \Lambda_{1} &  & \vartheta &  & \lambda_{\mathrm{v^{\prime\prime}}}\\
\lambda_{\mathrm{\bar{v}}} &  & \Lambda_{2} &  & \Lambda &  & \tau_{0}\end{smallmatrix}\right\} \Bigr)\end{array}$ \\
\hline\hline
\end{tabular}
\end{table}

\noindent{}where $\widetilde{S}_{\alpha\bar{\beta}}(\tau_{1})$ plays
a role of the effective one-particle reduced matrix element. The quantity
$\langle\lambda_{\alpha}m_{\alpha}\lambda_{\bar{\beta}}-m_{\bar{\beta}}\vert\tau_{1}m_{1}\rangle$
denotes the Clebsch-Gordan coefficient of $\mathrm{SU(2)}$. Moreover,
the form of Eq. (\ref{eq:3.6}) directly indicates that the 
one-particle operator, given by the product of $\widehat{O}_{1}=a_{\alpha}a_{\bar{\beta}}^{\dagger}$ and Eq. (\ref{eq:3.6}),
is simply equal to $W_{m_{1}}^{\tau_{1}}(\lambda_{\alpha}\widetilde{\lambda}_{\bar{\beta}})
\widetilde{S}_{\alpha\bar{\beta}}(\tau_{1})$, where the irreducible tensor operator  
$W^{\tau_{1}}(\lambda_{\alpha}\widetilde{\lambda}_{\bar{\beta}})
=[a^{\lambda_{\alpha}}\times\tilde{a}^{\lambda_{\bar{\beta}}}]^{\tau_{1}}$ is obtained by reducing $\widehat{O}_{1}$. 
The representation $\widetilde{\lambda}_{\bar{\beta}}$
designates the transposed annihilation operator $\tilde{a}_{m_{\bar{\beta}}}^{\lambda_{\bar{\beta}}}=(-1)^{\lambda_{\bar{\beta}}
-m_{\bar{\beta}}}a_{-m_{\bar{\beta}}}^{\lambda_{\bar{\beta}}\dagger}$.
The matrix elements of $W^{\tau_{1}}$ can be found in Ref. \cite{Rudzikas,Rudzikas2}.
Here and elsewhere, it is considered if necessary that the summation is fulfilled over all 
given one-electron orbitals of marked type. However, only the sum running over the repetitive orbitals 
($\zeta$, $\bar{\mu}$ in this case) will be written.

In Tabs. \ref{Tab1}-\ref{Tab3}, the index $i=1,2$ labels $V_{i}$,
while $j=1,2$ labels $\widehat{\Omega}_{j}^{(1)}$. The
operator $\widehat{R}$ replaces orbitals in denominators. For example,
the expression $\widehat{R}\bigl(\substack{\bar{\nu}\bar{\eta}\rightarrow\bar{\mu}\\
\beta\zeta\rightarrow\alpha}
\bigr)(\varepsilon_{\bar{\nu}\bar{\eta}}-\varepsilon_{\beta\zeta})^{-1}$ reads $(\varepsilon_{\bar{\mu}}-\varepsilon_{\alpha})^{-1}$. 
The quantities $\left\{ \begin{smallmatrix}j_{1} & j_{2} & j_{3}\\
l_{1} & l_{2} & l_{3}\end{smallmatrix}\right\} $ and $\left\{ \begin{smallmatrix}j_{1} & j_{2} & j_{3}\\
l_{1} & l_{2} & l_{3}\\
k_{1} & k_{2} & k_{3}\end{smallmatrix}\right\} $ denote $6j$- and $9j$-symbols. The elements
which are found by making the three-pair contractions
between $V_{2}$ and $\widehat{\Omega}_{2}^{(1)}$ vanish, if representations
$\lambda_{\alpha}\neq\lambda_{\bar{\beta}}$ in $W^{0}(\lambda_{\alpha}\widetilde{\lambda}_{\bar{\beta}})$. 
In this case, the orbitals $\zeta$, $\rho$, $\eta$ are identical for all
$\alpha$, $\bar{\beta}$: $\zeta=\mathrm{c}$, $\rho=\mathrm{e}$, $\eta=\mathrm{v}$.
Also, we mark off 
$K_{\alpha\bar{\beta}}\equiv\widetilde{S}_{\alpha\bar{\beta}},
\widetilde{S^{\prime}}_{\alpha\bar{\beta}}$
by the summation parameter $\mu\equiv n_{\mu}\lambda_{\mu}$: (i)
if $\mu=\mathrm{v}$, then $K_{\alpha\bar{\beta}}=\dot{K}_{\alpha\bar{\beta}}$;
(ii) if $\mu=\mathrm{e}$, then $K_{\alpha\bar{\beta}}=\ddot{K}_{\alpha\bar{\beta}}$;
(iii) if $\mu=\mathrm{c}$, then we simply write $K_{\alpha\bar{\beta}}$. The tildes designate that
the direct and exchanged parts of a two-particle matrix element are involved. If given $D_{\alpha\beta\bar{\mu}\bar{\nu}}(
\widetilde{U}u\tau_{1})$, then $\widetilde{U}$ marks $\widetilde{v}$, represented by $V_{2}$. Additionally, if given
$D_{\alpha\beta\bar{\mu}\bar{\nu}}(U\widetilde{u}\tau_{1})$, then $\widetilde{u}$ marks $\widetilde{v}$, fitted to 
$\widetilde{\Omega}_{2}^{(1)}$. For both $\widetilde{U}$ and $\widetilde{u}$, we write 
$\widetilde{D}_{\alpha\beta\bar{\mu}\bar{\nu}}(Uu\tau_{1})$.
A similar argument holds for the rest of elements in Tabs. \ref{Tab2}-\ref{Tab3}.
These elements are also separated by the summation parameters. If
$\xi=1$, then the notations are similar to $K_{\alpha\bar{\beta}}$ case.
If $\xi=2$ (see Tab. \ref{Tab2}), then for
$D_{\alpha\beta\bar{\mu}\bar{\nu}}(uu)$, we write: (i) if $\zeta,\rho=\mathrm{c}$,
then $D\equiv D$; (ii) if $\zeta,\rho=\mathrm{v}$, then $D=\dot{D}$;
(iii) if $\zeta,\rho=\mathrm{e}$, then $D=\ddot{D}$; (iv) if $\zeta=\mathrm{e}$
and $\rho=\mathrm{v}$, then $D=\dddot{D}$. For $\zeta=\mathrm{v}$
and $\rho=\mathrm{e}$, the similar triple doted $\dddot{D}$ is considered.
This is due to the symmetry properties of $v_{\alpha\beta\bar{\mu}\bar{\nu}}$.
Finally, for $\Delta_{\alpha\beta\bar{\mu}\bar{\nu}}(UU)$, we write:
(i) if $\rho=\mathrm{v}$, then $\Delta=\dot{\Delta}$; (ii) if $\rho=\mathrm{e}$,
then $\Delta=\ddot{\Delta}$. In these cases, $\zeta=\mathrm{c}$.

\subsection{The determination of terms of the third-order effective Hamiltonian}
\label{part2}

By Proposition \ref{prop1}, it follows that this part of computation requires significantly less time than the first one. 
Besides, the none zero terms of $\widehat{h}_{mn;\xi}^{(3)}$ are derived in accordance with Theorem \ref{theor1} which allows
to reject a large amount of $\widehat{\Omega}^{(2)}$ terms, attaching the zero-valued contributions. The operators 
$\widehat{h}_{mn;\xi}^{(3)}$ are considered by the formulas

\begin{equation}
\widehat{h}_{mn;\xi}^{(3)}={\displaystyle \sum_{\Lambda M}}W_{M}^{\Lambda}(\lambda_{\mathrm{v}}
\widetilde{\lambda}_{\mathrm{\bar{v}}})\mathfrak{h}_{mn;\xi}^{(3)}(\Lambda),
\label{eq:3.8}
\end{equation}

\noindent{}-- for $m+n-\xi=1$, and 

\begin{equation}
\widehat{h}_{mn;\xi}^{(3)}=-{\displaystyle \sum_{\substack{\Lambda_{1}\Lambda_{2}\\\Lambda M}}}
[W^{\Lambda_{1}}(\lambda_{\mathrm{v}}\lambda_{\mathrm{v^{\prime}}})
\times W^{\Lambda_{2}}(\widetilde{\lambda}_{\mathrm{\bar{v}}}
\widetilde{\lambda}_{\mathrm{\bar{v}^{\prime}}})]_{M}^{\Lambda}\mathfrak{h}_{mn;\xi}^{(3)}(\Lambda_{1}\Lambda_{2}\Lambda),
\label{eq:3.9}
\end{equation}

\noindent{}-- for $m+n-\xi=2$. Each coefficient $\mathfrak{h}_{mn;\xi}^{(3)}$
is additionally expressed by the sum of $\mathfrak{h}_{mn;\xi}^{(3)+}$
and $\mathfrak{h}_{mn;\xi}^{(3)-}$. The coefficients $\mathfrak{h}_{mn;\xi}^{(3)+}$ are presented
in an explicit form in Tabs. \ref{Tab4}-\ref{Tab6}. The coefficients $\mathfrak{h}_{mn;\xi}^{(3)-}$ are derived
from $\mathfrak{h}_{mn;\xi}^{(3)+}$
by making the following alterations:

\begin{enumerate}[\upshape (a)]
\item $\Omega_{\alpha\bar{\beta}}^{(2)+}(\Lambda)\rightarrow 
(-1)^{L_{\alpha\bar{\beta}}+M+1} \Omega_{\alpha\bar{\beta}}^{(2)-}(\Lambda);$

\item $\Omega_{\alpha \beta \bar{\mu} \bar{\nu}}^{(2)+}(\Lambda_{1} \Lambda_{2} \Lambda)\rightarrow 
(-1)^{L_{\alpha \beta \bar{\mu} \bar{\nu}}+M} \Omega_{\alpha \beta \bar{\mu} \bar{\nu}}^{(2)-}(\Lambda_{1} \Lambda_{2} \Lambda);$

\item 
$\Omega_{\alpha \beta \zeta \bar{\mu} \bar{\nu} \bar{\eta}}^{(2)+}(\Lambda_{1} \Lambda_{2} \Lambda_{3} \Lambda)\rightarrow 
(-1)^{L_{\alpha \beta \zeta \bar{\mu} \bar{\nu} \bar{\eta}}+M+M_{3}+1} 
\Omega_{\alpha \beta \zeta \bar{\mu} \bar{\nu} \bar{\eta}}^{(2)-}(\Lambda_{1} \Lambda_{2} \Lambda_{3} \Lambda).$
\end{enumerate}

\noindent{}Here $L_{\alpha\beta\ldots\zeta}=\lambda_{\alpha}+\lambda_{\beta}+\ldots+\lambda_{\zeta}$. In addition, there holds
one more rule: (d) each basis index (if such exists) in $\mathfrak{h}_{mn;\xi}^{(3)+}$, except for $m_{0}$,
is replaced by the opposite sign index. In (a)-(c), the indices $M$ and $M_{3}$ enumerate the 
basis for $\Lambda$ and $\Lambda_{3}$, respectively. The quantities $\Omega^{(2)\pm}$ are given in Appendix \ref{A}, while

\begin{equation}
 \widetilde{\Omega}_{\alpha\beta\bar{\mu}\bar{\nu}}^{(2)\pm}(\Lambda_{1}\Lambda_{2}\Lambda)
=\Omega_{\alpha\beta\bar{\mu}\bar{\nu}}^{(2)\pm}(\Lambda_{1}\Lambda_{2}\Lambda)
-a(\lambda_{\alpha}
\lambda_{\beta}\Lambda_{1})\Omega_{\beta\alpha\bar{\mu}\bar{\nu}}^{(2)\pm}(\Lambda_{1}\Lambda_{2}\Lambda).
\label{eq:3.10}
\end{equation}

\begin{table}
\caption{\label{Tab6}The expansion coefficients for two-body terms of the third-order contribution to the 
effective Hamiltonian (continued)}
\begin{tabular}{ll}
\hline\hline
$(mn\xi)$ & \raisebox{3.ex}{}\raisebox{-1.2ex}{}$\mathfrak{h}_{mn;\xi}^{(3)+}(\Lambda_{1}\Lambda_{2}\Lambda)$ \\
\hline
\raisebox{20.ex}{}$(233)$ & $\begin{array}{l}
2[\Lambda_{2}]^{1/2}\sum_{\mathrm{c}}(-1)^{\lambda_{\mathrm{c}}}\sum_{\overline{\Lambda}_{1}
\overline{\Lambda}_{2}\Lambda_{3}\overline{\Lambda}}(-1)^{\overline{M}}[\Lambda_{3},\overline{\Lambda}]^{1/2}
\Bigl((-1)^{\overline{\Lambda}}\langle\Lambda_{3}M_{3}\overline{\Lambda}-\overline{M}\vert\Lambda M
\rangle\bigl[\sum_{\mathrm{c^{\prime}}}\bigl([\Lambda_{1},\overline{\Lambda}_{1}]^{1/2}\\
\times(-1)^{\lambda_{\mathrm{c^{\prime}}}+\overline{\Lambda}_{1}}
\sum_{\mu=\mathrm{v},\mathrm{e}}\widetilde{z}(0\lambda_{\mathrm{c}}\lambda_{\mathrm{c^{\prime}}}
\lambda_{\mathrm{\bar{v}^{\prime}}}\lambda_{\mu^{\prime\prime}}\overline{\Lambda}_{2}\overline{\Lambda}_{2})
\{\delta_{\mu\mathrm{v}}\Omega_{\mathrm{v}\mu^{\prime\prime}\mathrm{v^{\prime}\bar{v}c^{\prime}c}}^{(2)+}
(\overline{\Lambda}_{1}\overline{\Lambda}_{2}\Lambda_{3}\overline{\Lambda})-a(\lambda_{\mathrm{v^{\prime}}}
\lambda_{\mu^{\prime\prime}}\overline{\Lambda}_{1}) \\
\times\Omega_{\mathrm{vv^{\prime}}\mu^{\prime\prime}
\mathrm{\bar{v}c^{\prime}c}}^{(2)+}(\overline{\Lambda}_{1}\overline{\Lambda}_{2}\Lambda_{3}
\overline{\Lambda})\}\left\{ \begin{smallmatrix}\overline{\Lambda}_{1} & \overline{\Lambda}_{2} & \overline{\Lambda}\\
\lambda_{\mathrm{\bar{v}^{\prime}}} & \lambda_{\mathrm{v^{\prime}}} & \lambda_{\mu^{\prime\prime}}\end{smallmatrix}\right\}
\left\{ \begin{smallmatrix}\lambda_{\mathrm{v}} & \lambda_{\mathrm{v^{\prime}}} & \Lambda_{1}\\
\lambda_{\mathrm{\bar{v}}} & \lambda_{\mathrm{\bar{v}^{\prime}}} & \Lambda_{2}\\
\Lambda_{3} & \overline{\Lambda} & \Lambda\end{smallmatrix}\right\} 
-\delta_{\Lambda_{1}
\overline{\Lambda}_{1}}(-1)^{\lambda_{\mathrm{c^{\prime}}}+\Lambda_{1}} \\
\times\sum_{\mathrm{v^{\prime\prime}}}\widetilde{z}(0\lambda_{\mathrm{c}}\lambda_{\mathrm{c^{\prime}}}
\lambda_{\mathrm{\bar{v}^{\prime}}}\lambda_{\mathrm{v^{\prime\prime}}}\overline{\Lambda}_{2}\overline{\Lambda}_{2})
\Omega_{\mathrm{v^{\prime\prime}vv^{\prime}\bar{v}c^{\prime}c}}^{(2)+}
(\Lambda_{1}\overline{\Lambda}_{2}\Lambda_{3}\overline{\Lambda})\left\{ \begin{smallmatrix}\Lambda_{2} & 
\overline{\Lambda}_{2} & \Lambda_{3}\\
\lambda_{\mathrm{v^{\prime\prime}}} & \lambda_{\mathrm{\bar{v}}} & \lambda_{\mathrm{\bar{v}^{\prime}}}
\end{smallmatrix}\right\} \left\{ \begin{smallmatrix}\Lambda_{1} & \Lambda_{2} & \Lambda\\
\Lambda_{3} & \overline{\Lambda} & \overline{\Lambda}_{2}\end{smallmatrix}\right\} \bigr)
-(-1)^{\lambda_{\mathrm{v^{\prime}}}+\Lambda_{1}} \\
\times[\Lambda_{1},\overline{\Lambda}_{2}]^{1/2}
\sum_{\mathrm{v^{\prime\prime}}}\sum_{\mu=\mathrm{v},\mathrm{e}}(-1)^{\lambda_{\mathrm{v^{\prime\prime}}}
+\lambda_{\bar{\mu}^{\prime\prime}}+\overline{\Lambda}_{1}+\overline{\Lambda}_{2}}
\widetilde{z}(0\lambda_{\mathrm{v}}\lambda_{\mathrm{c}}\lambda_{\bar{\mu}^{\prime\prime}}
\lambda_{\mathrm{v^{\prime\prime}}}\overline{\Lambda}_{1}\overline{\Lambda}_{1})
\Omega_{\mathrm{v^{\prime}v^{\prime\prime}}\bar{\mu}^{\prime\prime}\mathrm{\bar{v}\bar{v}^{\prime}c}}^{(2)+}
(\overline{\Lambda}_{1}\overline{\Lambda}_{2}\Lambda_{3}\overline{\Lambda}) \\
\times\left\{ \begin{smallmatrix}\overline{\Lambda}_{1} & \overline{\Lambda}_{2} & \overline{\Lambda}\\
\lambda_{\mathrm{\bar{v}^{\prime}}} & \lambda_{\mathrm{v}} & \lambda_{\mathrm{c}}\end{smallmatrix}\right\} 
\left\{ \begin{smallmatrix}\lambda_{\mathrm{v^{\prime}}} & \lambda_{\mathrm{v}} & \Lambda_{1}\\
\lambda_{\mathrm{\bar{v}}} & \lambda_{\mathrm{\bar{v}^{\prime}}} & \Lambda_{2}\\
\Lambda_{3} & \overline{\Lambda} & \Lambda\end{smallmatrix}\right\} \bigr]
+(-1)^{\lambda_{\mathrm{\bar{v}}}+\Lambda_{2}+\Lambda_{3}+\Lambda}[\Lambda_{1},\overline{\Lambda}_{1},
\overline{\Lambda}_{2}]^{1/2}\langle\Lambda_{3}M_{3}\overline{\Lambda}\:\overline{M}\vert\Lambda M
\rangle\sum_{u}[u]^{1/2} \\
\times\bigl[(-1)^{\overline{\Lambda}_{1}}
\sum_{\mu=\mathrm{v},\mathrm{e}}\widetilde{z}(0\lambda_{\mathrm{v}}\lambda_{\mathrm{c}}
\lambda_{\bar{\mu}^{\prime\prime}}\lambda_{\mu^{\prime\prime}}uu)\Omega_{\mu^{\prime\prime}
\bar{\mu}^{\prime\prime}\mathrm{v^{\prime}\bar{v}\bar{v}^{\prime}c}}^{(2)+}
(\overline{\Lambda}_{1}\overline{\Lambda}_{2}\Lambda_{3}\overline{\Lambda})\left\{ \begin{smallmatrix}u &  & 
\lambda_{\mathrm{\bar{v}}}\:\:\lambda_{\mathrm{c}} &  & \lambda_{\mathrm{\bar{v}^{\prime}}}\:\:\lambda_{\mathrm{v}} &  & 
\Lambda_{2}\\
\  & \lambda_{\mu^{\prime\prime}} &  & \overline{\Lambda}_{2} &  & \Lambda_{1}\\
\lambda_{\bar{\mu}^{\prime\prime}} &  & \Lambda_{3}\:\:\overline{\Lambda}_{1} &  & 
\overline{\Lambda}\:\:\lambda_{\mathrm{v^{\prime}}} &  & \Lambda\end{smallmatrix}\right\} \\
+(-1)^{\lambda_{\mathrm{v^{\prime}}}}
\sum_{\mathrm{v^{\prime\prime}}}\bigl(\sum_{\mathrm{e}}(-1)^{\lambda_{\mathrm{e}}}
\widetilde{z}(0\lambda_{\mathrm{v}}\lambda_{\mathrm{c}}\lambda_{\mathrm{e}}
\lambda_{\mathrm{v^{\prime\prime}}}uu)\Omega_{\mathrm{v^{\prime\prime}v^{\prime}e\bar{v}\bar{v}^{\prime}c}}^{(2)+}
(\overline{\Lambda}_{1}\overline{\Lambda}_{2}\Lambda_{3}\overline{\Lambda}) \\
\times\left\{ \begin{smallmatrix}u &  & 
\lambda_{\mathrm{\bar{v}}}\:\:\lambda_{\mathrm{c}} &  & \lambda_{\mathrm{\bar{v}^{\prime}}}\:\:\lambda_{\mathrm{v}} &  & 
\Lambda_{2}\\
\  & \lambda_{\mathrm{v^{\prime\prime}}} &  & \overline{\Lambda}_{2} &  & \Lambda_{1}\\
\lambda_{\mathrm{e}} &  & \Lambda_{3}\:\:\overline{\Lambda}_{1} &  & \overline{\Lambda}\:\:
\lambda_{\mathrm{v^{\prime}}} &  & \Lambda\end{smallmatrix}\right\}
-(-1)^{\Lambda_{1}+u}
\sum_{\mathrm{\bar{v}^{\prime\prime}}}(-1)^{\lambda_{\mathrm{\bar{v}^{\prime\prime}}}}
\widetilde{z}(0\lambda_{\mathrm{v}}\lambda_{\mathrm{c}}\lambda_{\mathrm{\bar{v}^{\prime\prime}}}
\lambda_{\mathrm{v^{\prime\prime}}}uu) \\
\times\Omega_{\mathrm{\mathrm{\bar{v}^{\prime\prime}}
v^{\prime}v^{\prime\prime}\bar{v}\bar{v}^{\prime}c}}^{(2)+}(\overline{\Lambda}_{1}\overline{\Lambda}_{2}\Lambda_{3}
\overline{\Lambda})
\left\{ \begin{smallmatrix}u &  & \lambda_{\mathrm{\bar{v}}}\:\:\lambda_{\mathrm{c}} &  & 
\lambda_{\mathrm{\bar{v}^{\prime}}}\:\:\lambda_{\mathrm{v}} &  & \Lambda_{2}\\
\  & \lambda_{\mathrm{\bar{v}^{\prime\prime}}} &  & \overline{\Lambda}_{2} &  & \Lambda_{1}\\
\lambda_{\mathrm{v^{\prime\prime}}} &  & \Lambda_{3}\:\:\overline{\Lambda}_{1} &  & \overline{\Lambda}\:\:
\lambda_{\mathrm{v^{\prime}}} &  & \Lambda\end{smallmatrix}\right\} \bigr)\bigr]\Bigr)\end{array}$ \\ \\
\raisebox{12.5ex}{}$(244)$ & $\begin{array}{l}
2\delta_{\Lambda_{1}\Lambda_{2}}\delta_{\Lambda0}(-1)^{\lambda_{\mathrm{\bar{v}}}
+\lambda_{\mathrm{\bar{v}^{\prime}}}}\sum_{\mathrm{cc^{\prime}}\overline{\Lambda}_{2}}
\Bigl(\sum_{\mathrm{v^{\prime\prime}}}\bigl(\sum_{\mathrm{\bar{v}^{\prime\prime}}}
a(\lambda_{\mathrm{v^{\prime\prime}}}\lambda_{\mathrm{\bar{v}^{\prime\prime}}}\Lambda_{1})
\widetilde{z}(0\lambda_{\mathrm{c}}\lambda_{\mathrm{c^{\prime}}}\lambda_{\mathrm{\bar{v}^{\prime\prime}}}
\lambda_{\mathrm{v^{\prime\prime}}}\overline{\Lambda}_{2}\overline{\Lambda}_{2}) \\
\times\Omega_{\mathrm{vv^{\prime}v^{\prime\prime}\bar{v}^{\prime\prime}\bar{v}^{\prime}\bar{v}c^{\prime}c}}^{(2)}
(\Lambda_{1}\Lambda_{1}\overline{\Lambda}_{2}\overline{\Lambda}_{2}0)
+\sum_{\mathrm{e}\overline{\Lambda}}[\overline{\Lambda}][a(\lambda_{\mathrm{e}}\lambda_{\mathrm{v^{\prime\prime}}}
\overline{\Lambda}_{2})[\Lambda_{1},\overline{\Lambda}_{2}]^{-1/2}\widetilde{z}(0\lambda_{\mathrm{c}}
\lambda_{\mathrm{c^{\prime}}}\lambda_{\mathrm{v^{\prime\prime}}}\lambda_{\mathrm{e}}\overline{\Lambda}_{2}
\overline{\Lambda}_{2}) \\
\times\Omega_{\mathrm{ev^{\prime\prime}vv^{\prime}\bar{v}^{\prime}\bar{v}c^{\prime}c}}^{(2)}
(\overline{\Lambda}_{2}\Lambda_{1}\Lambda_{1}\overline{\Lambda}_{2}\overline{\Lambda})-\sum_{\Lambda_{3}\Lambda_{4}}
[\overline{\Lambda}_{2},\Lambda_{3}]^{1/2}(-1)^{\Lambda_{4}}
a(\lambda_{\mathrm{e}}\lambda_{\mathrm{v^{\prime}}}\Lambda_{1})\left\{ \begin{smallmatrix}\Lambda_{1} & 
\overline{\Lambda}_{2} & \overline{\Lambda}\\
\lambda_{\mathrm{e}} & \lambda_{\mathrm{v^{\prime}}} & \lambda_{\mathrm{v}}\end{smallmatrix}\right\} 
\left\{ \begin{smallmatrix}\Lambda_{3} & \Lambda_{4} & \overline{\Lambda}\\
\lambda_{\mathrm{e}} & \lambda_{\mathrm{v^{\prime}}} & \lambda_{\mathrm{v^{\prime\prime}}}
\end{smallmatrix}\right\} \\
\times\widetilde{z}(0\lambda_{\mathrm{c}}\lambda_{\mathrm{c^{\prime}}}
\lambda_{\mathrm{v^{\prime\prime}}}\lambda_{\mathrm{e}}\Lambda_{4}\Lambda_{4})
\{\Omega_{\mathrm{evv^{\prime}v^{\prime\prime}\bar{v}^{\prime}\bar{v}c^{\prime}c}}^{(2)}
(\overline{\Lambda}_{2}\Lambda_{1}\Lambda_{3}\Lambda_{4}\overline{\Lambda})-a(\lambda_{\mathrm{v^{\prime}}}
\lambda_{\mathrm{v^{\prime\prime}}}\Lambda_{3})
\Omega_{\mathrm{evv^{\prime\prime}v^{\prime}\bar{v}^{\prime}\bar{v}c^{\prime}c}}^{(2)}
(\overline{\Lambda}_{2}\Lambda_{1}\Lambda_{3}\Lambda_{4}\overline{\Lambda})\}]\bigr) \\
+\sum_{\mu=\mathrm{v},\mathrm{e}}\sum_{\overline{\Lambda}}[\overline{\Lambda}][\Lambda_{1},
\overline{\Lambda}_{2}]^{-1/2}\widetilde{z}(0\lambda_{\mathrm{c}}\lambda_{\mathrm{c^{\prime}}}
\lambda_{\bar{\mu}^{\prime\prime}}\lambda_{\mu^{\prime\prime}}\overline{\Lambda}_{2}\overline{\Lambda}_{2})
\Omega_{\mu^{\prime\prime}\bar{\mu}^{\prime\prime}\mathrm{vv^{\prime}\bar{v}^{\prime}\bar{v}c^{\prime}c}}^{(2)}
(\overline{\Lambda}_{2}\Lambda_{1}\Lambda_{1}\overline{\Lambda}_{2}\overline{\Lambda})\\
+\sum_{\mathrm{v^{\prime\prime}\bar{v}^{\prime\prime}}}\sum_{\Lambda_{3}\Lambda_{4}\overline{\Lambda}}
a(\Lambda_{1}\Lambda_{3}\lambda_{\mathrm{\bar{v}^{\prime\prime}}})[\overline{\Lambda}][\overline{\Lambda}_{2},
\Lambda_{3}]^{1/2}\widetilde{z}(0\lambda_{\mathrm{c}}\lambda_{\mathrm{c^{\prime}}}\lambda_{\mathrm{\bar{v}^{\prime\prime}}}
\lambda_{\mathrm{v^{\prime\prime}}}\Lambda_{4}\Lambda_{4})\left\{ \begin{smallmatrix}\Lambda_{1} & \overline{\Lambda}_{2} & 
\overline{\Lambda}\\
\lambda_{\mathrm{\bar{v}^{\prime\prime}}} & \lambda_{\mathrm{v^{\prime}}} & \lambda_{\mathrm{v}}\end{smallmatrix}\right\} 
\left\{ \begin{smallmatrix}\Lambda_{3} & \Lambda_{4} & \overline{\Lambda}\\
\lambda_{\mathrm{\bar{v}^{\prime\prime}}} & \lambda_{\mathrm{v^{\prime}}} & \lambda_{\mathrm{v^{\prime\prime}}}
\end{smallmatrix}\right\} \\
\times[\Omega_{\mathrm{\bar{v}^{\prime\prime}vv^{\prime\prime}v^{\prime}\bar{v}^{\prime}\bar{v}c^{\prime}c}}^{(2)}
(\overline{\Lambda}_{2}\Lambda_{1}\Lambda_{3}\Lambda_{4}\overline{\Lambda})+a(\lambda_{\mathrm{v}}
\lambda_{\mathrm{v^{\prime}}}\overline{\Lambda}_{2})a(\lambda_{\mathrm{v^{\prime\prime}}}
\lambda_{\mathrm{\bar{v}^{\prime\prime}}}\Lambda_{3})\Omega_{\mathrm{v\bar{v}^{\prime\prime}v^{\prime}
v^{\prime\prime}\bar{v}^{\prime}\bar{v}c^{\prime}c}}^{(2)}(\overline{\Lambda}_{2}\Lambda_{1}\Lambda_{3}\Lambda_{4}
\overline{\Lambda}) \\
-a(\Lambda_{3}\Lambda_{4}\lambda_{\mathrm{v^{\prime}}})
\{\Omega_{\mathrm{\bar{v}^{\prime\prime}vv^{\prime}v^{\prime\prime}\bar{v}^{\prime}\bar{v}c^{\prime}c}}^{(2)}
(\overline{\Lambda}_{2}\Lambda_{1}\Lambda_{3}\Lambda_{4}\overline{\Lambda})+(-1)^{\lambda_{\mathrm{v^{\prime\prime}}}
+\Lambda_{3}}a(\lambda_{\mathrm{v}}\lambda_{\mathrm{v^{\prime}}}\overline{\Lambda}_{2}) \\
\times\Omega_{\mathrm{v\bar{v}^{\prime\prime}v^{\prime\prime}v^{\prime}\bar{v}^{\prime}\bar{v}c^{\prime}c}}^{(2)}
(\overline{\Lambda}_{2}\Lambda_{1}\Lambda_{3}\Lambda_{4}\overline{\Lambda})\}]\Bigr)\end{array}$ \\
\hline\hline
\end{tabular}
\end{table}

\noindent{}Particularly, the $\Omega^{(2)-}$ are derived by replacing the one-particle and two-particle matrix elements 
$v_{\alpha\bar{\beta}}$ and
$v_{\alpha\beta\bar{\mu}\bar{\nu}}$ with $v_{\bar{\beta}\alpha}$ and $v_{\bar{\mu}\bar{\nu}\alpha\beta}$ 
in order to obtain the standard form of the constructions that lay out in Tabs. \ref{Tab1}-\ref{Tab3}.
For example, the element
$S_{\alpha\bar{\beta}}(\tau_{1}\tau_{2}\tau)$ (Tab. \ref{Tab1})
is recognized from $\sum_{\mu}v_{\alpha\mu}v_{\mu\bar{\beta}}/(\varepsilon_{\bar{\beta}}-\varepsilon_{\mu})$
by excluding the $\mathrm{SO}(3)$-invariant part.
Thus, for $\alpha=\mathrm{v}$, $\beta=\mathrm{c}$, $\mu=\mathrm{e}$,
we get $\ddot{S}_{\mathrm{v\bar{c}}}(\tau_{1}\tau_{2}\tau)$ which
fits the definition of one-particle operator along with $W_{m}^{\tau}(\lambda_{\mathrm{v}}
\widetilde{\lambda}_{\mathrm{\bar{c}}})$.
On the other hand, another one-particle matrix element could
be obtained from $\sum_{\mu}v_{\alpha\mu}v_{\mu\bar{\beta}}/(\varepsilon_{\mu}-\varepsilon_{\alpha})$.
The last element must be written in a standard form $(-1)\sum_{\mu}v_{\bar{\beta}\mu}v_{\mu\alpha}/
(\varepsilon_{\alpha}-\varepsilon_{\mu})$
to arrange the element $S_{\bar{\beta}\alpha}(\tau_{1}\tau_{2}\tau)$
correctly. For concrete values $\alpha=\mathrm{v}$ and $\beta=\mathrm{c}$,
the orbital $\mu$ equals to $\mathrm{c}$ only (see Proposition \ref{prop3}), and the element
is denoted $S_{\mathrm{\bar{c}v}}(\tau_{1}\tau_{2}\tau)$. In a tensor
formalism, we gain the opposite sign of the basis index $m$ in $W_{m}^{\tau}(\lambda_{\mathrm{v}}
\widetilde{\lambda}_{\mathrm{\bar{c}}})$.

The diagrammatic interpretation of $\widehat{h}_{mn;\xi}^{(3)-}$
can be clarified as follows. The diagrams of the third-order effective
Hamiltonian $H_{eff}^{(3)}$, assembled in $\widehat{h}_{mn;\xi}^{(3)-}$, are derived by contracting
the perturbation $V_{m}$ with $\widehat{\Omega}_{n}^{(2)-}$, where $\widehat{\Omega}_{n}^{(2)-}$
includes: (i) the folded
diagrams; (ii) some diagrams, obtained by contracting the core orbitals
in Wick's series; (iii) the diagrams, acceded to the reflection of
a number of diagrams in $\widehat{\Omega}_{n}^{(2)+}$ about a horizontal axis. Note, 
$\Omega_{\alpha\beta\zeta\rho\bar{\mu}\bar{\nu}\bar{\eta}\bar{\sigma}}^{(2)+}\equiv
\Omega_{\alpha\beta\zeta\rho\bar{\mu}\bar{\nu}\bar{\eta}\bar{\sigma}}^{(2)}$, since it is obtained from the
disconnected diagrams $V_{2}\widehat{\Omega}^{(1)}_{2}$ (see Tab. \ref{Tab3}).

\subsection{The estimation of terms}
\label{count}

The evaluation of the number of computed one-body terms 
(Tab. \ref{Tab4}) of $H_{eff}^{(3)}$ is presented in Tab. \ref{Tab7}, 
where $d^{\pm}$ denotes
the number of direct terms in $\widehat{h}_{mn;\xi}^{(3)\pm}$, while $\bar{d}^{\pm}$ denotes the number of
direct terms in $\widehat{h}_{mn;\xi}^{(3)\pm}$, if the one-body interactions $v^{\tau_{i}}$ ($i=0,1,2$) are absent.
Totally, there are computed $188+70=258$ direct one-body terms of $\widehat{h}_{mn;\xi}^{(3)}$ and $72+20=92$ direct one-body 
terms of $\widehat{h}_{mn;\xi}^{(3)}$, including the two-particle interactions $g^{0}$ only.
For instance, Blundell et. al. \cite{Blundell2} calculated $84$ diagrams 
contributing to the third-order mono-valent removal
energy. In our considerations, their studied energies: a) $E_{A}^{(3)}-E_{H}^{(3)}$; b) $E_{I}^{(3)},E_{J}^{(3)}$ and
c) $E_{K}^{(3)},E_{L}^{(3)}$ 
(see Ref. \cite[Sec. II, Eq. (8)]{Blundell2}) denote the matrix elements of terms in $\widehat{h}_{22;3}^{(3)}$, 
$\widehat{h}_{23;4}^{(3)}$ and $\widehat{h}_{21;2}^{(3)}$, if $g^{0}$ represents the Coulomb interaction.

The estimation of the amount of two-body terms (Tabs. \ref{Tab5}-\ref{Tab6}) that contribute to 
$H_{eff}^{(3)}$ is provided in Tab. \ref{Tab8}.
There are $217+82=299$ direct two-body terms in $\widehat{h}_{mn;\xi}^{(3)\pm}$ and $125+42=167$ direct two-body
terms including the two-particle interactions $g^{0}$ only. In their study of beryllium and magnesium isoelectronic 
sequences, Ho et. al.\cite{Ho} calculated $218$ two-body diagrams of the third-order perturbation. Analogous disposition
to account for the two-particle interactions only, can be found and in other works \cite{Safronova,Blundell2,Chou}.
Additionally, most of them do not account for the folded diagrams.

Meanwhile, the expressions in Tabs. \ref{Tab4}-\ref{Tab6}, obtained by exploiting the properties of proposed
model space $\mathcal{P}$ (Sec. \ref{model}), have their own benefits: 
\begin{quote}
1. The third-order contributions to the effective Hamiltonian $H_{eff}$
are written in an operator form providing an opportunity to construct their matrix elements efficiently. Namely, the 
irreducible tensor operators, labeled by the representations $\Lambda$ (see Eqs. (\ref{eq:3.8})-(\ref{eq:3.9})), are
written apart from the
projection-independent parts. These angular coefficients include the structure coefficients $\Omega^{(2)\pm}$, 
multiplied by the
$3nj$-symbols. Particularly, the coefficients $\Bigl\{\begin{smallmatrix}j_{1}& &j_{2}& &j_{3}& &j_{4}\\ 
&l_{1}& &l_{2}& &l_{3}& &l_{4}\\ k_{1}& &k_{2}& &k_{3}& &k_{4}\end{smallmatrix}\Bigr\}$ and
$\Bigl\{\begin{smallmatrix}k_{1}& &k_{1}^{\prime}\:\:k& &k^{\prime}\:\:k_{2}& &k_{2}^{\prime}\\ 
&p_{1}& &p& &p_{2}& &\\ j_{1}& &j_{1}^{\prime}\:\:j& &j^{\prime}\:\:j_{2}& &j_{2}^{\prime}\end{smallmatrix}\Bigr\}$
(see Tabs. \ref{Tab5}-\ref{Tab6}), denote the $12j$-symbol of the first kind \cite[Sec. 4-33, Eq. (33.17), p. 207]{Jucys2} 
and the $15j$-symbol of the third kind \cite[Sec. 4-20, Eq. (20.3), p. 112]{Jucys}. 

2. The form of the expressions in 
Tabs. \ref{Tab4}-\ref{Tab6}, allows to evaluate the contributions of $n$-particle effects in CC approach. 
This is done by simply replacing $\Omega^{(2)\pm}$ with $\Omega_{n}$ ($n=1,2,3,4$). By generally accepted labeling,
such replacement leads to the transformation
$g_{\alpha\beta\ldots\zeta}\rightarrow \rho_{\alpha\beta\ldots\zeta}$, where $\rho$ denotes
the valence singles, doubles, triples, quadruples amplitude. 

3. Obtained terms of the third-order perturbation include, in addition,
the one-particle operators $v^{\tau_{i}}$ ($\tau_{i}=0,1,2,\ldots$) that represent the magnetic, hyperfine, etc. interactions. 
Moreover,
the general expressions also fit none relativistic as well as the relativistic approaches. These effects are embodied
in $z(0\lambda_{\alpha}\lambda_{\beta}\lambda_{\bar{\nu}}\lambda_{\bar{\mu}}\Gamma_{1}\Gamma_{1})$ coefficients.

4. The reduction scheme allows to write a large number of terms 
in a concise form. This feature becomes evident especially clearly when in contrast the terms are written side by side their 
diagrammatic representation.
\end{quote}

\begin{table}
\caption{\label{Tab7}The amount of one-body terms in MBPT(3)}
\begin{tabular}{rrrrrrrrrrrrrr}
\hline\hline
$(mn\xi)$ & & & & $d^{+}$ & & & $\bar{d}^{+}$ & & & $d^{-}$ & & & $\bar{d}^{-}$ \\
\hline
$(111)$ & & & & $13$ & & & $0$ & & & $3$ & & & $0$ \\
$(122)$ & & & & $37$ & & & $0$ & & & $18$ & & & $0$ \\
$(212)$ & & & & $14$ & & & $2$ & & & $2$ & & & $0$ \\
$(223)$ & & & & $67$ & & & $34$ & & & $29$ & & & $2$ \\
$(234)$ & & & & $57$ & & & $36$ & & & $18$ & & & $18$ \\
\hline 
Total: & & & & $188$ & & & $72$ & & & $70$ & & & $20$ \\
\hline\hline
\end{tabular}
\end{table}

\begin{table}
\caption{\label{Tab8}The amount of two-body terms in MBPT(3)}
\begin{tabular}{rrrrrrrrrrrrrr}
\hline\hline
$(mn\xi)$ & & & & $d^{+}$ & & & $\bar{d}^{+}$ & & & $d^{-}$ & & & $\bar{d}^{-}$ \\
\hline
$(121)$ & & & & $20$ & & & $0$ & & & $10$ & & & $0$ \\
$(211)$ & & & & $13$ & & & $2$ & & & $3$ & & & $0$ \\
$(222)$ & & & & $64$ & & & $32$ & & & $31$ & & & $4$ \\
$(132)$ & & & & $20$ & & & $16$ & & & $10$ & & & $10$ \\
$(233)$ & & & & $75$ & & & $50$ & & & $28$ & & & $28$ \\
$(244)$ & & & & $25$ & & & $25$ & & & $-$ & & & $-$ \\
\hline 
Total: & & & & $217$ & & & $125$ & & & $82$ & & & $42$ \\
\hline\hline
\end{tabular}
\end{table}

Maintaining the completeness of the present discuss, we note that $H_{eff}^{(3)}$ also includes the terms
$\widehat{h}_{mn;\xi}^{(3)}$ with $m+n-\xi=0,3,4,5$. For example, the coefficient 
$\mathfrak{h}_{22;1}^{(3)+}(E_{1}\Lambda_{1}E_{2}\Lambda_{2}\Lambda)$ reads

\begin{align}
\mathfrak{h}_{22;1}^{(3)+}(E_{1}\Lambda_{1}E_{2}\Lambda_{2}\Lambda)
&=(-1)^{\lambda_{\mathrm{v^{\prime\prime}}}+\lambda_{\mathrm{\bar{v}^{\prime}}}
+\lambda_{\mathrm{\bar{v}^{\prime\prime}}}+\Lambda_{2}+\Lambda}[\Lambda_{1},\Lambda_{2}]^{1/2}
{\displaystyle \sum_{\overline{\Lambda}_{1}}}\Bigl({\displaystyle \sum_{\overline{\Lambda}_{2}}}
a(\lambda_{\mathrm{v^{\prime}}}\lambda_{\mathrm{\bar{v}}}\overline{\Lambda}_{2}) \nonumber \\
&\times[E_{1},E_{2},\overline{\Lambda}_{1}]^{1/2}{\displaystyle \sum_{\mathrm{c}u}}
\widetilde{z}(0\lambda_{\mathrm{c}}\lambda_{\mathrm{v}}\lambda_{\mathrm{\bar{v}^{\prime}}}
\lambda_{\mathrm{\bar{v}}}uu)\Omega_{\mathrm{v^{\prime}v^{\prime\prime}c\bar{v}^{\prime\prime}}}^{(2)+}
(\overline{\Lambda}_{1}\overline{\Lambda}_{2}\Lambda) \nonumber \\
&\times\left\{ \begin{smallmatrix}\lambda_{\mathrm{v}} & \overline{\Lambda}_{2} & \Lambda_{2}\\
\lambda_{\mathrm{\bar{v}^{\prime\prime}}} & u & \lambda_{\mathrm{c}}\end{smallmatrix}\right\} 
\left\{ \begin{smallmatrix}\Lambda_{2} & \lambda_{\mathrm{v}} & \overline{\Lambda}_{2}\\
\overline{\Lambda}_{1} & \Lambda & \Lambda_{1}\end{smallmatrix}\right\} \left\{ \begin{smallmatrix}
\lambda_{\mathrm{v}} & \lambda_{\mathrm{v^{\prime}}} & E_{1}\\
\lambda_{\mathrm{v^{\prime\prime}}} & \Lambda_{1} & \overline{\Lambda}_{1}\end{smallmatrix}\right\} 
\left\{ \begin{smallmatrix}\lambda_{\mathrm{\bar{v}^{\prime\prime}}} & \lambda_{\mathrm{\bar{v}^{\prime}}} & E_{2}\\
\lambda_{\mathrm{\bar{v}}} & \Lambda_{2} & u\end{smallmatrix}\right\} \nonumber \\
&+(-1)^{E_{1}+E_{2}}{\displaystyle \sum_{\mathrm{e}}}\widetilde{z}(0\lambda_{\mathrm{v}}
\lambda_{\mathrm{v^{\prime}}}\lambda_{\mathrm{e}}\lambda_{\mathrm{\bar{v}}}E_{1}E_{1})
\Omega_{\mathrm{ev^{\prime\prime}\bar{v}^{\prime}\bar{v}^{\prime\prime}}}^{(2)+}(\overline{\Lambda}_{1}E_{2}\Lambda) \nonumber \\
&\times\left\{ \begin{smallmatrix}\Lambda_{1} & \overline{\Lambda}_{1} & \lambda_{\mathrm{\bar{v}}}\\
\lambda_{\mathrm{e}} & E_{1} & \lambda_{\mathrm{v^{\prime\prime}}}\end{smallmatrix}\right\} 
\left\{ \begin{smallmatrix}\Lambda_{1} & \Lambda_{2} & \Lambda\\
E_{2} & \overline{\Lambda}_{1} & \lambda_{\mathrm{\bar{v}}}\end{smallmatrix}\right\} 
\Bigr),\label{fin}
\end{align}

\noindent{}and it forms the three-body operator $\widehat{h}_{22;1}^{(3)}$ along with the irreducible tensor operator

\begin{equation}
[[W^{E_{1}}(\lambda_{\mathrm{v}}\lambda_{\mathrm{v^{\prime}}})
\times a^{\lambda_{\mathrm{v^{\prime\prime}}}}]^{\Lambda_{1}}\times [W^{E_{2}}(\widetilde{\lambda}_{\mathrm{\bar{v}}^{\prime\prime}}
\widetilde{\lambda}_{\mathrm{\bar{v}^{\prime}}})\times \tilde{a}^{\lambda_{\mathrm{\bar{v}}}}]^{\Lambda_{2}}]_{M}^{\Lambda}.
\label{fin2}\end{equation}

\noindent{}However, the triple and higher-order effects are not covered by the examination of this paper.

\section{Conclusions}
\label{summary}

We present an algebraic technique to evaluate the terms of MBPT. The method relies on two main circumstances:
the strictly determined operations of the Fock space operators on the vectors of the orthogonal subspaces of given
separable Hilbert space, and the specific reduction scheme of terms of the PT. The aspiration to determine the behavior of
the creation and annihilation operators is motivated by the fact that in higher-order perturbation theories, a huge
number of terms is generated and, particularly, not all computed terms of the wave operator attach none zero contributions
to the terms of effective interaction operator. Therefore the rules that allow to predetermine these valuable terms
become meaningful. Meanwhile, in the PT, another no less important procedure is to work up the generated terms in 
order to calculate their matrix elements efficiently.
If going on the traditional route, when each term is expressed side by side its diagrammatic representation,
then once again one deals with the redundancy of terms and the procedure of reduction becomes troublesome.
In this paper, however, we suggest the reduction scheme which is of versatile disposition. That is, the generated
terms are combined in groups related to the different types of one-electron orbitals. Afterwards, we construct
the effective matrix elements and apply the Wigner-Eckart theorem. In the result, we get the irreducible tensor
operators and the projection-independent parts that are located apart from each other. 
This means the form of reduced terms becomes universal, only
the $\mathrm{SO}(3)$-invariant parts that embody the dynamics of studied physical interactions are changed.
Since the method of reduction is applied to the third-order MBPT, we obtain maximum four-body effective
matrix elements that are derived from the four-body parts of the second-order wave operator.

Finally, it is worth to mention that obtained symbolic preparation of terms of the MBPT(3) can be extremely
simplified if applied to some special cases of interest. Mathematically, this means that some of the irreducible
representations drawn in the $3nj$-symbols, in most cases would be simply equal to zero. Particularly,
the $12j$- or $15j$-symbols would become the ordinary and widely used $6j$- or $9j$-symbols. Moreover, if
the one-electron interactions, characterized by the none scalar representations,
are not taken into account, then the irreducible one-body and two-body operators
of the effective Hamiltonian are simply scalar operators. Thus their matrix elements become even simpler.

\newpage{}

\appendix

\section{\label{A}SO(3)-invariant part of the second-order wave operator}{\small
\begin{verse}
\emph{One-body part.}
\end{verse}

\begin{subequations}\label{5.2.3}
\begin{equation}
\begin{array}{l}
\Omega_{\mu\mathrm{c}}^{(2)+}(\Lambda)(\varepsilon_{\mathrm{c}}-\varepsilon_{\mu})\\
\\=\delta_{\Lambda\tau}[\ddot{S}_{\mu\mathrm{c}}(\tau_{1}\tau_{2}\tau)+\dot{S}_{\mu\mathrm{c}}(\tau_{1}\tau_{2}\tau)]
+\delta_{\Lambda\tau_{1}}[\widetilde{\ddot{S}}_{\mu\mathrm{c}}(\tau_{1})+\widetilde{\dot{S}}_{\mu\mathrm{c}}(\tau_{1})]\\
+\delta_{\Lambda\tau_{2}}[\widetilde{\ddot{S^{\prime}}}_{\mu\mathrm{c}}(\tau_{2})+\widetilde{\dot{S^{\prime}}}_{\mu\mathrm{c}}
(\tau_{2})]+\delta_{\Lambda0}\widetilde{S}_{\mu\mathrm{c}},
\end{array}\label{eq:5.2.3a}\end{equation}

\begin{equation}
\Omega_{\mu\mathrm{c}}^{(2)-}(\Lambda)(\varepsilon_{\mathrm{c}}-\varepsilon_{\mu})
=\delta_{\Lambda\tau}S_{\mathrm{c}\mu}(\tau_{1}\tau_{2}\tau).\label{eq:5.2.3b}\end{equation}

\end{subequations}\begin{subequations}\label{c1}

\begin{equation}
\begin{array}{l}
\Omega_{\mathrm{ev}}^{(2)+}(\Lambda)(\varepsilon_{\mathrm{v}}-\varepsilon_{\mathrm{e}})\\
\\=\delta_{\Lambda\tau}\ddot{S}_{\mathrm{ev}}(\tau_{1}\tau_{2}\tau)+\delta_{\Lambda\tau_{1}}[\widetilde{\ddot{S}}_{
\mathrm{ev}}(\tau_{1})+\widetilde{\dot{S}}_{\mathrm{ev}}(\tau_{1})]\\+\delta_{\Lambda\tau_{2}}[\widetilde{\ddot{S}^{\prime}}_{
\mathrm{ev}}(\tau_{2})+\widetilde{\dot{S}^{\prime}}_{\mathrm{ev}}(\tau_{2})]
+\delta_{\Lambda0}\widetilde{S}_{\mathrm{ev}},\end{array}\label{eq:c1a}\end{equation}

\begin{equation}
\Omega_{\mathrm{ev}}^{(2)-}(\Lambda)(\varepsilon_{\mathrm{v}}-\varepsilon_{\mathrm{e}})
=\delta_{\Lambda\tau}[\dot{S}_{\mathrm{ve}}(\tau_{1}\tau_{2}\tau)+\ddot{S}_{\mathrm{ve}}(
\tau_{1}\tau_{2}\tau)].\label{eq:c1b}\end{equation}

\end{subequations}
\begin{verse}
\emph{Two-body part.}
\end{verse}
\begin{subequations}\label{5.2.5}

\begin{equation}
\begin{array}{l}
\Omega_{\mu\mu^{\prime}\mathrm{cc^{\prime}}}^{(2)+}(\Lambda_{1}\Lambda_{2}\Lambda)(\varepsilon_{\mathrm{cc^{\prime}}}
-\varepsilon_{\mu\mu^{\prime}})\\
\\=\delta_{\Lambda_{1}u}\delta_{\Lambda_{2}d}\delta_{\Lambda\tau}D_{\mu\mu^{\prime}\mathrm{cc^{\prime}}}(ud\tau)
+\delta_{\Lambda_{1}U}\delta_{\Lambda_{2}u}\delta_{\Lambda\tau_{1}}[\mathfrak{X}_{\mu\mu^{\prime}\mathrm{
\mathrm{cc^{\prime}}}}(Uu\tau_{1})\\
+{\scriptstyle \frac{1}{2}}\widetilde{\mathfrak{Y}}_{\mu\mu^{\prime}\mathrm{cc^{\prime}}}(Uu\tau_{1})]
+{\scriptstyle \frac{1}{2}}\delta_{\Lambda_{1}U}\delta_{\Lambda_{2}u}\delta_{\Lambda\tau_{2}}\mathcal{Z}_{\mu^{\prime}
\mu\mathrm{cc^{\prime}}}(Uu\tau_{2})\\
\times\widetilde{D^{\prime}}_{\mu^{\prime}\mu\mathrm{cc^{\prime}}}(Uu\tau_{2})
-\delta_{\Lambda_{1}\Lambda_{2}}\delta_{\Lambda_{1}u}\delta_{\Lambda0}[{\scriptstyle \frac{1}{4}}\ddot{D}_{
\mu\mu^{\prime}\mathrm{cc^{\prime}}}(\widetilde{u}u)\\
+{\scriptstyle \frac{1}{4}}\dot{D}_{\mu\mu^{\prime}
\mathrm{cc^{\prime}}}(\widetilde{u}u)
+{\scriptstyle \frac{1}{2}}\dddot{\mathfrak{Z}}_{\mu\mu^{\prime}
\mathrm{cc^{\prime}}}(\widetilde{u}u)
+\widetilde{\dot{\Delta}}_{\mu\mu^{\prime}\mathrm{cc^{\prime}}}(uu)+\widetilde{\ddot{\Delta}}_{\mu\mu^{\prime}
\mathrm{cc^{\prime}}}(uu)],\end{array}\label{eq:5.2.5a}\end{equation}

\begin{equation}
\begin{array}{l}
\Omega_{\mu\mu^{\prime}\mathrm{cc^{\prime}}}^{(2)-}(\Lambda_{1}\Lambda_{2}\Lambda)(\varepsilon_{\mathrm{cc^{\prime}}}
-\varepsilon_{\mu\mu^{\prime}})\\
\\=-{\scriptstyle \frac{1}{2}}\delta_{\Lambda_{1}U}\delta_{\Lambda_{2}u}\delta_{\Lambda\tau_{1}}\mathcal{Z}_{
\mathrm{c^{\prime}c}\mu\mu^{\prime}}(uU\tau_{1})\widetilde{D}_{\mathrm{c^{\prime}c}\mu\mu^{\prime}}(uU\tau_{1})\\
+{\scriptstyle \frac{1}{2}}\delta_{\Lambda_{1}U}\delta_{\Lambda_{2}u}\delta_{\Lambda\tau_{2}}
\mathcal{Z}_{\mathrm{c^{\prime}c}\mu\mu^{\prime}}(uU\tau_{2})[\widetilde{\ddot{D^{\prime}}}_{\mathrm{c^{\prime}c}
\mu\mu^{\prime}}(uU\tau_{2})\\
+\widetilde{\dot{D^{\prime}}}_{\mathrm{c^{\prime}c}\mu\mu^{\prime}}(uU\tau_{2})]
+{\scriptstyle \frac{1}{4}}\delta_{\Lambda_{1}\Lambda_{2}}\delta_{\Lambda_{1}u}\delta_{\Lambda0}
\mathcal{Z}_{\mathrm{cc^{\prime}}\mu\mu^{\prime}}(u)D_{\mathrm{cc^{\prime}}\mu\mu^{\prime}}
(u\widetilde{u}).\end{array}\label{eq:5.2.5b}\end{equation}

\end{subequations}\begin{subequations}\label{5.2.6}

\begin{equation}
\begin{array}{l}
\Omega_{\mathrm{evcc^{\prime}}}^{(2)+}(\Lambda_{1}\Lambda_{2}\Lambda)(\varepsilon_{\mathrm{cc^{\prime}}}
-\varepsilon_{\mathrm{ev}})\\
\\=-\delta_{\Lambda_{1}u}\delta_{\Lambda_{2}d}\delta_{\Lambda\tau}\mathcal{Z}_{\mathrm{vecc^{\prime}}}(ud\tau)
D_{\mathrm{vecc^{\prime}}}(ud\tau)+\delta_{\Lambda_{1}U}\delta_{\Lambda_{2}u}\delta_{\Lambda\tau_{1}}\\
\times[\ddot{D}_{\mathrm{evcc^{\prime}}}(Uu\tau_{1})
+\mathcal{Z}_{\mathrm{vec^{\prime}c}}(Uu\tau_{1})\dot{D}_{\mathrm{vec^{\prime}c}}(Uu\tau_{1})\\
+{\scriptstyle \frac{1}{2}}\widetilde{\dot{D}}_{\mathrm{evcc^{\prime}}}(Uu\tau_{1})
-{\scriptstyle \frac{1}{2}}\mathcal{Z}_{\mathrm{vecc^{\prime}}}(Uu\tau_{1})
\widetilde{\ddot{D}}_{\mathrm{vecc^{\prime}}}(Uu\tau_{1})]\\
+{\scriptstyle \frac{1}{2}}\delta_{\Lambda_{1}U}\delta_{\Lambda_{2}u}\delta_{\Lambda\tau_{2}}
[-\widetilde{D^{\prime}}_{\mathrm{evcc^{\prime}}}(Uu\tau_{2})
+\mathcal{Z}_{\mathrm{vecc^{\prime}}}(Uu\tau_{2})\\
\times\widetilde{D^{\prime}}_{\mathrm{vecc^{\prime}}}(Uu\tau_{2})]
-\delta_{\Lambda_{1}\Lambda_{2}}\delta_{\Lambda_{1}u}\delta_{\Lambda0}
[\widetilde{\ddot{\Delta}}_{\mathrm{evcc^{\prime}}}(uu)
+\widetilde{\dot{\Delta}}_{\mathrm{evcc^{\prime}}}(uu)\\
-\mathcal{Z}_{\mathrm{vecc^{\prime}}}(u)
\{\widetilde{\ddot{\Delta}}_{\mathrm{vecc^{\prime}}}(uu)
+\widetilde{\dot{\Delta}}_{\mathrm{vecc^{\prime}}}(uu)\}
+\mathcal{Z}_{\mathrm{evc^{\prime}c}}(u)\dddot{D}_{\mathrm{evc^{\prime}c}}(u\widetilde{u})],
\end{array}\label{eq:5.2.6a}\end{equation}

\begin{equation}
\begin{array}{l}
\Omega_{\mathrm{evcc^{\prime}}}^{(2)-}(\Lambda_{1}\Lambda_{2}\Lambda)(\varepsilon_{\mathrm{cc^{\prime}}}
-\varepsilon_{\mathrm{ev}})\\
\\=-\delta_{\Lambda_{1}U}\delta_{\Lambda_{2}u}\delta_{\Lambda\tau_{1}}\mathcal{Z}_{\mathrm{c^{\prime}cev}}(uU\tau_{1})
\widetilde{D}_{\mathrm{c^{\prime}cev}}(uU\tau_{1})\\
+\delta_{\Lambda_{1}U}\delta_{\Lambda_{2}u}\delta_{\Lambda\tau_{2}}
\mathcal{Z}_{\mathrm{c^{\prime}cev}}(uU\tau_{2})
[\widetilde{\ddot{D^{\prime}}}_{\mathrm{c^{\prime}cev}}(uU\tau_{2})\\
+\widetilde{\dot{D^{\prime}}}_{\mathrm{c^{\prime}cev}}(uU\tau_{2})].\end{array}\label{eq:5.2.6b}\end{equation}

\end{subequations}\begin{subequations}\label{c2}

\begin{equation}
\begin{array}{l}
\Omega_{\mu\mu^{\prime}\mathrm{c\bar{v}}}^{(2)+}(\Lambda_{1}\Lambda_{2}\Lambda)(\varepsilon_{\mathrm{c\bar{v}}}
-\varepsilon_{\mu\mu^{\prime}})\\
\\=\delta_{\Lambda_{1}u}\delta_{\Lambda_{2}d}\delta_{\Lambda\tau}\delta_{\mu\mathrm{e}}
[D_{\mathrm{ee^{\prime}c\bar{v}}}(ud\tau)-\mathcal{Z}_{\mathrm{ee^{\prime}\bar{v}c}}(ud\tau)
D_{\mathrm{ee^{\prime}\bar{v}c}}(ud\tau)]\\
+\delta_{\Lambda_{1}U}\delta_{\Lambda_{2}u}\delta_{\Lambda\tau_{1}}
[\widetilde{\ddot{D}}_{\mu\mu^{\prime}\mathrm{c\bar{v}}}(Uu\tau_{1})
-\mathcal{Z}_{\mu\mu^{\prime}\mathrm{\bar{v}c}}(Uu\tau_{1})\\
\times\widetilde{\dot{D}}_{\mu\mu^{\prime}\mathrm{\bar{v}c}}(Uu\tau_{1})]
+\delta_{\Lambda_{1}U}\delta_{\Lambda_{2}u}\delta_{\Lambda\tau_{2}}
\mathcal{Z}_{\mu^{\prime}\mu\mathrm{c\bar{v}}}(Uu\tau_{2})
\widetilde{D^{\prime}}_{\mu^{\prime}\mu\mathrm{c\bar{v}}}(Uu\tau_{2})\\
+\delta_{\Lambda_{1}\Lambda_{2}}\delta_{\Lambda_{1}u}\delta_{\Lambda0}
[\mathcal{Z}_{\mu\mu^{\prime}\mathrm{\bar{v}c}}(u)\{\widetilde{\ddot{\Delta}}_{\mu\mu^{\prime}\mathrm{\bar{v}c}}(uu)
+\widetilde{\dot{\Delta}}_{\mu\mu^{\prime}\mathrm{\bar{v}c}}(uu)\}\\
-\widetilde{\ddot{\Delta}}_{\mu\mu^{\prime}\mathrm{c\bar{v}}}(uu)
-\widetilde{\dot{\Delta}}_{\mu\mu^{\prime}\mathrm{c\bar{v}}}(uu)
-{\scriptstyle \frac{1}{2}}\widetilde{\dddot{D}}_{\mu\mu^{\prime}\mathrm{c\bar{v}}}(uu)],\end{array}\label{eq:c2a}\end{equation}

\begin{equation}
\begin{array}{l}
\Omega_{\mu\mu^{\prime}\mathrm{c\bar{v}}}^{(2)-}(\Lambda_{1}\Lambda_{2}\Lambda)(\varepsilon_{\mathrm{c\bar{v}}}
-\varepsilon_{\mu\mu^{\prime}})\\
\\=-\delta_{\Lambda_{1}U}\delta_{\Lambda_{2}u}\delta_{\Lambda\tau_{1}}
[\mathcal{Z}_{\mathrm{\bar{v}c}\mu\mu^{\prime}}(uU\tau_{1})\{\dot{D}_{\mathrm{\bar{v}c}\mu\mu^{\prime}}(uU\tau_{1})\\
+{\scriptstyle \frac{1}{2}}\widetilde{D}_{\mathrm{\bar{v}c}\mu\mu^{\prime}}(uU\tau_{1})\}
+\mathcal{Z}_{\mathrm{c\bar{v}}\mu^{\prime}\mu}(uU\tau_{1})D_{\mathrm{c\bar{v}}\mu^{\prime}\mu}(uU\tau_{1})]\\
-\delta_{\Lambda_{1}U}\delta_{\Lambda_{2}u}\delta_{\Lambda\tau_{2}}
[\mathcal{Z}_{\mathrm{c\bar{v}}\mu\mu^{\prime}}(uU\tau_{2})
\{\widetilde{\ddot{D}^{\prime}}_{\mathrm{c\bar{v}}\mu\mu^{\prime}}(uU\tau_{2})\\
+\widetilde{\dot{D}^{\prime}}_{\mathrm{c\bar{v}}\mu\mu^{\prime}}(uU\tau_{2})\}
-\mathcal{Z}_{\mathrm{\bar{v}c}\mu\mu^{\prime}}(uU\tau_{2})
\widetilde{\ddot{D}^{\prime}}_{\mathrm{\bar{v}c}\mu\mu^{\prime}}(uU\tau_{2})].\end{array}\label{eq:c2b}\end{equation}

\end{subequations}\begin{subequations}\label{c3}

\begin{equation}
\begin{array}{l}
\Omega_{\mathrm{ev\bar{v}c}}^{(2)+}(\Lambda_{1}\Lambda_{2}\Lambda)(\varepsilon_{\mathrm{c\bar{v}}}-\varepsilon_{\mathrm{ev}})\\
\\=\delta_{\Lambda_{1}u}\delta_{\Lambda_{2}d}\delta_{\Lambda\tau}[D_{\mathrm{ev\bar{v}c}}(ud\tau)
+\mathcal{Z}_{\mathrm{vec\bar{v}}}(ud\tau)D_{\mathrm{vec\bar{v}}}(ud\tau)]\\
+\delta_{\Lambda_{1}U}\delta_{\Lambda_{2}u}\delta_{\Lambda\tau_{1}}[\widetilde{\ddot{D}}_{\mathrm{ev\bar{v}c}}(Uu\tau_{1})
+\widetilde{\dot{D}}_{\mathrm{ev\bar{v}c}}(Uu\tau_{1})-\mathcal{Z}_{\mathrm{ve\bar{v}c}}(Uu\tau_{1})\\
\times\{\widetilde{\ddot{D}}_{\mathrm{ve\bar{v}c}}(Uu\tau_{1})+\widetilde{\dot{D}}_{\mathrm{ve\bar{v}c}}(Uu\tau_{1})\}]
+\delta_{\Lambda_{1}U}\delta_{\Lambda_{2}u}\delta_{\Lambda\tau_{2}}\\
\times[-\widetilde{D^{\prime}}_{\mathrm{ev\bar{v}c}}(Uu\tau_{2})
+\mathcal{Z}_{\mathrm{ve\bar{v}c}}(Uu\tau_{2})\widetilde{D^{\prime}}_{\mathrm{ve\bar{v}c}}(Uu\tau_{2})]\\
-\delta_{\Lambda_{1}\Lambda_{2}}\delta_{\Lambda_{1}u}\delta_{\Lambda0}
[\widetilde{\ddot{\Delta}}_{\mathrm{ev\bar{v}c}}(uu)+\widetilde{\dot{\Delta}}_{\mathrm{ev\bar{v}c}}(uu)
+\widetilde{\dddot{D}}_{\mathrm{ev\bar{v}c}}(uu)\\
-\mathcal{Z}_{\mathrm{ve\bar{v}c}}(u)\{\widetilde{\ddot{\Delta}}_{\mathrm{ve\bar{v}c}}(uu)
+\widetilde{\dot{\Delta}}_{\mathrm{ve\bar{v}c}}(uu)\}
-\mathcal{Z}_{\mathrm{evc\bar{v}}}(u)\{\widetilde{\ddot{\Delta}}_{\mathrm{evc\bar{v}}}(uu)\\
+\widetilde{\dot{\Delta}}_{\mathrm{evc\bar{v}}}(uu)\}
+\mathcal{Z}_{\mathrm{vec\bar{v}}}(u)\{\widetilde{\ddot{\Delta}}_{\mathrm{vec\bar{v}}}(uu)
+\widetilde{\dot{\Delta}}_{\mathrm{vec\bar{v}}}(uu)\}],\end{array}\label{eq:c3a}\end{equation}

\begin{equation}
\begin{array}{l}
\Omega_{\mathrm{ev\bar{v}c}}^{(2)-}(\Lambda_{1}\Lambda_{2}\Lambda)(\varepsilon_{\mathrm{c\bar{v}}}-\varepsilon_{\mathrm{ev}})\\
\\=\delta_{\Lambda_{1}U}\delta_{\Lambda_{2}u}\delta_{\Lambda\tau_{1}}
[\mathcal{Z}_{\mathrm{\bar{v}cev}}(uU\tau_{1})\{\widetilde{D}_{\mathrm{\bar{v}cev}}(uU\tau_{1})
-\widetilde{\dot{D}}_{\mathrm{\bar{v}cev}}(uU\tau_{1})\}\\
-\mathcal{Z}_{\mathrm{c\bar{v}ev}}(uU\tau_{1})\widetilde{D}_{\mathrm{c\bar{v}ev}}(uU\tau_{1})]
+\delta_{\Lambda_{1}U}\delta_{\Lambda_{2}u}\delta_{\Lambda\tau_{2}}[-\mathcal{Z}_{\mathrm{\bar{v}cev}}(uU\tau_{2})\\
\times\widetilde{\ddot{D}^{\prime}}_{\mathrm{\bar{v}cev}}(uU\tau_{2})
+\mathcal{Z}_{\mathrm{c\bar{v}ev}}(uU\tau_{2})\{\widetilde{\ddot{D}^{\prime}}_{\mathrm{c\bar{v}ev}}(uU\tau_{2})\\
-\widetilde{\dot{D}^{\prime}}_{\mathrm{c\bar{v}ev}}(uU\tau_{2})\}].\end{array}\label{eq:c3b}\end{equation}

\end{subequations}\begin{subequations}\label{c7}

\begin{equation}
\begin{array}{l}
\Omega_{\mathrm{ev\bar{v}\bar{v}^{\prime}}}^{(2)+}(\Lambda_{1}\Lambda_{2}\Lambda)(\varepsilon_{\mathrm{\bar{v}\bar{v}^{\prime}}}
-\varepsilon_{\mathrm{ev}})\\
\\=\delta_{\Lambda_{1}U}\delta_{\Lambda_{2}u}\delta_{\Lambda\tau_{1}}
[\ddot{D}_{\mathrm{ev\bar{v}\bar{v}^{\prime}}}(Uu\tau_{1})+\mathcal{Z}_{\mathrm{ve\bar{v}\bar{v}^{\prime}}}(Uu\tau_{1})\\
\times\{\dot{D}_{\mathrm{ve\bar{v}\bar{v}^{\prime}}}(Uu\tau_{1})
-{\scriptstyle \frac{1}{2}}\widetilde{\ddot{D}}_{\mathrm{ve\bar{v}\bar{v}^{\prime}}}(Uu\tau_{1})\}]
+{\scriptstyle \frac{1}{2}}\delta_{\Lambda_{1}U}\delta_{\Lambda_{2}u}\delta_{\Lambda\tau_{2}}\\
\times[\mathcal{Z}_{\mathrm{ve\bar{v}\bar{v}^{\prime}}}(Uu\tau_{2})
\widetilde{D^{\prime}}_{\mathrm{ve\bar{v}\bar{v}^{\prime}}}(Uu\tau_{2})
-\widetilde{D^{\prime}}_{\mathrm{ev\bar{v}\bar{v}^{\prime}}}(Uu\tau_{2})\\
-\widetilde{\dot{D}^{\prime}}_{\mathrm{ev\bar{v}\bar{v}^{\prime}}}(Uu\tau_{2})]
+\delta_{\Lambda_{1}\Lambda_{2}}\delta_{\Lambda_{1}u}\delta_{\Lambda0}
[-\widetilde{\ddot{\Delta}}_{\mathrm{ev\bar{v}\bar{v}^{\prime}}}(uu)\\
-\widetilde{\dot{\Delta}}_{\mathrm{ev\bar{v}\bar{v}^{\prime}}}(uu)
+\mathcal{Z}_{\mathrm{ve\bar{v}\bar{v}^{\prime}}}(u)
\{\widetilde{\ddot{\Delta}}_{\mathrm{ve\bar{v}\bar{v}^{\prime}}}(uu)
+\widetilde{\dot{\Delta}}_{\mathrm{ve\bar{v}\bar{v}^{\prime}}}(uu)\}],\end{array}\label{eq:c7a}\end{equation}

\begin{equation}
\begin{array}{l}
\Omega_{\mathrm{ev\bar{v}\bar{v}^{\prime}}}^{(2)-}(\Lambda_{1}\Lambda_{2}\Lambda)(\varepsilon_{\mathrm{\bar{v}\bar{v}^{\prime}}}
-\varepsilon_{\mathrm{ev}})\\
\\=-\delta_{\Lambda_{1}U}\delta_{\Lambda_{2}u}\delta_{\Lambda\tau_{1}}
\mathcal{Z}_{\mathrm{\bar{v}^{\prime}\bar{v}ev}}(uU\tau_{1})[\widetilde{D}_{\mathrm{\bar{v}^{\prime}\bar{v}ev}}(uU\tau_{1})\\
+\widetilde{\dot{D}}_{\mathrm{\bar{v}^{\prime}\bar{v}ev}}(uU\tau_{1})]
+\delta_{\Lambda_{1}U}\delta_{\Lambda_{2}u}\delta_{\Lambda\tau_{2}}\mathcal{Z}_{\mathrm{\bar{v}^{\prime}\bar{v}ev}}(uU\tau_{2})\\
\times\widetilde{\ddot{D}}_{\mathrm{\bar{v}^{\prime}\bar{v}ev}}(uU\tau_{2})
-{\scriptstyle \frac{1}{2}}\delta_{\Lambda_{1}\Lambda_{2}}\delta_{\Lambda_{1}u}\delta_{\Lambda0}
\mathcal{Z}_{\mathrm{\bar{v}\bar{v}^{\prime}ev}}(u)\\
\times\dot{D}_{\mathrm{\bar{v}\bar{v}^{\prime}ev}}(u\widetilde{u}).\end{array}\label{eq:c7b}\end{equation}

\end{subequations}\begin{subequations}\label{c8}

\begin{equation}
\begin{array}{l}
\Omega_{\mathrm{ee^{\prime}\bar{v}\bar{v}^{\prime}}}^{(2)+}(\Lambda_{1}\Lambda_{2}\Lambda)
(\varepsilon_{\mathrm{\bar{v}\bar{v}^{\prime}}}-\varepsilon_{\mathrm{ee^{\prime}}})\\
\\=\delta_{\Lambda_{1}u}\delta_{\Lambda_{2}d}\delta_{\Lambda\tau}
D_{\mathrm{ee^{\prime}\bar{v}\bar{v}^{\prime}}}(ud\tau)
+\delta_{\Lambda_{1}U}\delta_{\Lambda_{2}u}\delta_{\Lambda\tau_{1}}[{\scriptstyle \frac{1}{2}}
\widetilde{\ddot{D}}_{\mathrm{ee^{\prime}\bar{v}\bar{v}^{\prime}}}(Uu\tau_{1})\\
+\dot{D}_{\mathrm{ee^{\prime}\bar{v}\bar{v}^{\prime}}}(Uu\tau_{1})]
+{\scriptstyle \frac{1}{2}}\delta_{\Lambda_{1}U}\delta_{\Lambda_{2}u}\delta_{\Lambda\tau_{2}}
\mathcal{Z}_{\mathrm{e^{\prime}e\bar{v}\bar{v}^{\prime}}}(Uu\tau_{2})\\
\times\widetilde{D^{\prime}}_{\mathrm{e^{\prime}e\bar{v}\bar{v}^{\prime}}}(Uu\tau_{2})
-\delta_{\Lambda_{1}\Lambda_{2}}\delta_{\Lambda_{1}u}\delta_{\Lambda0}
[{\scriptstyle \frac{1}{2}}\{{\scriptstyle \frac{1}{2}}\ddot{D}_{\mathrm{ee^{\prime}\bar{v}\bar{v}^{\prime}}}(\widetilde{u}u)\\
-\mathcal{Z}_{\mathrm{ee^{\prime}\bar{v}^{\prime}\bar{v}}}(u)
\dddot{D}_{\mathrm{ee^{\prime}\bar{v}^{\prime}\bar{v}}}(\widetilde{u}u)\}
+\widetilde{\ddot{\Delta}}_{\mathrm{ee^{\prime}\bar{v}\bar{v}^{\prime}}}(uu)
+\widetilde{\dot{\Delta}}_{\mathrm{ee^{\prime}\bar{v}\bar{v}^{\prime}}}(uu)],\end{array}\label{eq:c8a}\end{equation}

\begin{equation}
\begin{array}{l}
\Omega_{\mathrm{ee^{\prime}\bar{v}\bar{v}^{\prime}}}^{(2)-}(\Lambda_{1}\Lambda_{2}\Lambda)
(\varepsilon_{\mathrm{\bar{v}\bar{v}^{\prime}}}-\varepsilon_{\mathrm{ee^{\prime}}})\\
\\=-\delta_{\Lambda_{1}U}\delta_{\Lambda_{2}u}\delta_{\Lambda\tau_{1}}
\mathcal{Z}_{\mathrm{\bar{v}^{\prime}\bar{v}ee^{\prime}}}(uU\tau_{1})
[D_{\mathrm{\bar{v}^{\prime}\bar{v}ee^{\prime}}}(uU\tau_{1})\\
+\widetilde{\dot{D}}_{\mathrm{\bar{v}^{\prime}\bar{v}ee^{\prime}}}(uU\tau_{1})]
+{\scriptstyle \frac{1}{2}}\delta_{\Lambda_{1}U}\delta_{\Lambda_{2}u}\delta_{\Lambda\tau_{2}}
\mathcal{Z}_{\mathrm{\bar{v}^{\prime}\bar{v}ee^{\prime}}}(uU\tau_{2})\\
\times\widetilde{\ddot{D}^{\prime}}_{\mathrm{\bar{v}^{\prime}\bar{v}ee^{\prime}}}(uU\tau_{2})
+{\scriptstyle \frac{1}{4}}\delta_{\Lambda_{1}\Lambda_{2}}\delta_{\Lambda_{1}u}\delta_{\Lambda0}
\mathcal{Z}_{\mathrm{\bar{v}\bar{v}^{\prime}ee^{\prime}}}(u)\\
\times[D_{\mathrm{\bar{v}\bar{v}^{\prime}ee^{\prime}}}(u\widetilde{u})
-\dot{D}_{\mathrm{\bar{v}\bar{v}^{\prime}ee^{\prime}}}(u\widetilde{u})].\end{array}\label{eq:c8b}\end{equation}

\end{subequations}

\noindent{}In Eq. (\ref{eq:5.2.5a}), the quantities $\mathfrak{X}_{\mu\mu^{\prime}\mathrm{cc^{\prime}}}$,
$\mathfrak{Y}_{\mu\mu^{\prime}\mathrm{cc^{\prime}}}$, $\mathfrak{Z}_{\mu\mu^{\prime}\mathrm{cc^{\prime}}}$
differ for distinct one-electron orbitals $\mu=\mathrm{v},\mathrm{e}$.
If $\mu=\mathrm{v}$, then $\mathfrak{X}_{\mathrm{vv^{\prime}cc^{\prime}}}\equiv\ddot{D}_{\mathrm{vv^{\prime}cc^{\prime}}}$,
$\mathfrak{Y}_{\mathrm{vv^{\prime}cc^{\prime}}}\equiv\dot{D}_{\mathrm{vv^{\prime}cc^{\prime}}}$
and $\mathfrak{Z}_{\mathrm{vv^{\prime}cc^{\prime}}}\equiv-\mathcal{Z}_{\mathrm{vv^{\prime}c^{\prime}c}}(u)
D_{\mathrm{vv^{\prime}c^{\prime}c}}$;
if $\mu=\mathrm{e}$, then $\mathfrak{X}_{\mathrm{ee^{\prime}cc^{\prime}}}\equiv\dot{D}_{\mathrm{ee^{\prime}cc^{\prime}}}$,
$\mathfrak{Y}_{\mathrm{ee^{\prime}cc^{\prime}}}\equiv\ddot{D}_{\mathrm{ee^{\prime}cc^{\prime}}}$,
$\mathfrak{Z}_{\mathrm{ee^{\prime}cc^{\prime}}}\equiv D_{\mathrm{ee^{\prime}cc^{\prime}}}$.
\begin{verse}
\emph{Three-body part.}
\end{verse}
\begin{subequations}\label{c4}

\begin{equation}
\begin{array}{l}
\Omega_{\mathrm{vv^{\prime}v^{\prime\prime}\bar{v}c^{\prime}c}}^{(2)+}
(\Lambda_{1}\Lambda_{2}\Lambda_{3}\Lambda)(\varepsilon_{\mathrm{\bar{v}c^{\prime}c}}
-\varepsilon_{\mathrm{vv^{\prime}v^{\prime\prime}}})\\
\\={\scriptstyle \frac{1}{2}}\delta_{\Lambda\Lambda_{3}}
\delta_{MM_{3}}[\ddot{T}_{\mathrm{vv^{\prime}v^{\prime\prime}\bar{v}c^{\prime}c}}
(\widetilde{\Lambda}_{1}\Lambda_{2}\Lambda)+{\scriptstyle \frac{1}{2}}
\widetilde{\dot{T}}_{\mathrm{vv^{\prime}v^{\prime\prime}\bar{v}c^{\prime}c}}(\Lambda_{1}\Lambda_{2}\Lambda)\\
-\sum_{{\scriptscriptstyle \overline{\Lambda}_{2}\overline{\Lambda}}}
J_{\mathrm{vc\bar{v}c^{\prime}}}(\overline{\Lambda}\:\overline{\Lambda}_{2}\Lambda\Lambda_{2}\Lambda_{1})
\{\widetilde{\ddot{T}}_{\mathrm{vv^{\prime}v^{\prime\prime}c^{\prime}\bar{v}c}}
(\Lambda_{1}\overline{\Lambda}_{2}\overline{\Lambda})\\
+\widetilde{\dot{T}}_{\mathrm{vv^{\prime}v^{\prime\prime}c^{\prime}\bar{v}c}}
(\Lambda_{1}\overline{\Lambda}_{2}\overline{\Lambda})\}]+(-1)^{\lambda_{\mathrm{v}}
+\lambda_{\mathrm{c^{\prime}}}}Y_{\mathrm{vc\bar{v}c^{\prime}}}(\Lambda_{1}\Lambda_{2}\Lambda\Lambda_{3}\tau_{1})\\
\times T_{\mathrm{vv^{\prime}v^{\prime\prime}c^{\prime}c\bar{v}}}(\Lambda_{1}\tau_{1})
+{\scriptstyle \frac{1}{2}}\sum_{u}a(\lambda_{\mathrm{v^{\prime}}}\lambda_{\mathrm{\bar{v}}}u)\\
\times I_{\mathrm{vv^{\prime}v^{\prime\prime}\bar{v}c^{\prime}c}}
(\Lambda_{1}\Lambda_{2}\Lambda_{3}\tau_{1}\Lambda u)
\widetilde{T^{\prime}}_{\mathrm{v^{\prime\prime}vv^{\prime}c\bar{v}c^{\prime}}}(u\tau_{2}),
\end{array}\label{eq:c4a}\end{equation}

\begin{equation}
\begin{array}{l}
\Omega_{\mathrm{vv^{\prime}v^{\prime\prime}\bar{v}c^{\prime}c}}^{(2)-}
(\Lambda_{1}\Lambda_{2}\Lambda_{3}\Lambda)(\varepsilon_{\mathrm{\bar{v}c^{\prime}c}}
-\varepsilon_{\mathrm{vv^{\prime}v^{\prime\prime}}})\\
\\={\scriptstyle \frac{1}{2}}\delta_{\Lambda\Lambda_{3}}\delta_{MM_{3}}
a(\lambda_{\mathrm{v^{\prime}}}\lambda_{\mathrm{v^{\prime\prime}}}\Lambda_{1})
[a(\lambda_{\mathrm{c}}\lambda_{\mathrm{c^{\prime}}}\Lambda_{2})
\widetilde{T}_{\mathrm{\bar{v}c^{\prime}cvv^{\prime}v^{\prime\prime}}}(\Lambda_{2}\Lambda_{1}\Lambda)\\
+\sum_{{\scriptscriptstyle \overline{\Lambda}_{2}\overline{\Lambda}}}
a(\lambda_{\mathrm{c}}\lambda_{\mathrm{c^{\prime}}}\overline{\Lambda}_{2})
J_{\mathrm{vc^{\prime}\bar{v}c}}(\overline{\Lambda}\:\overline{\Lambda}_{2}\Lambda\Lambda_{2}\Lambda_{1})
T_{\mathrm{cc^{\prime}\bar{v}vv^{\prime}v^{\prime\prime}}}
(\widetilde{\overline{\Lambda}}_{2}\Lambda_{1}\overline{\Lambda})].\end{array}\label{eq:c4b}\end{equation}

\end{subequations}\begin{subequations}\label{c5}

\begin{equation}
\begin{array}{l}
\Omega_{\mathrm{vv^{\prime}e\bar{v}c^{\prime}c}}^{(2)+}(\Lambda_{1}\Lambda_{2}\Lambda_{3}\Lambda)
(\varepsilon_{\mathrm{\bar{v}c^{\prime}c}}-\varepsilon_{\mathrm{vv^{\prime}e}})\\
\\=\delta_{\Lambda\Lambda_{3}}\delta_{MM_{3}}[{\scriptstyle \frac{1}{2}}
\{-a(\lambda_{\mathrm{c}}\lambda_{\mathrm{c^{\prime}}}\Lambda_{2})
\dot{T}_{\mathrm{vv^{\prime}e\bar{v}cc^{\prime}}}(\widetilde{\Lambda}_{1}\Lambda_{2}\Lambda)\\
+{\scriptstyle \frac{1}{2}}\widetilde{\ddot{T}}_{\mathrm{vv^{\prime}e\bar{v}c^{\prime}c}}(\Lambda_{1}\Lambda_{2}\Lambda)\}
-a(\lambda_{\mathrm{e}}\lambda_{\mathrm{v^{\prime}}}\Lambda_{1})\{\ddot{T}_{\mathrm{vev^{\prime}\bar{v}c^{\prime}c}}
(\widetilde{\Lambda}_{1}\Lambda_{2}\Lambda)\\
+{\scriptstyle \frac{1}{2}}\widetilde{\dot{T}}_{\mathrm{vev^{\prime}\bar{v}c^{\prime}c}}
(\Lambda_{1}\Lambda_{2}\Lambda)\}-\sum_{{\scriptscriptstyle \overline{\Lambda}_{2}\overline{\Lambda}}}
J_{\mathrm{vc\bar{v}c^{\prime}}}(\overline{\Lambda}\:\overline{\Lambda}_{2}\Lambda\Lambda_{2}\Lambda_{1})\\
\times\{{\scriptstyle \frac{1}{2}}(\widetilde{\dot{T}}_{\mathrm{vv^{\prime}ec^{\prime}\bar{v}c}}
(\Lambda_{1}\overline{\Lambda}_{2}\overline{\Lambda})
+\widetilde{\ddot{T}}_{\mathrm{vv^{\prime}ec^{\prime}\bar{v}c}}(\Lambda_{1}\overline{\Lambda}_{2}\overline{\Lambda}))\\
+a(\lambda_{\mathrm{e}}\lambda_{\mathrm{v^{\prime}}}\Lambda_{1})
a(\lambda_{\mathrm{c}}\lambda_{\mathrm{\bar{v}}}\overline{\Lambda}_{2})
(\widetilde{\dot{T}}_{\mathrm{vev^{\prime}c^{\prime}c\bar{v}}}(\Lambda_{1}\overline{\Lambda}_{2}\overline{\Lambda})\\
+\widetilde{\ddot{T}}_{\mathrm{vev^{\prime}c^{\prime}c\bar{v}}}(\Lambda_{1}\overline{\Lambda}_{2}\overline{\Lambda}))\}]
+(-1)^{\lambda_{\mathrm{v}}+\lambda_{\mathrm{c^{\prime}}}}
Y_{\mathrm{vc\bar{v}c^{\prime}}}(\Lambda_{1}\Lambda_{2}\Lambda\Lambda_{3}\tau_{1})\\
\times\widetilde{T}_{\mathrm{vv^{\prime}ec^{\prime}c\bar{v}}}(\Lambda_{1}\tau_{1})
+{\scriptstyle \frac{1}{2}}a(\lambda_{\mathrm{e}}\lambda_{\mathrm{\bar{v}}}\Lambda_{2})
Y_{\mathrm{vv^{\prime}e\bar{v}}}^{\prime}(\tau_{1}\Lambda_{1}\Lambda\Lambda_{3}\Lambda_{2})\\
\times T_{\mathrm{evv^{\prime}\bar{v}c^{\prime}c}}(\Lambda_{2}\tau_{1})+\sum_{u}
[a(\lambda_{\mathrm{c}}\lambda_{\mathrm{v^{\prime}}}\Lambda_{2})\\
\times I_{\mathrm{vv^{\prime}e\bar{v}c^{\prime}c}}(\Lambda_{1}\Lambda_{2}\Lambda_{3}\tau_{1}\Lambda u)
T_{\mathrm{evv^{\prime}c^{\prime}c\bar{v}}}(u\tau_{1})+{\scriptstyle \frac{1}{2}}
a(\lambda_{\mathrm{v^{\prime}}}\lambda_{\mathrm{\bar{v}}}u)\\
\times I_{\mathrm{vv^{\prime}e\bar{v}c^{\prime}c}}(\Lambda_{1}\Lambda_{2}\Lambda_{3}\tau_{2}\Lambda u)
\widetilde{T^{\prime}}_{\mathrm{evv^{\prime}c\bar{v}c^{\prime}}}
(u\tau_{2})+(-1)^{\lambda_{\mathrm{v}}-\lambda_{\mathrm{\bar{v}}}}\\
\times a(\lambda_{\mathrm{e}}\lambda_{\mathrm{v^{\prime}}}\Lambda_{1})
I_{\mathrm{vev^{\prime}\bar{v}c^{\prime}c}}(\Lambda_{1}\Lambda_{2}\Lambda_{3}\tau_{2}\Lambda u)
\widetilde{T^{\prime}}_{\mathrm{v^{\prime}evc\bar{v}c^{\prime}}}(u\tau_{2})],\end{array}\label{eq:c5a}\end{equation}

\begin{equation}
\begin{array}{l}
\Omega_{\mathrm{vv^{\prime}e\bar{v}c^{\prime}c}}^{(2)-}(\Lambda_{1}\Lambda_{2}\Lambda_{3}\Lambda)
(\varepsilon_{\mathrm{\bar{v}c^{\prime}c}}-\varepsilon_{\mathrm{vv^{\prime}e}})\\
\\=\delta_{\Lambda\Lambda_{3}}\delta_{MM_{3}}a(\lambda_{\mathrm{e}}\lambda_{\mathrm{v^{\prime}}}\Lambda_{1})
[a(\lambda_{\mathrm{c}}\lambda_{\mathrm{c^{\prime}}}\Lambda_{2})
\widetilde{T}_{\mathrm{\bar{v}c^{\prime}cvv^{\prime}e}}(\Lambda_{2}\Lambda_{1}\Lambda)\\
+{\scriptstyle \frac{1}{2}}\sum_{{\scriptscriptstyle \overline{\Lambda}_{2}\overline{\Lambda}}}
\{J_{\mathrm{vc\bar{v}c^{\prime}}}(\overline{\Lambda}\:\overline{\Lambda}_{2}\Lambda\Lambda_{2}\Lambda_{1})
\widetilde{T}_{\mathrm{c^{\prime}c\bar{v}vv^{\prime}e}}(\overline{\Lambda}_{2}\Lambda_{1}\overline{\Lambda})\\
+(-1)^{\lambda_{\mathrm{v}}-\lambda_{\mathrm{e}}}\sum_{{\scriptscriptstyle \overline{\Lambda}_{1}}}
I_{\mathrm{vv^{\prime}e\bar{v}c^{\prime}c}}^{\prime}
(\overline{\Lambda}_{1}\overline{\Lambda}_{2}\Lambda_{1}\Lambda_{2}\overline{\Lambda}\Lambda)
(T_{\mathrm{c^{\prime}c\bar{v}evv^{\prime}}}(\widetilde{\overline{\Lambda}}_{2}\overline{\Lambda}_{1}\overline{\Lambda})\\
-a(\lambda_{\mathrm{c}}\lambda_{\mathrm{\bar{v}}}\overline{\Lambda}_{2})
\widetilde{T}_{\mathrm{c^{\prime}\bar{v}cevv^{\prime}}}
(\overline{\Lambda}_{2}\overline{\Lambda}_{1}\overline{\Lambda}))\}].\end{array}\label{eq:c5b}\end{equation}

\end{subequations}\begin{subequations}\label{c6}

\begin{equation}
\begin{array}{l}
\Omega_{\mathrm{ee^{\prime}v\bar{v}c^{\prime}c}}^{(2)+}(\Lambda_{1}\Lambda_{2}\Lambda_{3}\Lambda)
(\varepsilon_{\mathrm{\bar{v}c^{\prime}c}}-\varepsilon_{\mathrm{ee^{\prime}v}})\\
\\=\delta_{\Lambda\Lambda_{3}}\delta_{MM_{3}}[{\scriptstyle \frac{1}{2}}
\{\ddot{T}_{\mathrm{ee^{\prime}v\bar{v}c^{\prime}c}}(\widetilde{\Lambda}_{1}\Lambda_{2}\Lambda)
+{\scriptstyle \frac{1}{2}}\widetilde{\dot{T}}_{\mathrm{ee^{\prime}v\bar{v}c^{\prime}c}}(\Lambda_{1}\Lambda_{2}\Lambda)\}\\
+a(\lambda_{\mathrm{v}}\lambda_{\mathrm{e^{\prime}}}\Lambda_{1})
\{a(\lambda_{\mathrm{c}}\lambda_{\mathrm{c^{\prime}}}\Lambda_{2})
\dot{T}_{\mathrm{eve^{\prime}\bar{v}cc^{\prime}}}(\widetilde{\Lambda}_{1}\Lambda_{2}\Lambda)\\
-{\scriptstyle \frac{1}{2}}\widetilde{\ddot{T}}_{\mathrm{eve^{\prime}\bar{v}c^{\prime}c}}
(\Lambda_{1}\Lambda_{2}\Lambda)\}+\sum_{{\scriptscriptstyle \overline{\Lambda}_{2}\overline{\Lambda}}}
J_{\mathrm{ec\bar{v}c^{\prime}}}(\overline{\Lambda}\:\overline{\Lambda}_{2}\Lambda\Lambda_{2}\Lambda_{1})\\
\times\{{\scriptstyle \frac{1}{2}}(a(\lambda_{\mathrm{c}}\lambda_{\mathrm{\bar{v}}}\overline{\Lambda}_{2})
\widetilde{\ddot{T}}_{\mathrm{ee^{\prime}vc^{\prime}c\bar{v}}}(\Lambda_{1}\overline{\Lambda}_{2}\overline{\Lambda})
-\widetilde{\dot{T}}_{\mathrm{ee^{\prime}vc^{\prime}\bar{v}c}}(\Lambda_{1}\overline{\Lambda}_{2}\overline{\Lambda}))\\
+a(\lambda_{\mathrm{v}}\lambda_{\mathrm{e^{\prime}}}\Lambda_{1})
(\widetilde{\ddot{T}}_{\mathrm{eve^{\prime}c^{\prime}\bar{v}c}}
(\Lambda_{1}\overline{\Lambda}_{2}\overline{\Lambda})
+\widetilde{\dot{T}}_{\mathrm{eve^{\prime}c^{\prime}\bar{v}c}}(\Lambda_{1}\overline{\Lambda}_{2}\overline{\Lambda}))\}]\\
+\delta_{\Lambda_{1}\Lambda_{2}}\delta_{\Lambda_{3}\tau_{1}}\delta_{\Lambda0}
T_{\mathrm{ee^{\prime}v\bar{v}c^{\prime}c}}(\Lambda_{1}\tau_{1})+(-1)^{\lambda_{\mathrm{e}}+\lambda_{\mathrm{c^{\prime}}}}\\
\times Y_{\mathrm{ec\bar{v}c^{\prime}}}(\Lambda_{1}\Lambda_{2}\Lambda\Lambda_{3}\tau_{1})
\widetilde{T}_{\mathrm{ee^{\prime}vc^{\prime}c\bar{v}}}(\Lambda_{1}\tau_{1})+\sum_{u}
[a(\lambda_{\mathrm{e^{\prime}}}\lambda_{\mathrm{c^{\prime}}}\Lambda_{2})\\
\times I_{\mathrm{ee^{\prime}v\bar{v}cc^{\prime}}}(\Lambda_{1}\Lambda_{2}\Lambda_{3}\tau_{1}\Lambda u)
T_{\mathrm{vee^{\prime}c^{\prime}c\bar{v}}}(u\tau_{1})+a(\lambda_{\mathrm{e^{\prime}}}\lambda_{\mathrm{\bar{v}}}u)\\
\times\{{\scriptstyle \frac{1}{2}}I_{\mathrm{ee^{\prime}v\bar{v}c^{\prime}c}}
(\Lambda_{1}\Lambda_{2}\Lambda_{3}\tau_{2}\Lambda u)
\widetilde{T^{\prime}}_{\mathrm{vee^{\prime}c\bar{v}c^{\prime}}}(u\tau_{2})\\
+(-1)^{\lambda_{\mathrm{e}}-\lambda_{\mathrm{v}}}
I_{\mathrm{eve^{\prime}\bar{v}c^{\prime}c}}(\Lambda_{1}\Lambda_{2}\Lambda_{3}\tau_{2}\Lambda u)
\widetilde{T^{\prime}}_{\mathrm{e^{\prime}vec\bar{v}c^{\prime}}}(u\tau_{2})\}],\end{array}\label{eq:c6a}\end{equation}

\begin{equation}
\begin{array}{l}
\Omega_{\mathrm{ee^{\prime}v\bar{v}c^{\prime}c}}^{(2)-}(\Lambda_{1}\Lambda_{2}\Lambda_{3}\Lambda)
(\varepsilon_{\mathrm{\bar{v}c^{\prime}c}}-\varepsilon_{\mathrm{ee^{\prime}v}})\\
\\=\delta_{\Lambda\Lambda_{3}}\delta_{MM_{3}}a(\lambda_{\mathrm{v}}\lambda_{\mathrm{e^{\prime}}}\Lambda_{1})
[a(\lambda_{\mathrm{c}}\lambda_{\mathrm{c^{\prime}}}\Lambda_{2})
\widetilde{T}_{\mathrm{\bar{v}c^{\prime}cee^{\prime}v}}(\Lambda_{2}\Lambda_{1}\Lambda)\\
+{\scriptstyle \frac{1}{2}}\sum_{{\scriptscriptstyle \overline{\Lambda}_{2}\overline{\Lambda}}}
\{J_{\mathrm{ec\bar{v}c^{\prime}}}(\overline{\Lambda}\:\overline{\Lambda}_{2}\Lambda\Lambda_{2}\Lambda_{1})
\widetilde{T}_{\mathrm{c^{\prime}c\bar{v}ee^{\prime}v}}(\overline{\Lambda}_{2}\Lambda_{1}\overline{\Lambda})\\
+(-1)^{\lambda_{\mathrm{e}}-\lambda_{\mathrm{v}}}\sum_{{\scriptscriptstyle \overline{\Lambda}_{1}}}
I_{\mathrm{ee^{\prime}v\bar{v}c^{\prime}c}}^{\prime}(\overline{\Lambda}_{1}
\overline{\Lambda}_{2}\Lambda_{1}\Lambda_{2}\overline{\Lambda}\Lambda)
(T_{\mathrm{c^{\prime}c\bar{v}vee^{\prime}}}(\widetilde{\overline{\Lambda}}_{2}\overline{\Lambda}_{1}\overline{\Lambda})\\
-a(\lambda_{\mathrm{c}}\lambda_{\mathrm{\bar{v}}}\overline{\Lambda}_{2})
\widetilde{T}_{\mathrm{c^{\prime}\bar{v}cvee^{\prime}}}
(\overline{\Lambda}_{2}\overline{\Lambda}_{1}\overline{\Lambda}))\}].\end{array}\label{eq:c6b}\end{equation}

\end{subequations}\begin{subequations}\label{c9}

\begin{equation}
\begin{array}{l}
\Omega_{\mathrm{vv^{\prime}v^{\prime\prime}\bar{v}\bar{v}^{\prime}c}}^{(2)+}
(\Lambda_{1}\Lambda_{2}\Lambda)(\varepsilon_{\mathrm{\bar{v}\bar{v}^{\prime}c}}
-\varepsilon_{\mathrm{vv^{\prime}v^{\prime\prime}}})\\
\\={\scriptstyle \frac{1}{2}}\delta_{\Lambda\Lambda_{3}}\delta_{MM_{3}}
[\widetilde{\ddot{T}}_{\mathrm{vv^{\prime}v^{\prime\prime}\bar{v}\bar{v}^{\prime}c}}(\Lambda_{1}\Lambda_{2}\Lambda)
+\widetilde{\dot{T}}_{\mathrm{vv^{\prime}v^{\prime\prime}\bar{v}\bar{v}^{\prime}c}}(\Lambda_{1}\Lambda_{2}\Lambda)\\
+a(\lambda_{\mathrm{c}}\lambda_{\mathrm{\bar{v}^{\prime}}}\Lambda_{2})
\sum_{{\scriptscriptstyle \overline{\Lambda}_{2}\overline{\Lambda}}}
J_{\mathrm{v\bar{v}^{\prime}\bar{v}c}}(\overline{\Lambda}\:\overline{\Lambda}_{2}\Lambda\Lambda_{2}\Lambda_{1})
\ddot{T}_{\mathrm{vv^{\prime}v^{\prime\prime}c\bar{v}\bar{v}^{\prime}}}
(\widetilde{\Lambda}_{1}\overline{\Lambda}_{2}\overline{\Lambda})],\end{array}\label{eq:c9a}\end{equation}

\begin{equation}
\begin{array}{l}
\Omega_{\mathrm{vv^{\prime}v^{\prime\prime}\bar{v}\bar{v}^{\prime}c}}^{(2)-}
(\Lambda_{1}\Lambda_{2}\Lambda)(\varepsilon_{\mathrm{\bar{v}\bar{v}^{\prime}c}}
-\varepsilon_{\mathrm{vv^{\prime}v^{\prime\prime}}})\\
\\=\delta_{\Lambda\Lambda_{3}}\delta_{MM_{3}}
a(\lambda_{\mathrm{v^{\prime}}}\lambda_{\mathrm{v^{\prime\prime}}}\Lambda_{2})
[{\scriptstyle \frac{1}{4}}a(\lambda_{\mathrm{c}}\lambda_{\mathrm{\bar{v}^{\prime}}}\Lambda_{1})
\widetilde{T}_{\mathrm{\bar{v}\bar{v}^{\prime}cvv^{\prime}v^{\prime\prime}}}(\Lambda_{2}\Lambda_{1}\Lambda)\\
-T_{\mathrm{\bar{v}c\bar{v}^{\prime}vv^{\prime}v^{\prime\prime}}}
(\widetilde{\Lambda}_{2}\Lambda_{1}\Lambda)].\end{array}\label{eq:c9b}\end{equation}

\end{subequations}\begin{subequations}\label{c10}

\begin{equation}
\begin{array}{l}
\Omega_{\mathrm{vv^{\prime}e\bar{v}\bar{v}^{\prime}c}}^{(2)+}
(\Lambda_{1}\Lambda_{2}\Lambda_{3}\Lambda)
(\varepsilon_{\mathrm{\bar{v}\bar{v}^{\prime}c}}-\varepsilon_{\mathrm{vv^{\prime}e}})\\
\\=\delta_{\Lambda\Lambda_{3}}\delta_{MM_{3}}[{\scriptstyle \frac{1}{2}}
\{\widetilde{\ddot{T}}_{\mathrm{vv^{\prime}e\bar{v}\bar{v}^{\prime}c}}(\Lambda_{1}\Lambda_{2}\Lambda)
+\widetilde{\dot{T}}_{\mathrm{vv^{\prime}e\bar{v}\bar{v}^{\prime}c}}(\Lambda_{1}\Lambda_{2}\Lambda)\}\\
-a(\lambda_{\mathrm{e}}\lambda_{\mathrm{v^{\prime}}}\Lambda_{1})
\{\widetilde{\ddot{T}}_{\mathrm{vev^{\prime}\bar{v}\bar{v}^{\prime}c}}(\Lambda_{1}\Lambda_{2}\Lambda)
+\widetilde{\dot{T}}_{\mathrm{vev^{\prime}\bar{v}\bar{v}^{\prime}c}}(\Lambda_{1}\Lambda_{2}\Lambda)\}\\
+a(\lambda_{\mathrm{c}}\lambda_{\mathrm{\bar{v}^{\prime}}}\Lambda_{2})
\sum_{{\scriptscriptstyle \overline{\Lambda}_{2}\overline{\Lambda}}}
J_{\mathrm{v\bar{v}^{\prime}\bar{v}c}}(\overline{\Lambda}\:\overline{\Lambda}_{2}\Lambda\Lambda_{2}\Lambda_{1})\\
\times\{{\scriptstyle \frac{1}{4}}\widetilde{\ddot{T}}_{\mathrm{vv^{\prime}ec\bar{v}\bar{v}^{\prime}}}
(\Lambda_{1}\overline{\Lambda}_{2}\overline{\Lambda})-{\scriptstyle \frac{1}{2}}
a(\lambda_{\mathrm{\bar{v}}}\lambda_{\mathrm{\bar{v}^{\prime}}}\overline{\Lambda}_{2})
\dot{T}_{\mathrm{vv^{\prime}ec\bar{v}^{\prime}\bar{v}}}(\widetilde{\Lambda}_{1}\overline{\Lambda}_{2}\overline{\Lambda})\\
-a(\lambda_{\mathrm{e}}\lambda_{\mathrm{v^{\prime}}}\Lambda_{1})
\ddot{T}_{\mathrm{vev^{\prime}c\bar{v}\bar{v}^{\prime}}}
(\widetilde{\Lambda}_{1}\overline{\Lambda}_{2}\overline{\Lambda})\}]+(-1)^{\lambda_{\mathrm{e}}-\lambda_{\mathrm{\bar{v}}}}\\
\times[(-1)^{\lambda_{\mathrm{c}}+\lambda_{\mathrm{\bar{v}^{\prime}}}}
Y_{\mathrm{vv^{\prime}e\bar{v}}}^{\prime}(\tau_{1}\Lambda_{1}\Lambda\Lambda_{3}\Lambda_{2})
T_{\mathrm{evv^{\prime}\bar{v}c\bar{v}^{\prime}}}(\Lambda_{2}\tau_{1})\\
+a(\lambda_{\mathrm{v}}\lambda_{\mathrm{v^{\prime}}}\Lambda_{2})
Y_{\mathrm{v\bar{v}^{\prime}\bar{v}c}}(\Lambda_{1}\Lambda_{2}\Lambda\Lambda_{3}\tau_{1})
T_{\mathrm{vev^{\prime}c\bar{v}\bar{v}^{\prime}}}(\Lambda_{1}\tau_{1})]\\
+{\scriptstyle \frac{1}{2}}\sum_{u}[(-1)^{\lambda_{\mathrm{e}}-\lambda_{\mathrm{\bar{v}}}}
a(\lambda_{\mathrm{v}}\lambda_{\mathrm{v^{\prime}}}\Lambda_{1})
I_{\mathrm{vev^{\prime}\bar{v}\bar{v}^{\prime}c}}(\Lambda_{1}\Lambda_{2}\Lambda_{3}\tau_{2}\Lambda u)\\
\times\widetilde{T^{\prime}}_{\mathrm{v^{\prime}evc\bar{v}\bar{v}^{\prime}}}(u\tau_{2})
+a(\lambda_{\mathrm{v^{\prime}}}\lambda_{\mathrm{\bar{v}^{\prime}}}\Lambda_{2})\\
\times I_{\mathrm{vv^{\prime}e\bar{v}c\bar{v}^{\prime}}}(\Lambda_{1}\Lambda_{2}\Lambda_{3}\tau_{2}\Lambda u)
\widetilde{T^{\prime}}_{\mathrm{evv^{\prime}\bar{v}^{\prime}c\bar{v}}}(u\tau_{2})],\end{array}\label{eq:c10a}\end{equation}

\begin{equation}
\begin{array}{l}
\Omega_{\mathrm{vv^{\prime}e\bar{v}\bar{v}^{\prime}c}}^{(2)-}(\Lambda_{1}\Lambda_{2}\Lambda_{3}\Lambda)
(\varepsilon_{\mathrm{\bar{v}\bar{v}^{\prime}c}}-\varepsilon_{\mathrm{vv^{\prime}e}})\\
\\=\delta_{\Lambda\Lambda_{3}}\delta_{MM_{3}}[\widetilde{T}_{\mathrm{\bar{v}c\bar{v}^{\prime}vev^{\prime}}}
(\Lambda_{2}\Lambda_{1}\Lambda)-{\scriptstyle \frac{1}{2}}
a(\lambda_{\mathrm{c}}\lambda_{\mathrm{\bar{v}^{\prime}}}\Lambda_{2})\\
\times\widetilde{T}_{\mathrm{\bar{v}\bar{v}^{\prime}cvev^{\prime}}}(\Lambda_{2}\Lambda_{1}\Lambda)
+\sum_{{\scriptscriptstyle \overline{\Lambda}_{1}\overline{\Lambda}_{2}\overline{\Lambda}}}
a(\lambda_{\mathrm{v}}\lambda_{\mathrm{v^{\prime}}}\overline{\Lambda}_{2})\\
\times\{a(\lambda_{\mathrm{\bar{v}}}\lambda_{\mathrm{\bar{v}^{\prime}}}\overline{\Lambda}_{1})
T_{\mathrm{c\bar{v}\bar{v}^{\prime}evv^{\prime}}}(\widetilde{\overline{\Lambda}}_{2}\overline{\Lambda}_{1}\overline{\Lambda})\\
\times I_{\mathrm{vv^{\prime}e\bar{v}c\bar{v}^{\prime}}}^{\prime}(\overline{\Lambda}_{1}
\overline{\Lambda}_{2}\Lambda_{1}\Lambda_{2}\overline{\Lambda}\Lambda)
-{\scriptstyle \frac{1}{4}}a(\lambda_{\mathrm{c}}\lambda_{\mathrm{\bar{v}}}\overline{\Lambda}_{1})\\
\times\widetilde{T}_{\mathrm{\bar{v}^{\prime}\bar{v}cevv^{\prime}}}
(\overline{\Lambda}_{2}\overline{\Lambda}_{1}\overline{\Lambda})
I_{\mathrm{vv^{\prime}e\bar{v}\bar{v}^{\prime}c}}^{\prime}
(\overline{\Lambda}_{1}\overline{\Lambda}_{2}\Lambda_{1}
\Lambda_{2}\overline{\Lambda}\Lambda)\}].\end{array}\label{eq:c10b}\end{equation}

\end{subequations}\begin{subequations}\label{c11}

\begin{equation}
\begin{array}{l}
\Omega_{\mathrm{ee^{\prime}v\bar{v}\bar{v}^{\prime}c}}^{(2)+}(\Lambda_{1}\Lambda_{2}\Lambda_{3}\Lambda)
(\varepsilon_{\mathrm{\bar{v}\bar{v}^{\prime}c}}-\varepsilon_{\mathrm{ee^{\prime}v}})\\
\\=\delta_{\Lambda\Lambda_{3}}\delta_{MM_{3}}[{\scriptstyle \frac{1}{2}}
\{\widetilde{\ddot{T}}_{\mathrm{ee^{\prime}v\bar{v}\bar{v}^{\prime}c}}(\Lambda_{1}\Lambda_{2}\Lambda)
+\widetilde{\dot{T}}_{\mathrm{ee^{\prime}v\bar{v}\bar{v}^{\prime}c}}(\Lambda_{1}\Lambda_{2}\Lambda)\}\\
-a(\lambda_{\mathrm{e^{\prime}}}\lambda_{\mathrm{v}}\Lambda_{1})
\{\widetilde{\ddot{T}}_{\mathrm{eve^{\prime}\bar{v}\bar{v}^{\prime}c}}
(\Lambda_{1}\Lambda_{2}\Lambda)+\widetilde{\dot{T}}_{\mathrm{eve^{\prime}\bar{v}\bar{v}^{\prime}c}}
(\Lambda_{1}\Lambda_{2}\Lambda)\}\\
+a(\lambda_{\mathrm{c}}\lambda_{\mathrm{\bar{v}^{\prime}}}\Lambda_{2})
\sum_{{\scriptscriptstyle \overline{\Lambda}_{2}\overline{\Lambda}}}
J_{\mathrm{e\bar{v}^{\prime}\bar{v}c}}(\overline{\Lambda}\:\overline{\Lambda}_{2}\Lambda\Lambda_{2}\Lambda_{1})
\{{\scriptstyle \frac{1}{2}}\ddot{T}_{\mathrm{ee^{\prime}vc\bar{v}\bar{v}^{\prime}}}
(\widetilde{\Lambda}_{1}\overline{\Lambda}_{2}\overline{\Lambda})\\
-a(\lambda_{\mathrm{e^{\prime}}}\lambda_{\mathrm{v}}\Lambda_{1})({\scriptstyle \frac{1}{2}}
\widetilde{\ddot{T}}_{\mathrm{eve^{\prime}c\bar{v}\bar{v}^{\prime}}}
(\Lambda_{1}\overline{\Lambda}_{2}\overline{\Lambda})+\widetilde{\dot{T}}_{\mathrm{eve^{\prime}c\bar{v}\bar{v}^{\prime}}}
(\Lambda_{1}\overline{\Lambda}_{2}\overline{\Lambda}))\}]\\
+\delta_{\Lambda_{1}\Lambda_{2}}\delta_{\Lambda_{3}\tau_{1}}\delta_{\Lambda0}
\widetilde{T}_{\mathrm{ee^{\prime}v\bar{v}\bar{v}^{\prime}c}}(\Lambda_{1}\tau_{1})
+a(\lambda_{\mathrm{e}}\lambda_{\mathrm{\bar{v}}}\Lambda_{1})(-1)^{\Lambda_{2}}\\
\times Y_{\mathrm{e\bar{v}^{\prime}\bar{v}c}}(\Lambda_{1}\Lambda_{2}\Lambda\Lambda_{3}\tau_{1})
T_{\mathrm{ee^{\prime}vc\bar{v}\bar{v}^{\prime}}}(\Lambda_{1}\tau_{1})+\sum_{u}[{\scriptstyle \frac{1}{2}}
a(\lambda_{\mathrm{e^{\prime}}}\lambda_{\mathrm{\bar{v}}}u)\\
\times\{I_{\mathrm{ee^{\prime}v\bar{v}\bar{v}^{\prime}c}}(\Lambda_{1}\Lambda_{2}\Lambda_{3}\tau_{1}\Lambda u)
T_{\mathrm{vee^{\prime}c\bar{v}\bar{v}^{\prime}}}(u\tau_{1})\\
+I_{\mathrm{ee^{\prime}v\bar{v}\bar{v}^{\prime}c}}(\Lambda_{1}\Lambda_{2}\Lambda_{3}\tau_{2}\Lambda u)
T_{\mathrm{vee^{\prime}c\bar{v}\bar{v}^{\prime}}}^{\prime}(u\tau_{2})\}
+a(\lambda_{\mathrm{e^{\prime}}}\lambda_{\mathrm{v}}\Lambda_{1})\\
\times\{a(\lambda_{\mathrm{v}}\lambda_{\mathrm{\bar{v}}}u)
a(\lambda_{\mathrm{c}}\lambda_{\mathrm{\bar{v}^{\prime}}}\Lambda_{2})
I_{\mathrm{eve^{\prime}\bar{v}c\bar{v}^{\prime}}}(\Lambda_{1}\Lambda_{2}\Lambda_{3}\tau_{2}\Lambda u)\\
\times\widetilde{T^{\prime}}_{\mathrm{e^{\prime}ev\bar{v}^{\prime}\bar{v}c}}(u\tau_{2})
+{\scriptstyle \frac{1}{2}}(-1)^{\lambda_{\mathrm{e}}-\lambda_{\mathrm{\bar{v}}}}
I_{\mathrm{eve^{\prime}\bar{v}\bar{v}^{\prime}c}}(\Lambda_{1}\Lambda_{2}\Lambda_{3}\tau_{2}\Lambda u)\\
\times\widetilde{T^{\prime}}_{\mathrm{e^{\prime}vec\bar{v}\bar{v}^{\prime}}}(u\tau_{2})\}],
\end{array}\label{eq:c11a}\end{equation}

\begin{equation}
\begin{array}{l}
\Omega_{\mathrm{ee^{\prime}v\bar{v}\bar{v}^{\prime}c}}^{(2)-}(\Lambda_{1}\Lambda_{2}\Lambda_{3}\Lambda)
(\varepsilon_{\mathrm{\bar{v}\bar{v}^{\prime}c}}-\varepsilon_{\mathrm{ee^{\prime}v}})\\
\\=-\delta_{\Lambda\Lambda_{3}}\delta_{MM_{3}}[a(\lambda_{\mathrm{e^{\prime}}}\lambda_{\mathrm{v}}\Lambda_{1})
\{\widetilde{T}_{\mathrm{\bar{v}c\bar{v}^{\prime}ee^{\prime}v}}(\Lambda_{2}\Lambda_{1}\Lambda)\\
-{\scriptstyle \frac{1}{2}}a(\lambda_{\mathrm{c}}\lambda_{\mathrm{\bar{v}^{\prime}}}\Lambda_{2})
\widetilde{T}_{\mathrm{\bar{v}\bar{v}^{\prime}cee^{\prime}v}}(\Lambda_{2}\Lambda_{1}\Lambda)\}\\
+\sum_{{\scriptscriptstyle \overline{\Lambda}_{1}\overline{\Lambda}_{2}\overline{\Lambda}}}
a(\lambda_{\mathrm{e}}\lambda_{\mathrm{e^{\prime}}}\overline{\Lambda}_{1})
a(\lambda_{\mathrm{c}}\lambda_{\mathrm{\bar{v}}}\overline{\Lambda}_{2})\{(-1)^{\Lambda_{2}}\\
\times I_{\mathrm{ee^{\prime}v\bar{v}c\bar{v}^{\prime}}}^{\prime}(\overline{\Lambda}_{1}
\overline{\Lambda}_{2}\Lambda_{1}\Lambda_{2}\overline{\Lambda}\Lambda)
T_{\mathrm{c\bar{v}\bar{v}^{\prime}vee^{\prime}}}(\widetilde{\overline{\Lambda}}_{2}\overline{\Lambda}_{1}\overline{\Lambda})\\
+{\scriptstyle \frac{1}{4}}I_{\mathrm{ee^{\prime}v\bar{v}\bar{v}^{\prime}c}}^{\prime}
(\overline{\Lambda}_{1}\overline{\Lambda}_{2}\Lambda_{1}\Lambda_{2}\overline{\Lambda}\Lambda)
\widetilde{T}_{\mathrm{\bar{v}^{\prime}\bar{v}cvee^{\prime}}}(\overline{\Lambda}_{2}\overline{\Lambda}_{1}
\overline{\Lambda})\}].\end{array}\label{eq:c11b}\end{equation}

\end{subequations}

\begin{verse}
\emph{Four-body part.}
\end{verse}
\begin{equation}
\begin{array}{l}
\Omega_{\mathrm{vv^{\prime}v^{\prime\prime}v^{\prime\prime\prime}\bar{v}\bar{v}^{\prime}cc^{\prime}}}^{(2)}
(\Lambda_{1}\Lambda_{2}\Lambda_{3}\Lambda_{4}\Lambda)(\varepsilon_{\mathrm{\bar{v}\bar{v}^{\prime}cc^{\prime}}}
-\varepsilon_{\mathrm{vv^{\prime}v^{\prime\prime}v^{\prime\prime\prime}}})\\
\\={\scriptstyle \frac{1}{2}}(-1)^{\Lambda_{3}}a(\lambda_{\mathrm{c}}\lambda_{\mathrm{\bar{v}}}\Lambda_{1})
F_{\mathrm{c^{\prime}c\bar{v}\bar{v}^{\prime}}}(\Lambda_{1}\Lambda_{2}\Lambda_{3}\Lambda_{4}\Lambda)\\
\times Q_{\mathrm{vv^{\prime}v^{\prime\prime}v^{\prime\prime\prime}\bar{v}cc^{\prime}\bar{v}^{\prime}}}
(\widetilde{\Lambda}_{1}\Lambda_{3}).\end{array}\label{eq:c13}\end{equation}

\begin{equation}
\begin{array}{l}
\Omega_{\mathrm{evv^{\prime}v^{\prime\prime}\bar{v}\bar{v}^{\prime}cc^{\prime}}}^{(2)}
(\Lambda_{1}\Lambda_{2}\Lambda_{3}\Lambda_{4}\Lambda)(\varepsilon_{\mathrm{\bar{v}\bar{v}^{\prime}cc^{\prime}}}
-\varepsilon_{\mathrm{evv^{\prime}v^{\prime\prime}}})\\
\\={\scriptstyle \frac{1}{4}}\delta_{\Lambda_{1}\Lambda_{2}}\delta_{\Lambda_{3}\Lambda_{4}}
\delta_{\Lambda0}Q_{\mathrm{evv^{\prime}v^{\prime\prime}\bar{v}\bar{v}^{\prime}cc^{\prime}}}
(\widetilde{\Lambda}_{1}\Lambda_{3})+(-1)^{\Lambda_{3}}a(\lambda_{\mathrm{e}}\lambda_{\mathrm{v}}\Lambda_{1})\\
\times F_{\mathrm{c^{\prime}c\bar{v}\bar{v}^{\prime}}}(\Lambda_{1}\Lambda_{2}\Lambda_{3}\Lambda_{4}\Lambda)
Q_{\mathrm{vev^{\prime}v^{\prime\prime}c\bar{v}c^{\prime}\bar{v}^{\prime}}}(\widetilde{\Lambda}_{1}\Lambda_{3})
+{\scriptstyle \frac{1}{2}}(-1)^{\Lambda_{3}}\\
\times a(\lambda_{\mathrm{e}}\lambda_{\mathrm{v}}\Lambda_{4})
F_{\mathrm{evv^{\prime}v^{\prime\prime}}}(\Lambda_{4}\Lambda_{3}\Lambda_{2}\Lambda_{1}\Lambda)
Q_{\mathrm{vv^{\prime}ev^{\prime\prime}cc^{\prime}\bar{v}\bar{v}^{\prime}}}(\Lambda_{4}\Lambda_{2})\\
+{\scriptstyle \frac{1}{2}}\sum_{ud}G_{\mathrm{evv^{\prime}v^{\prime\prime}\bar{v}\bar{v}^{\prime}cc^{\prime}}}
(ud\Lambda_{1}\Lambda_{3}\Lambda_{2}\Lambda_{4}\Lambda)
\widetilde{Q}_{\mathrm{vv^{\prime}ev^{\prime\prime}\bar{v}cc^{\prime}\bar{v}^{\prime}}}(ud).
\end{array}\label{eq:c14}\end{equation}

\begin{equation}
\begin{array}{l}
\Omega_{\mathrm{ee^{\prime}vv^{\prime}\bar{v}\bar{v}^{\prime}cc^{\prime}}}^{(2)}
(\Lambda_{1}\Lambda_{2}\Lambda_{3}\Lambda_{4}\Lambda)(\varepsilon_{\mathrm{\bar{v}\bar{v}^{\prime}cc^{\prime}}}
-\varepsilon_{\mathrm{ee^{\prime}vv^{\prime}}})\\
\\={\scriptstyle \frac{1}{4}}\delta_{\Lambda_{1}\Lambda_{2}}\delta_{\Lambda_{3}\Lambda_{4}}\delta_{\Lambda0}
\{Q_{\mathrm{ee^{\prime}vv^{\prime}\bar{v}\bar{v}^{\prime}cc^{\prime}}}(\widetilde{\Lambda}_{1}\Lambda_{3})
+Q_{\mathrm{vv^{\prime}ee^{\prime}cc^{\prime}\bar{v}\bar{v}^{\prime}}}(\widetilde{\Lambda}_{3}\Lambda_{1})\}\\
-{\scriptstyle \frac{1}{2}}(-1)^{\Lambda_{3}}\{F_{\mathrm{c^{\prime}c\bar{v}\bar{v}^{\prime}}}
(\Lambda_{1}\Lambda_{2}\Lambda_{3}\Lambda_{4}\Lambda)Q_{\mathrm{ee^{\prime}vv^{\prime}c\bar{v}c^{\prime}\bar{v}^{\prime}}}
(\widetilde{\Lambda}_{1}\Lambda_{3})\\
+(-1)^{\Lambda_{1}+\Lambda_{4}}F_{\mathrm{e^{\prime}evv^{\prime}}}(\Lambda_{4}\Lambda_{3}\Lambda_{2}\Lambda_{1}\Lambda)
Q_{\mathrm{eve^{\prime}v^{\prime}cc^{\prime}\bar{v}\bar{v}^{\prime}}}(\widetilde{\Lambda}_{4}\Lambda_{2})\\
+a(\lambda_{\mathrm{\bar{v}}}\lambda_{\mathrm{\bar{v}^{\prime}}}\Lambda_{2})
a(\lambda_{\mathrm{c}}\lambda_{\mathrm{c^{\prime}}}\Lambda_{4})F_{\mathrm{cc^{\prime}\bar{v}^{\prime}\bar{v}}}
(\Lambda_{1}\Lambda_{2}\Lambda_{3}\Lambda_{4}\Lambda)\\
\times Q_{\mathrm{vv^{\prime}ee^{\prime}c\bar{v}c^{\prime}\bar{v}^{\prime}}}(\widetilde{\Lambda}_{3}\Lambda_{1})
+a(\lambda_{\mathrm{v}}\lambda_{\mathrm{v^{\prime}}}\Lambda_{3})a(\lambda_{\mathrm{e}}\lambda_{\mathrm{e^{\prime}}}\Lambda_{4})\\
\times F_{\mathrm{ee^{\prime}v^{\prime}v}}(\Lambda_{4}\Lambda_{3}\Lambda_{2}\Lambda_{1}\Lambda)
Q_{\mathrm{eve^{\prime}v^{\prime}\bar{v}\bar{v}^{\prime}cc^{\prime}}}(\widetilde{\Lambda}_{2}\Lambda_{4})\}
+a(\lambda_{\mathrm{e}}\lambda_{\mathrm{e^{\prime}}}\Lambda_{1})\\
\times\sum_{ud}a(\lambda_{\mathrm{c^{\prime}}}\lambda_{\mathrm{\bar{v}^{\prime}}}d)
\widetilde{Q}_{\mathrm{eve^{\prime}v^{\prime}\bar{v}c\bar{v}^{\prime}c^{\prime}}}(ud)\\
\times G_{\mathrm{e^{\prime}evv^{\prime}\bar{v}\bar{v}^{\prime}cc^{\prime}}}(ud\Lambda_{1}\Lambda_{3}\Lambda_{2}
\Lambda_{4}\Lambda).\end{array}\label{eq:c15}\end{equation}

\noindent{}The coefficients $J$, $Y$, $Y^{\prime}$, $I$, $I^{\prime}$,
$F$, $G$ are defined by the following formulas

\begin{equation}
J_{\alpha\beta\bar{\mu}\bar{\nu}}(\Lambda_{1}\Lambda_{2}\overline{\Lambda}_{1}\overline{\Lambda}_{2}\Lambda)
=(-1)^{\Lambda_{2}+\overline{\Lambda}_{2}}[\Lambda_{1}][\Lambda_{2},\overline{\Lambda}_{2}]^{1/2}
\left\{ \begin{smallmatrix}\lambda_{\alpha} & \lambda_{\bar{\nu}} & \Lambda_{1}\\
\lambda_{\bar{\mu}} & \lambda_{\beta} & \Lambda_{2}\\
\overline{\Lambda}_{1} & \overline{\Lambda}_{2} & \Lambda\end{smallmatrix}\right\} ,\label{eq:J}\end{equation}

\begin{align}
&Y_{\alpha\beta\bar{\mu}\bar{\nu}}(\Lambda_{1}\Lambda_{2}\overline{\Lambda}_{1}\overline{\Lambda}_{2}\Lambda)
=(-1)^{\Lambda_{1}+\Lambda_{2}+\overline{\Lambda}_{1}
+\overline{\Lambda}_{2}}[\overline{\Lambda}_{1},\Lambda_{2},\Lambda]^{1/2} \nonumber \\
&\times\langle\Lambda M\overline{\Lambda}_{1}\overline{M}_{1}\vert\overline{\Lambda}_{2}
\overline{M}_{2}\rangle\left\{ \begin{smallmatrix}\Lambda_{1} & \Lambda_{2} & \overline{\Lambda}_{1}\\
\lambda_{\bar{\nu}} & \lambda_{\bar{\mu}} & \lambda_{\beta}\end{smallmatrix}\right\} 
\left\{ \begin{smallmatrix}\overline{\Lambda}_{1} & \overline{\Lambda}_{2} & \Lambda\\
\lambda_{\alpha} & \lambda_{\bar{\nu}} & \lambda_{\bar{\mu}}\end{smallmatrix}\right\},\label{eq:Y}
\end{align}

\begin{align}
&Y_{\alpha\beta\bar{\mu}\bar{\nu}}^{\prime}(\Lambda_{1}\Lambda_{2}\overline{\Lambda}_{1}\overline{\Lambda}_{2}\Lambda)
=(-1)^{\Lambda_{1}+\Lambda_{2}}[\Lambda_{1},\overline{\Lambda}_{1},\Lambda_{2}]^{1/2} \nonumber \\
&\times\langle\Lambda_{1}M_{1}\overline{\Lambda}_{1}\overline{M}_{1}\vert\overline{\Lambda}_{2}
\overline{M}_{2}\rangle\left\{ \begin{smallmatrix}\overline{\Lambda}_{1} & \Lambda_{2} & \Lambda\\
\lambda_{\beta} & \lambda_{\alpha} & \lambda_{\bar{\mu}}\end{smallmatrix}\right\} 
\left\{ \begin{smallmatrix}\overline{\Lambda}_{1} & \overline{\Lambda}_{2} & \Lambda_{1}\\
\lambda_{\bar{\nu}} & \lambda_{\bar{\mu}} & \lambda_{\alpha}\end{smallmatrix}\right\} ,\label{eq:Yp}
\end{align}

\begin{align}
&I_{\alpha\beta\zeta\bar{\mu}\bar{\nu}\bar{\eta}}(\Lambda_{1}\Lambda_{2}
\overline{\Lambda}_{1}\overline{\Lambda}_{2}\Lambda\overline{\Lambda})
=(-1)^{\Lambda+\overline{\Lambda}_{1}}[\Lambda_{1},\Lambda_{2},\overline{\Lambda}_{2},\Lambda,\overline{\Lambda}]^{1/2} \nonumber \\
&\times\langle\overline{\Lambda}_{2}\overline{M}_{2}\Lambda M\vert\overline{\Lambda}_{1}
\overline{M}_{1}\rangle\left\{ \begin{smallmatrix}\lambda_{\alpha} & \lambda_{\beta} & \overline{\Lambda}\\
\lambda_{\bar{\nu}} & \lambda_{\bar{\mu}} & \overline{\Lambda}_{1}\end{smallmatrix}\right\} 
\left\{ \begin{smallmatrix}\lambda_{\beta} & \lambda_{\zeta} & \Lambda_{1}\\
\lambda_{\bar{\nu}} & \lambda_{\bar{\eta}} & \Lambda_{2}\\
\overline{\Lambda}_{1} & \overline{\Lambda}_{2} & \Lambda\end{smallmatrix}\right\} ,\label{eq:I}
\end{align}

\begin{align}
&I_{\alpha\beta\zeta\bar{\mu}\bar{\nu}\bar{\eta}}^{\prime}(\Lambda_{1}\Lambda_{2}\overline{\Lambda}_{1}
\overline{\Lambda}_{2}\Lambda\overline{\Lambda})=(-1)^{\lambda_{\alpha}+\lambda_{\zeta}+\lambda_{\bar{\mu}}
+\lambda_{\bar{\eta}}+\Lambda_{1}+\overline{\Lambda}_{1}+\Lambda_{2}+\Lambda+\overline{\Lambda}} \nonumber \\
&\times[\Lambda][\Lambda_{1},\Lambda_{2},\overline{\Lambda}_{1},\overline{\Lambda}_{2}]^{1/2}
\left[\begin{smallmatrix}\overline{\Lambda} & \lambda_{\beta} & \lambda_{\bar{\eta}} & \Lambda\\
\lambda_{\bar{\mu}} & \overline{\Lambda}_{2} & \Lambda_{1} & \lambda_{\zeta}\\
\lambda_{\alpha} & \Lambda_{2} & \overline{\Lambda}_{1} & \lambda_{\bar{\nu}}\end{smallmatrix}
\right],\label{eq:Ip}
\end{align}

\begin{align}
&F_{\alpha\beta\zeta\rho}(\Lambda_{1}\Lambda_{2}\overline{\Lambda}_{1}\overline{\Lambda}_{2}\Lambda)
=(-1)^{\overline{\Lambda}_{2}}[\Lambda_{2},\overline{\Lambda}_{2}]^{1/2}
\left\{ \begin{smallmatrix}\Lambda_{1} & \Lambda_{2} & \Lambda\\
\lambda_{\rho} & \lambda_{\beta} & \lambda_{\zeta}\end{smallmatrix}\right\}
\left\{ \begin{smallmatrix}\overline{\Lambda}_{1} & \overline{\Lambda}_{2} & \Lambda\\
\lambda_{\beta} & \lambda_{\rho} & \lambda_{\alpha}\end{smallmatrix}\right\} ,\label{eq:F}
\end{align}

\begin{align}
&G_{\alpha\beta\zeta\rho\bar{\mu}\bar{\nu}\bar{\eta}\bar{\sigma}}(\Lambda_{1}\Lambda_{2}
\overline{\Lambda}_{1}\overline{\Lambda}_{2}\overline{\overline{\Lambda}}_{1}
\overline{\overline{\Lambda}}_{2}\Lambda)=(-1)^{\Lambda}a(\lambda_{\alpha}
\lambda_{\beta}\overline{\Lambda}_{1})a(\lambda_{\bar{\nu}}\lambda_{\bar{\sigma}}\Lambda_{2}) \nonumber \\
&\times[\Lambda_{1},\Lambda_{2},\overline{\Lambda}_{1},\overline{\Lambda}_{2},
\overline{\overline{\Lambda}}_{1},\overline{\overline{\Lambda}}_{2}]^{1/2}
\left\{ \begin{smallmatrix}\lambda_{\beta} &  & \lambda_{\bar{\mu}}\:\:\overline{\Lambda}_{1} &  & 
\overline{\overline{\Lambda}}_{1}\:\:\lambda_{\alpha} &  & \lambda_{\bar{\nu}}\\
 & \Lambda_{1} &  & \Lambda &  & \Lambda_{2}\\
\lambda_{\zeta} &  & \lambda_{\bar{\eta}}\:\:\overline{\Lambda}_{2} &  & 
\overline{\overline{\Lambda}}_{2}\:\:\lambda_{\rho} &  & \lambda_{\bar{\sigma}}\end{smallmatrix}\right\}. 
\label{eq:G}
\end{align}

\begin{table}
\caption{\label{TabE}The expressions for phase factors $\mathcal{Z}_{\alpha^{\prime}\beta^{\prime}\bar{\mu}^{\prime}
\bar{\nu}^{\prime}}$}
\begin{tabular}{rrrrr}
\hline\hline
\raisebox{2.5ex}{}\raisebox{-1.ex}{}$\lambda_{\alpha^{\prime}}$ & 
$\lambda_{\beta^{\prime}}$ & $\lambda_{\bar{\mu}^{\prime}}$ & $\lambda_{\bar{\nu}^{\prime}}$ & 
$\mathcal{Z}_{\alpha^{\prime}\beta^{\prime}\bar{\mu}^{\prime}\bar{\nu}^{\prime}}(\Lambda_{1}\Lambda_{2}\Lambda)$ \\
\hline
\raisebox{2.5ex}{}\raisebox{-1.ex}{}$\lambda_{\alpha}$ & $\lambda_{\beta}$ & 
$\lambda_{\bar{\mu}}$ & $\lambda_{\bar{\nu}}$ & $1$ \\
$\lambda_{\alpha}$ & $\lambda_{\beta}$ & $\lambda_{\bar{\nu}}$ & $\lambda_{\bar{\mu}}$ & 
$a(\lambda_{\bar{\mu}}\lambda_{\bar{\nu}}\Lambda_{2})$ \\
$\lambda_{\beta}$ & $\lambda_{\alpha}$ & $\lambda_{\bar{\mu}}$ & $\lambda_{\bar{\nu}}$ & 
$a(\lambda_{\alpha}\lambda_{\beta}\Lambda_{1})$ \\
$\lambda_{\beta}$ & $\lambda_{\alpha}$ & $\lambda_{\bar{\nu}}$ & $\lambda_{\bar{\mu}}$ & 
$a(\lambda_{\alpha}\lambda_{\beta}\Lambda_{1})$ $a(\lambda_{\bar{\mu}}\lambda_{\bar{\nu}}\Lambda_{2})$ \\
$\lambda_{\bar{\mu}}$ & $\lambda_{\bar{\nu}}$ & $\lambda_{\alpha}$ & $\lambda_{\beta}$ & 
$a(\lambda_{\alpha}\lambda_{\beta}\Lambda_{2})$ $a(\lambda_{\bar{\mu}}\lambda_{\bar{\nu}}\Lambda_{1})$ \\
$\lambda_{\bar{\mu}}$ & $\lambda_{\bar{\nu}}$ & $\lambda_{\beta}$ & $\lambda_{\alpha}$ & 
$a(\lambda_{\bar{\mu}}\lambda_{\bar{\nu}}\Lambda_{1})$ \\
$\lambda_{\bar{\nu}}$ & $\lambda_{\bar{\mu}}$ & $\lambda_{\alpha}$ & $\lambda_{\beta}$ & 
$a(\lambda_{\alpha}\lambda_{\beta}\Lambda_{2})$ \\
$\lambda_{\bar{\nu}}$ & $\lambda_{\bar{\mu}}$ & $\lambda_{\beta}$ & $\lambda_{\alpha}$ & $1$ \\
\hline\hline
\end{tabular}
\end{table}

\noindent{}In Eq. (\ref{eq:Ip}), the last term with the brackets $\left[\ldots\right]$
denotes $12j$-symbol of the second kind (see Ref. \cite[Sec. 4-19, Eq. (19.3), p. 102]{Jucys}).
In Eqs. (\ref{5.2.5})-(\ref{c8}),
the expressions for the coefficients $\mathcal{Z}_{\alpha^{\prime}\beta^{\prime}\bar{\mu}^{\prime}\bar{\nu}^{\prime}}
(\Lambda_{1}\Lambda_{2}\Lambda)$,
where $\{\alpha^{\prime},\beta^{\prime},\bar{\mu}^{\prime},\bar{\nu}^{\prime}\}$
denote somehow permuted orbitals $\{\alpha,\beta,\bar{\mu},\bar{\nu}\}$
of the coefficient $\Omega_{\alpha\beta\bar{\mu}\bar{\nu}}^{(2)}(\Lambda_{1}\Lambda_{2}\Lambda)$,
are displayed in Tab. \ref{TabE}. Particularly, we use the abbreviation
$\mathcal{Z}_{\alpha^{\prime}\beta^{\prime}\bar{\mu}^{\prime}\bar{\nu}^{\prime}}(\Lambda_{1}\Lambda_{1}0)
\equiv\mathcal{Z}_{\alpha^{\prime}\beta^{\prime}\bar{\mu}^{\prime}\bar{\nu}^{\prime}}(\Lambda_{1})$.


\end{document}